\documentclass{article}
\usepackage{authblk}
\usepackage{physics}
\usepackage{amsmath,amssymb,amsthm}
\usepackage[]{qcircuit}
\usepackage{graphicx}
\usepackage{bm}
\usepackage{bbm}
\usepackage{caption}
\usepackage{subcaption}
\usepackage{tikz}
\usepackage{nicefrac}
\usepackage{placeins}
\usepackage{algorithm}
\usepackage[noend]{algpseudocode}
\usepackage{overpic}

\usepackage[colorlinks=true,citecolor=blue]{hyperref}
\usepackage{kbordermatrix}

\usepackage{breakurl}

\usepackage[capitalise,noabbrev]{cleveref}
\crefname{equation}{}{}
\crefname{enumi}{}{}

\usepackage[]{color}
\usepackage[natbib=true, style=authoryear-comp, maxbibnames=8, maxcitenames=1, uniquelist=false, citetracker=true, backend=biber,sortcites=false]{biblatex}				
\AtEveryCitekey{\ifciteseen{}{\defcounter{maxnames}{2}}} 
\bibliography{references.bib}

\usepackage{siunitx}
\usepackage{todonotes}

\providecommand{\keywords}[1]
{
  \small	
  \textbf{\textit{Keywords---}} #1
}

\makeatletter
\newrobustcmd*{\setmincitenames}{\numdef\blx@mincitenames}
\makeatother

\def\i{\mathrm{i}}
\def\e{\mathrm{e}}

\def\rev{\mathrm{rev}}

\def\calS{\mathtt{S}}
\def\calC{\mathtt{C}}
\def\calA{\mathtt{A}}
\def\calU{\mathcal{U}}

\def\Cliff{\mathcal{C}}
\def\calU{\mathcal{U}}

\def\sfx{\mathsf{x}}
\def\sfX{\mathsf{X}}

\def\Diag{\mathrm{Diag}}
\def\pf{\mathrm{pf}}

\def\deg{\mathrm{deg}}
\def\triu{\mathrm{triu}}
\def\skew{\mathrm{skew}}

\def\columns{\mathrm{columns}}

\def\parity{\mathrm{parity}}

\def\rmi{\mathrm{i}}

\def\K{\Pi_{\mathrm{im}(X)}}

\providecommand{\keywords}[1]{\textbf{\textit{Keywords ---}} #1}

\newtheorem{theorem}{Theorem}[section]
\newtheorem{lemma}[theorem]{Lemma}
\newtheorem{corollary}[theorem]{Corollary}
\newtheorem{proposition}[theorem]{Proposition}
\newtheorem{definition}[theorem]{Definition}
\newtheorem{example}[theorem]{Example}

\newtheorem{remark}[theorem]{Remark}

\DeclareMathOperator{\bin}{\mathrm{bin}}

\newcommand{\teal}[1]{\textcolor{teal}{#1}}

\title{\bf Bypassing orthogonalization in\\ the quantum DPP sampler}
\author[1]{Michaël Fanuel}
\author[1]{R\'emi Bardenet}
\date{}
\affil[1]{
\it Universit\'e de Lille, CNRS, Centrale Lille,\authorcr{}
\it UMR 9189 -- CRIStAL, F-59000 Lille, France
}

\begin{document}
\maketitle

\begin{abstract}
    Given an $n\times r$ matrix $X$ of rank $r$, consider the problem of sampling
$r$ integers $\mathtt{C}\subset \{1, \dots, n\}$ with probability proportional
to the squared determinant of the rows of $X$ indexed by $\mathtt{C}$. The
distribution of $\mathtt{C}$ is called a projection determinantal point process
(DPP). The vanilla classical algorithm to sample a DPP works in two steps, an
orthogonalization in $\mathcal{O}(nr^2)$ and a sampling step of the same cost.
The bottleneck of recent quantum approaches to DPP sampling remains that
preliminary orthogonalization step. For instance, (Kerenidis and Prakash, 2022)
proposed an algorithm with the same $\mathcal{O}(nr^2)$ orthogonalization,
followed by a $\mathcal{O}(nr)$ classical step to find the gates in a quantum
circuit. The classical $\mathcal{O}(nr^2)$ orthogonalization thus still
dominates the cost. Our first contribution is to reduce preprocessing to
normalizing the columns of $X$, obtaining $\mathsf{X}$ in $\mathcal{O}(nr)$
classical operations. We show that a simple circuit inspired by the formalism
of Kerenidis et al., 2022 samples a DPP of a type we had never encountered in
applications, which is different from our target DPP. Plugging this circuit
into a rejection sampling routine, we recover our target DPP after an expected
$1/\det \mathsf{X}^\top\mathsf{X} = 1/a$ preparations of the quantum circuit.
Using amplitude amplification, our second contribution is to boost the
acceptance probability from $a$ to $1-a$ at the price of a circuit depth of
$\mathcal{O}(r\log n/\sqrt{a})$ and $\mathcal{O}(\log n)$ extra qubits.
Prepending a fast, sketching-based classical approximation of $a$, we obtain a
pipeline to sample a projection DPP on a quantum computer, where the former
$\mathcal{O}(nr^2)$ preprocessing bottleneck has been replaced by the
$\mathcal{O}(nr)$ cost of normalizing the columns and the cost of our
approximation of $a$.

\end{abstract}
    
\keywords{
    Determinantal point processes, free fermions, quantum circuits. 
}
\newpage
\tableofcontents

\section{Introduction}
\label{sec:intro}
We say that a random subset $Y \subseteq \{1,\dots, n\}$ is drawn from a discrete determinantal point process (DPP) of correlation kernel $K \in \mathbb{C}^{n\times n}$ when
\begin{equation}
    \mathbb{P}(\calC \subseteq Y) = \det K_{\calC\calC}, \quad \forall \calC \subseteq \{1,\dots, n\}, \label{eq:inclusion_proba}
\end{equation}
where $K_{\calC\calC} = [K_{ij}]_{i,j\in \calC}$ is the principal submatrix of $K$ indexed by $\calC$.
When it exists, the distribution described by \eqref{eq:inclusion_proba} is denoted by $\mathrm{DPP}(K)$.
In particular, if $K$ is symmetric, $\mathrm{DPP}(K)$ is well-defined if and only if the eigenvalues of $K$ are in $[0,1]$ \citep{Mac72,Sos00}. 
Existence conditions for nonsymmetric correlation kernels take a more complicated form \citep{brunel2018learning}.

While originating in quantum statistical optics \citep{Mac72}, DPPs have been studied from several perspectives, such as machine learning \citep{kulesza2012determinantal,gartrell2021scalable}, randomized numerical algebra \citep{DeMa21}, probability \citep{HKPV09}, and statistics \citep{lavancier2015determinantal}.
In data science applications, a key asset of DPPs is that they can be sampled exactly in polynomial time.
In this paper, we consider the problem of sampling a subclass of DPPs, for which the correlation kernel $K = \K$ is the orthogonal projector onto the vector space $\mathrm{range}(X)$ spanned by the columns of a given matrix
$$
  X\in \mathbb{R}^{n \times r}, \quad r\leq n.
$$ 
For simplicity, we assume throughout the paper that the columns of $X$ are linearly independent.
In other words, the kernel $K$ is the projector 
\begin{equation}
  \K = X (X^\top X)^{-1} X^\top. \label{eq:target_kernel}
\end{equation}
Such DPPs are called \emph{projection DPPs}, and are common in applications to data science.
The flagship example of a projection DPP is the set of edges of a uniform spanning tree in a connected graph, where $X$ in \eqref{eq:target_kernel} is the vertex-edge incidence matrix of the graph, with any row removed to make it full-rank \citep{Pemantle91}. 
Sampling this DPP helps to build preconditioners for Laplacian linear systems \citep{KS2018,FanBar2022}.
Other data science applications include feature selection for linear regression \citep{BeBaCh20b}, graph filtering \citep{tremblay2017graph,jaquard2023smoothing}, or minibatch sampling in stochastic gradient descent \citep{BaGhLi21}. 
Another argument in favor of restricting one's attention to projection DPPs is that many DPP-related distributions involved in applications are statistical mixtures of projection DPPs, such as volume sampling \citep{DRVW06} or $k$-DPPs \citep{kulesza2012determinantal}.

The standard classical algorithm by \citet{HKPV09} to sample from $\mathrm{DPP}(K)$, as well as recent refinements \citep{BTA22}, require to orthogonalize the columns of $X$ as a preprocessing, followed by a chain rule argument.
The preprocessing typically takes the form of a $QR$-decomposition, i.e., 
\begin{equation}
  X = \mathsf{Q}R \tag{column orthogonalization}\label{eq:QR}
\end{equation}
where $\mathsf{Q}$ is an $n\times r$ real matrix, whose columns are orthonormal and span the same vector space as the columns of $X$, 
and $R$ is upper triangular; see e.g.\ \citep[Chapter II]{trefethen2022numerical}. 
The cost of this preprocessing is $\mathcal{O}(nr^2)$.
Similarly to the classical algorithm, a natural approach to sample $\mathrm{DPP}(K)$ on a quantum computer involves decomposing $\mathsf{Q}$ in \eqref{eq:QR} as a product of so-called \emph{Givens rotations}, which is actually one way to implement the QR decomposition \eqref{eq:QR}. 
These Givens rotations are then used to parametrize a quantum circuit of linear depth to sample the DPP of interest \citep*{bardenet2024sampling,kerenidis2022quantum}.
Because of this classical preprocessing, and neglecting potential gains in parallelizing it, the computational complexity of the whole quantum algorithm remains the same as the best available classical implementation of the algorithm of \cite{HKPV09}.
%
\subsection{Contributions}
To summarize, we give a quantum algorithm that takes as input an $n\times r$ full-rank matrix $X$, $r\leq n$, and outputs a sample of $\mathrm{DPP}(\K)$ in \eqref{eq:target_kernel}.
The quantum circuit involved has a similar depth and gate count as previous approaches, but our sampler is faster in the sense that the complexity of the classical preprocessing is taken down from $\mathcal{O}(nr^2)$ to $\mathcal{O}(nr + c(n,r))$, where $c(n,r)$ is the cost of approximating $a = \det (\sfX^\top \sfX)$,
\begin{equation}
  \sfX = X \cdot \Diag(\|X_{:1}\|_2^{-1}, \dots, \|X_{:r}\|_2^{-1}),\tag{column normalization}\label{eq:col_norm}
\end{equation}
and $X_{:1}, \dots, X_{:r}$ are the $r$ columns of $X$.
Using a QR decomposition, we see that $c(n,r) = \mathcal{O}(nr^2)$, but importantly, there are faster known classical algorithms than QR to approximate the determinant $a$.
In particular, we discuss below how $c(n,r)$ can be $\mathcal{O}(nr)$ in favourable cases, thus identifying a class of DPPs that are significantly faster to sample on a quantum computer.
As a side result of independent interest, we study a DPP of a new type, which naturally appears as an intermediate step in our construction. 
Because we hope to interest an interdisciplinary audience, we strove to make the paper accessible to readers with a background in probability and machine learning but comparatively little background in quantum computing.
We now go over our contributions in more detail.

\paragraph{Lighter preprocessing and rejection sampling algorithm.}
We reduce the DPP sampling cost by considering a \emph{computationally cheaper} classical preprocessing step than \citep{bardenet2024sampling,kerenidis2022quantum}.
We use the Clifford loaders of \citep[Section 4]{kerenidis2022quantum}, with Givens rotations, but in contrast to the QR-based approach, we decompose each column of $X$ individually. 
Explicitly, \cref{eq:QR} is traded off for a simple column-by-column normalization \eqref{eq:col_norm}.
The complexity of this preprocessing is $\mathcal{O}(nr)$.
Then, the columns of $\sfX$ are loaded (i.e., a quantum circuit is built) at cost $\mathcal{O}(nr)$, with the
Clifford loaders using Givens rotations.
Overall, preprocessing thus remains $\mathcal{O}(nr)$, shaving off a power $2$ on $r$ when compared with QR's cost. 
The resulting quantum circuit has dimensions of the same order as in previous works, so that the overall procedure is cheaper. 
Yet the price is that observing the obtained quantum state in the computational basis only yields the target DPP conditionally on the number of $1$s being the rank $r$ of $\K$.
Rejecting (``post-selecting") the output until this cardinality condition is filled yields a sampler for $\mathrm{DPP}(\K)$, where preprocessing cost has been traded in for an expected number $a=\det\sfX^\top \sfX$ of state preparations until acceptance.
Note that, as a byproduct, our analysis also describes the effect on the algorithm of \cite{kerenidis2022quantum} of numerical errors during the orthogonalization preprocessing, which are likely if $X$ is ill-conditioned.

\paragraph{Reducing rejection rate with amplitude amplification.}
From there, we show how to boost the acceptance probability of our rejection sampler using amplitude amplification, at the expense of adding $\mathcal{O}(\log n)$ control qubits and leveraging recent results of \citep{rethinasamy2024logarithmic} on coherent Hamming measurements; see \cref{a:rejection_sampling_Clifford_amplified}.
Prepending the randomized sketching-based estimation of $a$ proposed by \cite{boutsidis2017randomized}, we obtain, with a probability larger than $1-\delta$, an $\epsilon$-approximation $\widehat{a}$ at cost 
\begin{equation}
    \mathcal{O}\left(\text{nnz}(\sfX^\top \sfX)/\theta_1 \log (1/\epsilon)/\epsilon^{2}\log(1/\delta)\right),\label{eq:cost_sketch_theta_1}
\end{equation}
where $\theta_1$ is a user-chosen parameter such that  $0 < \theta_1 < \lambda_{\min}(\sfX^\top \sfX)$.
Then, using 
$$
  \lfloor \pi/(4\asin\sqrt{\widehat{a}}^{1/(1+2\epsilon)} )\rfloor
$$ 
Grover iterations, the acceptance probability of the sampling algorithm is amplified to at least $1 - \widehat{a} + o_{\epsilon}(1)$.
The exact lower bound on the amplified acceptance is given more formally in \cref{thm:guarantee_nb_Grover_epsilon}.

The complexity of the sketching step in \cref{eq:cost_sketch_theta_1} is part of the cost of the entire pipeline. 
Actually, the cost of sketching $a$ with the algorithm of \citep{boutsidis2017randomized} strongly depends on: (i) the asymptotic behaviour of $\lambda_{\min}(\sfX^\top \sfX)$ as a function of $n$ and $r$, and (ii) the user's ability to determine a lower bound $\theta_1$.
The former is key to determine when the cost \cref{eq:cost_sketch_theta_1} is low, and determining $\theta_1$ in general will be application-dependent.
To focus on the case of uniform spanning trees, consider an $n\times r$ matrix $X$ that is the edge-vertex incidence matrix of a connected graph \emph{where one column is removed}, so that $n$ is the number of edges and $r+1$ the number of nodes.
We consider in Section~\ref{sec:sketch_graph_case} two intuitively extreme cases; informally, a well-connected and a weakly connected graph.
\begin{itemize}
  \item For a connected graph where \emph{$\varrho$ is connected to all the other nodes} (a hub), we show that $\theta_1 = 1/(r+1)$ is a valid lower bound, and that the relevant function in the asymptotic upper bound \cref{eq:cost_sketch_theta_1} satisfies 
  $$
  \text{nnz}(\sfX^\top \sfX)/\theta_1 = \mathcal{O}\left(nr \right),
  $$ 
  yielding the same dependence in $nr$ as \cref{eq:col_norm}, and thus a DPP that can be sampled in $\mathcal{O}(nr)$.
  \item For a \emph{path graph} and for an endpoint node $\varrho$, we have that 
  $
  \lambda_{\min}(\sfX^\top \sfX)
  $
  behaves asymptotically as $\Theta(\sin^2(\frac{\pi}{4r + 2}))$.
  For instance, as shown in \cref{sec:sketch_graph_case}, we can choose $\theta_1 = \sin^2(\frac{\pi}{4r + 2})$, leading to
  $$
  \text{nnz}(\sfX^\top \sfX)/\theta_1 = \Theta\left(nr^2 \right),
  $$ 
  namely, the same scaling as the cost of QR.
\end{itemize} 
Although we do not have a general strategy to determine $\theta_1$ in the graph case, we expect the sketching cost to be low whenever the graph associated with $X$ is well-connected.

\paragraph{A DPP with non-symmetric correlation matrix.} 
As a side contribution, we identify in \cref{sec:novel_DPP} the distribution arising from a measurement in the computational basis without rejection nor amplitude amplification. 
It turns out that this is a DPP as well, of an unusual type since the probability of occurrence of a subset is in this case the determinant of a \emph{skewsymmetric} matrix.
To give intuition, we describe the particular case where the DPP selects edges in a graph, which generalizes a well-known determinantal measure over spanning trees in statistical physics and computer science.
In particular, we show that this new measure is supported on \emph{dimer-rooted forests}, and we highlight connections with Pfaffians.
Furthermore, since sampling this DPP relies on loading sparse columns of an edge-vertex incidence matrix, we also propose a \emph{sparse loader architecture} (see \cref{sec:sparse_architecture}), which intends to reduce the number of gates in the generic loaders described in \citep*{kerenidis2022quantum}.

\paragraph{Implementation and numerics.}
We implemented our algorithm using the Qiskit library \citep{Qiskit}. 
Results of simulations and executions on a real quantum computer are given in \cref{sec:numerics}, in the particular case of determinantal measures over spanning trees and dimer-rooted forests. 
Our code is freely available.\footnote{\url{https://github.com/mrfanuel/DPPs-with-clifford-loaders}}
\subsection{Notations}

For convenience, since we consider Pfaffians of submatrices, we have to define ordered subsets.
An ordered subset $\calS = (i_1,\dots, i_s)$ of $\{1, \dots,n\}$ is a sequence of distinct integers with a given order.
We denote the \emph{reversed} ordered subset by $\rev(i_1,\dots, i_s)= (i_s,\dots, i_1)$.
To simplify, we also abuse notation by writing $r + \{1,\dots,n\}$ for $\{r+1,\dots, r+n\}$.
We shall denote matrices with unit-norm columns by \emph{sans serif} fonts such as $\mathsf{X}$, $\mathsf{Q}$, $\mathsf{B}$ or $\sfx$.
When necessary, we denote the $n\times n$ identity matrix by $I_n$ whereas $0_n$ denotes the zero matrix of the same size.
\section{Sampling a projection DPP without orthogonalization \label{sec:without_orthogonalization}}
In this section, we introduce our quantum circuit. 
We find it easier to introduce it using operators satisfying canonical anticommutation relations (CAR), following a line of earlier work on fermionic simulation \citep{JSKSB18,TerDiVi02,bardenet2024sampling}. We thus start in \cref{sec:CAR_and_Clifford} by recalling the corresponding vocabulary, and translate the key notion of Clifford loader from \cite{kerenidis2022quantum} in these terms. 
%
\subsection{CAR algebra and Clifford loaders \label{sec:CAR_and_Clifford}}
%
Let $c_1, \dots, c_n$ be linear operators on a Hilbert space of finite dimension and denote by $c_1^*, \dots, c_n^*$ their adjoints. 
We assume that they satisfy the following algebraic relations, called \emph{Canonical Anticommutation Relations},\footnote{
  Mathematically, the relations \eqref{e:CAR} are used to generate an abstract $C^*$-algebra that models fermionic particles; see e.g.\ \citep[Section 5.2.2]{bratteli2012operator}, but we shall only need a particular representation in our paper.
}
\begin{align}
  \label{e:CAR}
  \tag{CAR}
  c_i c_j + c_j c_i = c_i^* c_j^* + c_j^* c_i^*  = 0
  \qquad \text{and} \qquad
  c_i c_j^* + c_j^* c_i= \delta_{ij} \mathbb{I},
\end{align}
for all $1\leq i,j\leq n$.

In this paper, we consider the Hilbert space $(\mathbb{C}^{2})^{\otimes n}$ of $n$ qubits.
The fermionic \cref{e:CAR} can be represented on this space with the Jordan-Wigner construction; see \citep{bardenet2022point} and references therein.
Recall the well-known \emph{Pauli matrices}, given by
\begin{equation*}
	I \triangleq I_{2} = \begin{bmatrix} 1 & 0\\ 0 & 1\end{bmatrix},\quad
	\sigma_x = \begin{bmatrix} 0 & 1\\ 1 & 0\end{bmatrix}, \quad
	\sigma_y = \begin{bmatrix} 0 & -\i\\ \i & 0\end{bmatrix}, \quad
	\sigma_z = \begin{bmatrix} 1 & 0\\ 0 & -1\end{bmatrix},
	\label{e:pauli_matrices}
\end{equation*}
and which satisfy $\sigma_x^2 = \sigma_y^2 = \sigma_z^2= I$, and $\sigma_x\sigma_y = \i\sigma_z$.
Now, the Jordan-Wigner representation of CAR is given by 
\begin{equation*}
    c_j = \underbrace{\sigma_z \otimes \dots \otimes \sigma_z}_{j-1\text{ times}} \otimes \left(\frac{\sigma_x-\i \sigma_y}{2}\right) \otimes \underbrace{I \otimes \dots \otimes I}_{n -(j-1) \text{ times}},
\end{equation*}
for all $1\leq i \leq n$ and by defining $c_j^*$ to be the adjoint of $c_j$.
Above, we slightly abuse notations by identifying an annihilation operator $c_j$ with its Jordan-Wigner representation.
It is customary to call the $c^*$'s \emph{creation operators} and the $c$'s \emph{annihilation operators}.
This becomes intuitive when considering Fock space.

Denote by $\ket{\emptyset}$ the Fock vacuum of $c_1, \dots, c_n$, namely the 
element of the Hilbert space such that 
$c_j \ket{\emptyset}= 0$ for all $1\leq j\leq n$.
Its representation is given by the tensor product
$
\ket{0 \dots 0} \triangleq \ket{0}\otimes \dots \otimes \ket{0}
$ with the factor  
$\ket{0}$ represented as the vector
$\left(\begin{smallmatrix}
  1 \\
  0
\end{smallmatrix}\right).
$
Similarly, it is conventional to represent $c_1^*\ket{\emptyset}$ by $\ket{10 \dots 0}$ with the factor $\ket{1}$ represented by $\left(\begin{smallmatrix}
  0 \\
  1
\end{smallmatrix}\right)$.
We often call the states $c_j^*\ket{\emptyset}$ for $1\leq j \leq n$ the one-particle states.
For $\sfx = [\sfx_1, \dots , \sfx_n]^\top\in \mathbb{R}^n$ such that $\|\sfx\|_2 = 1$, define the Clifford loader
\begin{equation}
    \Cliff(\sfx) = \sfx_1 \cdot (c_1 + c_1^*) + \dots + \sfx_n \cdot (c_n + c_n^*).\label{eq:Clifford_operator}
\end{equation}
For example, the Jordan-Wigner representation of $c_1 + c_1^*$ is simply $\sigma_x \otimes I \otimes \dots \otimes I$.
This operator is Hermitian by construction and a simple calculation shows that 
\begin{equation}
\Cliff(\sfx)\Cliff(\sfx^\prime) + \Cliff(\sfx^\prime)\Cliff(\sfx) = 2 (\sfx^\top \sfx^\prime) \cdot \mathbb{I} \text{ with } \|\sfx\|_2 = \|\sfx^\prime\|_2 = 1. \label{eq:Anticommutation_Cs}
\end{equation}
Still, if $ \|\sfx\|_2= 1$, the latter identity shows that $\Cliff(\sfx)^2 = \mathbb{I}$ and thus $\Cliff(\sfx)$ is unitary.
To say it in other words, a Clifford loader embeds a unit vector into a Clifford algebra -- here the Clifford algebra of dimension $n$ over the reals; see e.g. \citep[Section 2.3]{baez2002octonions} for a physics perspective.
%
%
To our knowledge, Clifford loaders were introduced in \citep*{kerenidis2022quantum} as a way to load the orthonormal columns of a matrix by successive applications of their loaders to the Fock vacuum $\ket{\emptyset}$. 
To summarize, \citet*{kerenidis2022quantum} propose the following architectures, named in this paper as follows,
\begin{itemize}
  \item pyramid loaders with $\mathcal{O}(n)$ gates and depth $\mathcal{O}(n)$,
  \item parallel loaders with $\mathcal{O}(n)$ gates and depth $\mathcal{O}(\log n)$.
\end{itemize}
We refer to \cref{sec:architectures_loaders} and \cref{app:Clifford} for more details about the construction of the loaders, and a connection to spherical coordinate systems.
Furthermore, \cref{sec:sparse_architecture} describes a sparse architecture that we propose in this paper.
\begin{remark}[Complexity of a Clifford loader]\label{rem:complexity_clifford}
  We can assume that $n = 2^k$ for some integer $k$, since zeroes can be appended to any vector so that the resulting length becomes a power of $2$.
  If there is no restriction on qubit connectivity, the loader $\Cliff(\sfx)$ can be constructed with a circuit of $\mathcal{O}(n)$ two-qubit gates and of depth $\mathcal{O}(\log n)$ (parallel architecture); see \citep*{kerenidis2022quantum} and \citep{johri2021nearest}, or \cref{app:Clifford}.
  If we add the requirement that the two-qubit gates act only on neighbouring qubits arranged as a path graph, then the circuit depth is $\mathcal{O}(n)$ (pyramid architecture).
  In contrast, the parallel loader of logarithmic depth requires two-qubit gates acting on distant qubits in the path graph.

\end{remark}

\subsection{State preparation with Clifford loaders\label{sec:state_preparation}}
Let $\sfX$ be an $n\times r$ matrix, $r\leq n$, with unit-norm columns $\sfX_{:1}, \dots, \sfX_{:r}$.
A direct substitution of the expression of the Clifford loader \cref{eq:Clifford_operator} gives
\begin{equation}
  \Cliff(\sfX_{:1}) \dots \Cliff(\sfX_{:r}) \ket{\emptyset}= \ket{\text{remainder}} + \sum_{i_1, \dots, i_r = 1}^n \sfX_{i_1 1 } \dots \sfX_{i_r r} \cdot c_{i_1}^* \dots c_{i_r}^* \ket{\emptyset},\label{eq:dpp_state}
\end{equation}
where $\ket{\text{remainder}}$ is orthogonal to each of the states of the form $c_{i_1}^* \dots c_{i_r}^* \ket{\emptyset}$. 
For conciseness, the explicit form of $\ket{\text{remainder}}$ is only given in \cref{eq:expression_colX} in the sequel.
Rearranging the terms in the sum on the right-hand side of \cref{eq:dpp_state} with the help of \eqref{e:CAR}, we find 
\begin{equation}
  \Cliff(\sfX_{:1}) \dots \Cliff(\sfX_{:r}) \ket{\emptyset}= \ket{\text{remainder}} + \sum_{\substack{\calS: |\calS| = r\\ \text{ ordered }}} \det(\sfX_{\calS:}) \prod_{i\in \rev(\calS)}c_{i}^*  \ket{\emptyset},\label{eq:remainplusdet}
\end{equation}
where the composition of creation operators on the right-hand side should appear in the order determined by $\calS$.
Note that $\ket{\text{remainder}}$ vanishes if the columns of $\sfX$ are mutually orthogonal as a consequence of \cref{eq:Anticommutation_Cs}.
For simplicity, denote by $\columns(\sfX)$ the ordered set of columns of $\sfX$.
Define
\begin{equation}
  \ket{\columns(\sfX)} \triangleq \Cliff(\sfX_{:1}) \dots \Cliff(\sfX_{:r}) \ket{\emptyset}.
  \label{eq:def_psi}
\end{equation}
The state \cref{eq:def_psi} generalizes the definition of subspace states \citep*{kerenidis2022quantum}. 
In particular, the latter correspond to  \cref{eq:def_psi} when all the columns of $\sfX$ are orthonormal, so that the remainder state in \eqref{eq:remainplusdet} is absent.
Note that, unlike subspace states, our construction \eqref{eq:def_psi} depends on the order of the columns because the remainder state does.
Now, from the complexity viewpoint, there is a circuit preparing the state \eqref{eq:def_psi} with $\mathcal{O}(nr)$ 2-qubit gates and a depth $\mathcal{O}(r \log n)$, in light of \cref{rem:complexity_clifford}.

%
\subsection{Observing occupation numbers\label{sec:defining_PP}}
As a consequence of \cref{e:CAR}, the fermion occupation number operators 
\begin{equation}
  N_1 = c_1^* c_1, \dots, N_n = c_n^* c_n,\label{eq:number_operators}
\end{equation}
are mutually commuting, and have as only eigenvalues $0$ and $1$ (respectively, no or one particle).
At this point, we list two trivial remarks which follow from \cref{e:CAR}: for all $1\leq i\leq n$
\begin{itemize}
  \item $N_i$ is the projector onto the eigenspace of $N_i$ of eigenvalue $1$,
  \item $\mathbb{I} - N_i$ is the projector onto the eigenspace of $N_i$ of eigenvalue $0$.
\end{itemize} 
For compactness, for all $1\leq i \leq n$, we define the indicator of $b\in\{0,1\}$ by the operator
\begin{align*}
  \mathbbm{1}_b(N_i) = b N_i + (1-b) (\mathbb{I} - N_i),
\end{align*}
which simply returns the \emph{projector onto the eigenspace of $N_i$ with eigenvalue $b$}.
Now, we are ready to introduce the random vector $[Y_1, \dots, Y_n]$, valued in $\{0,1\}^n$, defined by
\begin{equation}
  \mathbb{P}\left([Y_1 , \dots, Y_n] = [b_1, \dots, b_n]\right) = \bra{\columns(\sfX)} \mathbbm{1}_{b_1}(N_1) \dots \mathbbm{1}_{b_n}(N_n) \ket{\columns(\sfX)}.\label{eq:law_Ys}
\end{equation}
The above definition of random variables from a quantum state and commuting Hermitian operators is often called Born's rule.
To keep compact expressions, we do not indicate that this random vector depends on $\ket{\columns(\sfX)}$.
Abusing notation, we define the point process $Y$ on $\{1,\dots, n\}$ as follows. 
Let $\calC$ be a subset of $\{1,\dots,n\}$ and denote by $1(\calC)\in \mathbb{R}^n$ the \emph{indicator} vector with entries indexed by $\calC$ equal to $1$ and to $0$ otherwise.
Then, we define the law of $Y$ by
\begin{equation}
  \mathbb{P}(Y = \calC) \triangleq \mathbb{P}\left([Y_1 , \dots, Y_n] = 1(\calC)\right), \label{eq:def_Y}
\end{equation}
where the law of $[Y_1 , \dots, Y_n]$ is given by \eqref{eq:law_Ys}.
The point process $Y$ is in fact a DPP, which is described in further detail in \cref{sec:novel_DPP}.
For now, we rather focus on the process resulting from conditioning on $Y$ having cardinality $r$.
\begin{lemma}[Law of conditioned process]
  \label{l:conditional_probability_is_the_right_DPP}
  Let $Y$ be the point process defined in \cref{eq:def_Y}. We have 
  \begin{align*}
    \mathbb{P}\left(Y = \calC \big| |Y| = r\right) = \frac{1}{\det\sfX^\top \sfX}|\det\sfX_{\calC:}|^2.
  \end{align*}
\end{lemma}
In other words, by conditioning $Y$ on $|Y| = r$, we obtain the (projection) DPP with correlation kernel $\K$ described in \cref{eq:target_kernel}.
Note also that the conditioned process is invariant to permutation of the columns of $\sfX$.
\begin{proof}
First, by inspection of the coefficients at the right-hand side of \cref{eq:remainplusdet} and by using the law of $Y$ given by Born's rule \cref{eq:law_Ys}, for a subset $\calC$ such that $|\calC| = r$,  we have 
\begin{align}
  \mathbb{P}(Y = \calC) =  |\det\sfX_{\calC:}|^2, \label{eq:ProbY=C}
\end{align}
since the operator appearing in \cref{eq:law_Ys} with $[b_1,\dots,b_n] = 1(\calC)$ is the projector onto the subspace generated by $\prod_{i\in \calC}c_{i}^*  \ket{\emptyset}$.
Thus the probability that the result of the measure of $N_{1}+ \dots + N_{n}$ is $r$ is given by 
$$
\mathbb{P}\left(|Y| = r\right) = \sum_{\calS : |\calS| = r }\mathbb{P}\left(Y = \calS\right) = \sum_{\calS: |\calS| = r} |\det\sfX_{\calS:}|^2 = \det\sfX^\top\sfX,
$$
as a consequence of the Cauchy-Binet identity.
\end{proof}
\subsection{A rejection sampling algorithm}\label{sec:algo}
We give a rejection sampling algorithm for $\mathrm{DPP}(\K)$ as \cref{a:rejection_sampling_Clifford}.
\begin{algorithm}[h!]
  \begin{algorithmic}
    \Require $X\in \mathbb{R}^{n\times r}$ with $r< n $ and $\rank(X)=r$.
    \State Initialize $s \leftarrow 0$.
    \While{$s \neq r$} \hfill\emph{\teal{$\sharp$ Expected nb repetitions: $1/\det(\sfX^\top\sfX)$.}}
    \State  Build $\sfX$ with \cref{eq:col_norm}. \hfill\emph{\teal{$\sharp$ Cost: $\mathcal{O}(nr)$.}}
    \State Prepare $\ket{\columns(\sfX)} $. \hfill \emph{\teal{$\sharp$ Cost: $\mathcal{O}(nr)$.}}
    \State $[b_1, \dots, b_n]\leftarrow $ measure $[N_1, \dots, N_n]$ in $\ket{\columns(\sfX)}$.
    \State $s\leftarrow b_1 + \dots + b_n$.
    \EndWhile
    \Ensure $[b_1, \dots, b_n]\in \{0,1\}^n$.
  \end{algorithmic}
  \caption{Rejection sampling algorithm for $\mathrm{DPP}(\K)$.
  }
  \label{a:rejection_sampling_Clifford}
\end{algorithm}
Note that the state preparation is performed by using \cref{eq:def_psi} and \cref{app:Clifford}.
As a consequence of \cref{l:conditional_probability_is_the_right_DPP}, the acceptance probability of the output of this algorithm is $a = \det\sfX^\top\sfX$.
However, the latter determinant can take a small numerical value if the columns of $\sfX$ are close to be linearly dependent.
For instance, we give in \cref{ex:acceptance_barbell_graph} in the sequel an example where $X$ is the incidence matrix of a \emph{barbell} graph, for which $a$ decays exponentially with the number of nodes in the bar.
In such cases, the gain in preprocessing cost could be compensated by a long waiting time until acceptance in \cref{a:rejection_sampling_Clifford}.
In \cref{sec:amplitude_amplification}, we thus describe a procedure to increase the acceptance probability of \cref{a:rejection_sampling_Clifford} by modifying the quantum state, using the well-known amplitude amplification technique.
After that, \cref{sec:novel_DPP} further describes the point process obtained by omitting the rejection step, which has an interesting structure in the case where $X$ is an edge-vertex incidence matrix of a graph.





%
\section{Amplitude amplification to reduce the number of rejections \label{sec:amplitude_amplification}}
\cref{a:rejection_sampling_Clifford} so far repeatedly prepares $\ket{\columns(\sfX)}$ in 
\eqref{eq:def_psi} and observes all number operators until a sample of cardinality $r$ is obtained.
The number of repetitions until a sample is obtained is a geometric random variable of mean $\det \sfX^\top\sfX $, as an immediate corollary of \cref{l:conditional_probability_is_the_right_DPP}.
In this section, we study the possibility to reduce the number of rejections by manipulating $\ket{\columns(\sfX)}$ to boost the probability of obtaining a sample of cardinality $r$, without changing the conditional distribution. 
This procedure, known as amplitude amplification \citep*{Brassard-Hoyer}, is possible in our case thanks to a circuit, proposed recently  by \cite{rethinasamy2024logarithmic}, that performs a coherent Hamming measurement at the expense of adding extra ancillary qubits.

\subsection{Coherent Hamming measurements}
\label{s:coherent_hamming}

For all $s\in\{0, \dots, n\}$, denote by $P_s$ the orthogonal projector onto
\begin{equation}
  \mathrm{Span}\left ( \ket{i_1, \dots, i_n}: i_1, \dots, i_n \in \{0,1\} \text{ and } \sum_{\ell = 1}^n i_\ell = s \right ),\label{eq:span_r}
\end{equation}
namely, the eigenspace of the number operator $N_1 + \dots + N_n$ with eigenvalue $s$; see \cref{eq:number_operators} for their definition.
Thus, $s$ is here the \emph{Hamming weight} of the string $i_1\dots i_n$, and can take any value between $0$ and $n$. 
We use control qubits to store this Hamming weight.
Assuming here for simplicity $n + 1 = 2^k$, \cite{rethinasamy2024logarithmic} take in total $k+n$ qubits, and construct a unitary $V$ acting on $(\mathbb{C}^2)^{\otimes (k+n)}$ such that, for any unit-norm $\ket{\psi} \in (\mathbb{C}^2)^{\otimes n}$, 
\begin{equation}
  V \left( \ket{0}^{\otimes k}\otimes \ket{\psi} \right) = \sum_{s=0}^n \ket{\bin(s)} \otimes P_s \ket{\psi}.
  \label{e:augmented_state}
\end{equation}
Here, the first register of $k$ qubits is thought of as \emph{control} qubits.
For all $s \in \{0,\dots,n\}$,  we defined
$
\bin(s) = [s_0, \dots, s_{k-1}] \in \{0,1\}^k$ where $s = \sum_{\ell=0}^{k-1} s_\ell 2^{\ell}$, and thus 
\begin{equation}
  \ket{\bin(s)} = \left((\sigma_x)^{s_0}\otimes \dots \otimes (\sigma_x)^{s_{k-1}}\right)\ket{0}^{\otimes k} = \ket{s_0, \dots, s_{k-1}}.\label{eq:bin}
\end{equation}
In other words, the state $\ket{\bin(s)}$ consists in preparing $\ket{0}^{\otimes k}$ in this control register, and then setting the $\ell$th qubit to $1$ if and only if the $\ell$th digit of $s$ in base $2$ is $1$.
For the sake of completeness, we sketch here the construction by \cite{rethinasamy2024logarithmic} of the operator $V$ in \eqref{e:augmented_state}, and its decomposition in common single- and two-qubit gates.
For $s\in\{0, \dots, n\}$, consider first
\begin{equation}
  U_s = \sum_{r=0}^n \e^{\frac{2\i\pi r s}{n+1}}P_r.\label{eq:U_s}
\end{equation}  
Letting $R_Z(\phi) = \e^{-\i \phi/2 \times \sigma_z}$ be a so-called $Z$-rotation of angle $\phi$, one can check that\footnote{
  By, e.g., checking the matrix of $R_z(\phi)^{\otimes n}$ in the computational basis; see \citep[Section II]{rethinasamy2024logarithmic}.
}  
\begin{equation}
  U_s = \e^{\frac{\i\pi s n}{n+1}} (R_Z\left(\nicefrac{2\pi s}{n+1}\right))^{\otimes n},
  \label{e:def_U}
\end{equation}
so that $U_s$ only requires a circuit of depth $1$ of single-qubit gates. 
At this point, we introduce a notation for the following controlled-unitary gate
\begin{align*}
\mathrm{CU}_m =(I^{\otimes (m-1)}\otimes  \ketbra{0} \otimes I^{\otimes (k-m)}) \otimes I^{\otimes n}
+ (I^{\otimes (m-1)}\otimes  \ketbra{1}\otimes I^{\otimes (k-m)}) \otimes U_{2^{m-1}},
\end{align*}
which is controlled by the $m$th qubit of the control block and where the unitary transformation is given by \cref{eq:U_s} for $s=2^{m-1}$.
Again recall that $I=\left(\begin{smallmatrix}
  1 & 0\\
  0 & 1
\end{smallmatrix}\right)$.
Now let $V$ be defined by
\begin{equation}
  \label{e:def_V}
  V = (\mathrm{IQFT}\otimes I^{\otimes n}) \left(\prod_{m=1}^k \mathrm{CU}_m\right)  (H^{\otimes k} \otimes I^{\otimes n}),
\end{equation}
where $H = \nicefrac{1}{\sqrt{2}}\left(\begin{smallmatrix}
  1 & 1\\
  1 & -1
\end{smallmatrix}\right)$ denotes the Hadamard gate.
That $V$ can be built using one- and two-qubit gates is obvious from its form.
Indeed, $\mathrm{CU}_m$ is a controlled-$U_{2^{m-1}}$ gate, with a single control qubit.
By \eqref{e:def_U}, this can be written as a depth-$1$ product of controlled phase gates and a depth-$n$ product of controlled $Z$-rotations; see Remark 3 in \citep{rethinasamy2024logarithmic} for the exact implementation. 
The inverse quantum Fourier transform $\mathrm{IQFT}$ on the control register can be implemented using $\mathcal{O}(k^2)$ = $\mathcal{O}(\log^2 n)$ single- and two-qubit gates \citep[Chapter 5]{NiCh10}, and thus qualitatively corresponds to a negligible overhead.
Note that a constant number of gates can be further suppressed in the construction of $V$ by transferring rotations on the control qubits \citep{rethinasamy2024logarithmic}.
It remains to see that $V$ as defined in \eqref{e:def_V} satisfies \eqref{e:augmented_state}.
Indeed,
\begin{align*}
  & V \left( \ket{0}^{\otimes k} \otimes \ket{\psi} \right)\\
  &= \frac{1}{\sqrt{n+1}} (\mathrm{IQFT}\otimes I^{\otimes n})  \left(\prod_{m=1}^k \mathrm{CU}_m\right) \left( \sum_{y=0}^{n} \ket{\bin(y)} \otimes \ket{\psi} \right)\\
  &= \frac{1}{\sqrt{n+1}} \sum_{y=0}^{n} (\mathrm{IQFT}\otimes I^{\otimes n}) \left(\ket{\bin(y)} \otimes \prod_{m=1}^k (U_{2^{m-1}})^{y_m}\ket{\psi}\right),
\end{align*}
where $y_m\in \{0,1\}$ is the $m$-th digit (starting counting at zero) of $y$ in base $2$; see \cref{eq:bin}.
This gives
\begin{align*}
  V \left( \ket{0}^{\otimes k} \otimes \ket{\psi} \right)
  &= \frac{1}{\sqrt{n+1}} \sum_{y=0}^{n+1} (\mathrm{IQFT}\otimes I^{\otimes n}) \left(\ket{\bin(y)} \otimes U_y \ket{\psi}\right),\\
  &= \frac{1}{n+1} \sum_{y=0}^{n} \sum_{z=0}^{n} \sum_{r=0}^{n}  \e^{2\i\pi y \frac{s-z}{n+1}} \left(\ket{\bin(z)} \otimes P_s\ket{\psi}\right),
\end{align*}
where we substituted \cref{eq:U_s} to obtain the last equality.
But for $s,z\in\{0, \dots, n\}$, $\sum_{y=0}^{n+1} \e^{2\i\pi y \frac{s-z}{n+1}} = (n+1)\delta_{sz}$, so that
$$
  V \left( \ket{0}^{\otimes k} \otimes \ket{\psi} \right) = \sum_{s=0}^n  \ket{\bin(s)}\otimes P_s\ket{\psi},
$$
as claimed.
At the price of further sophistication, which we do not detail here, one can further manipulate the circuit to make the depth logarithmic in $n$, by adding a control register of $n$ qubits \citep[Section IV]{rethinasamy2024logarithmic}.

\subsection{Amplitude amplification}\label{sec:Amplitude_amplification}
Denote for convenience $\Cliff_{\sfX}= \Cliff(\sfX_{:1}) \dots \Cliff(\sfX_{:r})$, and for matching conventional notations of amplitude amplification, we write $\ket{\psi} \triangleq \Cliff_{\sfX}\ket{0}^{\otimes n} = \ket{\columns(\sfX)}$.
Define the following notations 
\begin{equation}
  \ket{\psi_1} = P_r \ket{\psi} \text{ and } \ket{\psi_0} =  (\mathbb{I} - P_r) \ket{\psi}, \label{eq:psi0_psi1}
\end{equation}
where $P_r$ is the orthogonal projector onto the space of states with Hamming weight $r$, namely the subspace defined in \cref{eq:span_r}.
Then 
\begin{equation}
\ket{\psi} = \ket{\psi_0} + \ket{\psi_1},\label{eq:psi_decomposition}
\end{equation}
and the probability, when measuring in the computational basis, to obtain a state with Hamming weight $r$ is 
\begin{equation}
  a \triangleq \braket{\psi_1}{\psi_1}. \label{eq:a}
\end{equation}
\citet{Brassard-Hoyer} first note that the following operator\footnote{Reusing the notations of \citep{Brassard-Hoyer}, Grover's operator is denoted by $Q$ -- hoping that it is clear from the context that it no related to a $QR$ decomposition. }
\begin{equation}
  Q(r,\sfX) \triangleq -\Cliff_{\sfX}S_0 \Cliff_{\sfX}^{*} S_r \label{eq:Q}
\end{equation}
leaves $\mathrm{Span}(\ket{\psi_0}, \ket{\psi_1})$ invariant, where 
\begin{equation}
  S_0 = \mathbb{I} - 2 P_0 \label{eq:S_0}
\end{equation}
flips the sign of $\ket{0}^{\otimes n}$ while leaving invariant its orthogonal complement.
Similarly,  
\begin{equation}
  S_r  = \mathbb{I} - 2 P_r \label{eq:S_r}
\end{equation}
flips the signs of components of Hamming weight $r$ and leaves invariant the orthogonal complement.
Note that $S_r$ reduces to $S_0$ when $r=0$.
To avoid cumbersome expressions, we now omit the dependence of $Q$ on $r$ and $\sfX$.
To better understand the definition of $Q$, we note that $Q$ can also be expressed as the composition of two reflections, i.e.,
\begin{equation}
    Q = - (\mathbb{I}  - 2P_{\ket{\psi}})(\mathbb{I} - 2 P_r),\label{eq:Q_as_reflections}
\end{equation}
where $P_{\ket{\psi}}$ denotes the orthogonal projector onto $\mathrm{Span}(\ket{\psi})$.
Hence, by inspection of \cref{eq:psi0_psi1} and \cref{eq:Q_as_reflections}, we see that -- when restricted to $\mathrm{Span}(\ket{\psi_0}, \ket{\psi_1})$ -- $Q$ acts as a rotation of angle $2\theta_a$, where $\theta_a \in [0, \pi/2]$ is defined by 
\begin{equation}
  \sin^2\theta_a = a.\label{eq:theta_a}
\end{equation}
This is the key ingredient of Grover's algorithm and amplitude amplification.
In particular, applying $m\geq 1$ times $Q$ to $\ket{\psi}$ yields 
\begin{equation}
  Q^m \ket{\psi} =  \sin ((2m+1)\theta_a) \frac{\ket{\psi_1}}{\sqrt{a}} +  \cos ((2m+1)\theta_a) \frac{\ket{\psi_0}}{\sqrt{1-a}}. \label{eq:Q**m_psi}
\end{equation}
  
At this point, recall the normalization \cref{eq:a}.
What has changed with respect to $\ket{\psi}$ in \cref{eq:psi_decomposition} is the probability to obtain a Hamming weight $r$, which is now $\sin^2 ((2m+1)\theta_a)$ rather than $a$. 
If $m = \lfloor \pi/4\theta_a\rfloor$, the amplified acceptance probability is guaranteed to satisfy 
\begin{align}
  \sin^2 ((2m+1)\theta_a) \geq \max(a, 1-a). \label{eq:amplified_acceptance_prob}
\end{align}
see \citep[Section 3]{boyer1998tight}.
Hence, if $m = \lfloor \pi/4\theta_a\rfloor$ and $a <1/2$, the probability that the resulting circuit's output is accepted is larger than $1-a$.
To achieve this amplification, the number of applications of $Q$ -- and thus $\Cliff_{\sfX}$ and its inverse in light of \cref{eq:Q} -- is
$$
m = \lfloor \pi/4\theta_a\rfloor = \Theta(1/\sqrt{a}),
$$
where we simply used \cref{eq:theta_a}.
Note that, if $a > 1/2$, $\lfloor \pi/4\theta_a\rfloor =0$.
\begin{algorithm}[h!]
  \begin{algorithmic}
      \Require  $X\in \mathbb{R}^{n\times r}$ with $r< n $ and $\rank(X)=r$, and $a = \det \sfX^\top \sfX$.
      \State Initialize $s \leftarrow 0$.
      \State  Build $\sfX$ with \cref{eq:col_norm}. \hfill\emph{\teal{$\sharp$ Cost: $\mathcal{O}(nr)$.}}            
      \If{$a<1/2$} 
      \State{$m \leftarrow \lfloor \frac{\pi}{4\asin\sqrt{a}}\rfloor$.} 
      \Else
      \State { $m \leftarrow 0$.}
      \EndIf
      \While{$s\neq r$}   \hfill\emph{\teal{$\sharp$  Expected nb repetitions: $ \leq 1/(1-a)$.}}
      \State Prepare $Q^m\ket{\columns(\sfX)}$.  \hfill\emph{\teal{$\sharp$ Cost: $\mathcal{O}(nr)$.}}
      \State $[b_1, \dots, b_n] \leftarrow$  measure $[N_1, \dots, N_n]$ in $Q^m\ket{\columns(\sfX)}$.
      \State $s \leftarrow b_1 + \dots + b_n$.
      \EndWhile
      \Ensure $[b_1, \dots, b_n]\in \{0,1\}^n$.
  \end{algorithmic}
  \caption{Amplified RS for $\mathrm{DPP}(\K)$ given $\det \sfX^\top \sfX$.}
  \label{a:rejection_sampling_Clifford_amplified}
\end{algorithm}

\begin{remark}[Rule of thumb for the quadratic speedup]\label{rem:rule_of_thumb}
Consider the case of a small acceptance probability -- say $a < 0.01$. 
By quadratically approximating the amplified acceptance probability $\sin^2 ((2m+1)\theta_a)$, we see that $m=1$ already multiplies $a$ by a factor close to $9$, while $m=2$ corresponds to a factor close to $25$.
\end{remark}
Because of the analogy with Grover's algorithm, we call the parameter $m$ the number of Grover iterations.
The optimal choice of $m$ depends on the knowledge of $a$, which is an issue in practice. 
\cite{Brassard-Hoyer} propose a modification of their algorithm to actually estimate $a$ before running amplitude amplification. 
Another possibility is to estimate $a$ by classical means, as long as the computational complexity of this estimation step remains comparable to the overall cost of the sampling procedure.
In the context of this paper, the acceptance probability -- defined in \cref{eq:a} -- reads 
$$
a = \det\sfX^\top \sfX,
$$
and can be computed in $\mathcal{O}(nr^2)$ operations, since $X$ is $n\times r$.
Computing $a$ would thus defeat our initial purpose of lowering the cost of the QR preprocessing; see \cref{sec:intro}.
Thus, we have to abandon exactness somewhere.

Fortunately, this determinant can also be approximated with the help of the sketching technique for log-determinants in \citep[Eq.\ 5]{boutsidis2017randomized}, which we now rephrase in our context where the eigenvalues of $\sfX^\top \sfX$ are upper bounded by $1$.
Let $\epsilon>0$ and $\delta\in (0,1)$. 
There is a randomized algorithm, known as Gaussian sketching, which takes $\sfX$ as input and outputs $\widehat{a}$ such that 
\begin{equation}
  \mathbb{P}(a^{2\epsilon} \leq \widehat{a}/a \leq 1/a^{2\epsilon})\geq 1-\delta,\label{eq:rel_guarantee_on_a}
\end{equation}
in time 
\begin{equation}
  \mathcal{O}\left(\frac{\text{nnz}(\sfX^\top \sfX)}{\theta_1} \frac{\log (1/\epsilon)}{\epsilon^{2}} \log(1/\delta)\right) \quad \text{ with } 0 < \theta_1 < \lambda_{\min}(\sfX^\top \sfX), \label{eq:sketching_cost}
\end{equation}
where $\text{nnz}(\sfX^\top \sfX)$ and $\lambda_{\min}(\sfX^\top \sfX)$ are respectively the number of non-zero entries and the smallest eigenvalue of $\sfX^\top \sfX$.
The RS algorithm using this sketch is \cref{a:rejection_sampling_Clifford_amplified_approx} and the guarantee for boosting acceptance is given in \cref{thm:guarantee_nb_Grover_epsilon}.
\begin{algorithm}[h!]
  \begin{algorithmic}
      \Require  $X\in \mathbb{R}^{n\times r}$ with $r< n $ and $\rank(X)=r$, $\epsilon \in (0, 1/2)$ and $\theta_1$ as in \cref{eq:sketching_cost}.
      \State Initialize $s \leftarrow 0$.
      \State  Build $\sfX$ with \cref{eq:col_norm}. \hfill\emph{\teal{$\sharp$ Cost: $\mathcal{O}(nr)$.}}
      \State  Compute an $\epsilon$-approximation $\widehat{a}$.\hfill\emph{\teal{$\sharp$ Cost: $\mathcal{O}\left(\frac{\text{nnz}(\sfX^\top \sfX)}{\theta_1} \frac{\log (1/\epsilon)}{\epsilon^{2}}\right)$ whp.}}  
      \If{$\widehat{a}^{\frac{1}{1+2\epsilon}}<1/2$} 
      \State{$m \leftarrow \left\lfloor \frac{\pi}{4\asin\sqrt{\widehat{a}}^{\frac{1}{1+2\epsilon}}} \right\rfloor$.} 
      \Else
      \State { $m \leftarrow 0$.}
      \EndIf
      \While{$s\neq r$}   \hfill\emph{\teal{$\sharp$  Expected nb repetitions: $\leq 1/( 1- \widehat{a} + o_\epsilon(1))$.}}
      \State Prepare $Q^m\ket{\columns(\sfX)}$.  \hfill\emph{\teal{$\sharp$ Cost: $\mathcal{O}(nr)$.}}
      \State $[b_1, \dots, b_n] \leftarrow$  measure $[N_1, \dots, N_n]$ in $Q^m\ket{\columns(\sfX)}$. \emph{\teal{$\sharp$ i.e.\ sample \cref{eq:law_Ys}.}} 
      \State $s \leftarrow b_1 + \dots + b_n$.
      \EndWhile
      \Ensure $[b_1, \dots, b_n]\in \{0,1\}^n$.
  \end{algorithmic}
  \caption{Amplified RS for $\mathrm{DPP}(\K)$ by sketching $\det \sfX^\top \sfX$.}
  \label{a:rejection_sampling_Clifford_amplified_approx}
\end{algorithm}
\begin{theorem}[Effect of an $\epsilon$-approximation on number of Grover steps]
  \label{thm:guarantee_nb_Grover_epsilon}
  Let $\epsilon \in(0 , 1/2)$ and let $\widehat{a}$ be an $\epsilon$-approximation of $
  a = \det\sfX^\top \sfX$ in the sense of \cref{eq:rel_guarantee_on_a}.
  Assume $\widehat{a}^{\frac{1}{1+2\epsilon}} <1/2$.
  If we take 
  \begin{equation}
    m(\widehat{a},\epsilon) = \lfloor \pi/4\asin\sqrt{\widehat{a}}^{\frac{1}{1+2\epsilon}} \rfloor \label{eq:m_epsilon}
  \end{equation} 
  Grover iterations, the acceptance probability of rejection sampling in \cref{a:rejection_sampling_Clifford_amplified_approx} is larger than
  \begin{equation}
    q(\widehat{a},\epsilon) \triangleq\sin^2\left(\left( 2 m(\widehat{a},\epsilon) + 1\right)\asin\sqrt{\widehat{a}}^{\frac{1}{1-2\epsilon}}\right).\label{eq:lower_bound_epsilon}
  \end{equation}
  Also, we have $\lim_{\epsilon\to 0_+}q(\widehat{a},\epsilon) \geq 1 - \widehat{a}$.
\end{theorem}
Before proving \cref{thm:guarantee_nb_Grover_epsilon}, note that the lower bound on the acceptance probability \cref{eq:lower_bound_epsilon} as well as the number of Grover steps \cref{eq:m_epsilon} are easily computed by a user of \cref{a:rejection_sampling_Clifford_amplified_approx} given an $\epsilon$-approximation of the acceptance probability.
The increase of this lower bound as $\epsilon$ decreases is manifest in \cref{fig:lower_bound_epsilon} for three values of $\widehat{a}$.
Indeed, the bound goes rapidly above $1/2$ as $\epsilon$ decreases below $0.05$, whereas the number of Grover iterates in \cref{fig:m_epsilon} remains rather small.
  \begin{figure}[h!]
    \centering
    \begin{subfigure}{0.45\textwidth}
          \centering
          \includegraphics[scale = 0.4]{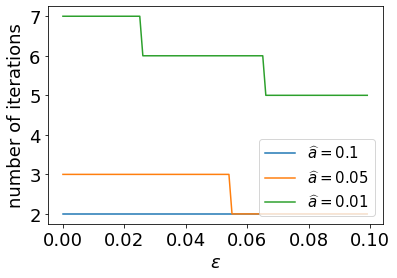}
          \caption{$m(\widehat{a},\epsilon)$.\label{fig:m_epsilon}}
    \end{subfigure}
    \hfill
    \begin{subfigure}{0.45\textwidth}
          \centering
          \includegraphics[scale = 0.4]{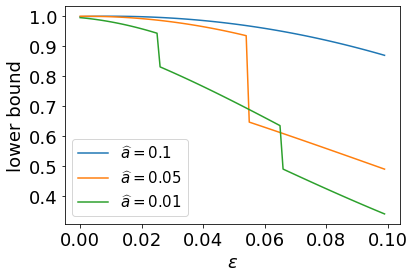}
          \caption{$q(\widehat{a},\epsilon)$.
        \label{fig:lower_bound_epsilon}}
    \end{subfigure}
  \caption{Effect of $\epsilon$-approximation of $a$ for amplitude amplification. 
  On the left-hand side, we report the number of Grover iterations \cref{eq:m_epsilon} for boosting the acceptance probability as in \cref{thm:guarantee_nb_Grover_epsilon}.
  On the right-hand side, we display the lower bound \cref{eq:lower_bound_epsilon} on the acceptance probability of \cref{a:rejection_sampling_Clifford_amplified_approx}.}
\end{figure}
\begin{proof}
  First, note that \cref{eq:rel_guarantee_on_a} guarantees that $\widehat{a}\leq 1$ if $0 < \epsilon < 1/2$. 
  Still, using \cref{eq:rel_guarantee_on_a}, we have
  \begin{equation}
      0 < \widehat{a}^{\frac{1}{1-2\epsilon}} \leq a \leq \widehat{a}^{\frac{1}{1+2\epsilon}} \leq 1\quad \text{ with } 0 < \epsilon < 1/2. \label{eq:delta_interval_a_hat}
  \end{equation}
  By using \cref{eq:theta_a} and the fact that $\asin(\cdot)$ is an increasing function on $[0,1]$, we obtain
  $$
  0 < \underbrace{\asin\widehat{a}^{\frac{1}{2-4\epsilon}}}_{\theta_{a,-} }\leq \theta_a \leq\underbrace{\asin\widehat{a}^{\frac{1}{2+4\epsilon}}}_{\theta_{a,+}}  \leq \pi/2.
  $$
  If $(m+1/2)\theta_{a,+} \leq \pi/4$, we have the following guarantee on the acceptance probability
  $$
  0 < \sin^2((2m+1)\theta_{a,-})  \leq \sin^2((2m+1)\theta_{a}) \leq \sin^2((2m+1)\theta_{a,+})   \leq 1,
  $$
  since $\sin^2(\cdot)$ is increasing on $[0,\pi/2]$.
  Thus, if we take $m = m(\widehat{a},\epsilon) =\lfloor \pi/4\theta_{a,+}\rfloor$, we can guarantee that the success probability is larger than $\sin^2((2m+1)\theta_{a,-})$.
\end{proof}
To complete the picture, we need to have a circuit for the so-called ``Grover operator" $Q$. 

\FloatBarrier
\subsection{A circuit for the Grover operator \label{sec:circuit_Grover}}
We have studied in \cref{sec:CAR_and_Clifford} a circuit for $\Cliff_{\sfX} = \Cliff(\sfX_{:1}) \dots \Cliff(\sfX_{:r})$, from which it is easy to deduce a circuit for its inverse $\Cliff_{\sfX}^* = \Cliff(\sfX_{:r}) \dots \Cliff(\sfX_{:1})$ by using \cref{eq:Anticommutation_Cs}. 
The only new components of the circuit are $S_0$ and $S_r$, given respectively in \cref{eq:S_0} and \cref{eq:S_r}.
While $S_0$ is standard (see \cref{prop:S_0}), we explain in this section how we use the construction of \cref{s:coherent_hamming} to build a circuit for $S_r$.


We start by defining a reflection.
\begin{lemma}\label{prop:Grover_oracle}
  Let $1\leq r < n$ and denote by $\bin(r) = [r_0, \dots , r_{k-1}]$ the vector of its decomposition in base $2$.
  Define the unitary 
  $$
  O(r) = \left( \sigma_x^{1-r_0} \otimes \dots \otimes  \sigma_x^{1-r_{k-1}}\right)\mathrm{CZ}_k \left( \sigma_x^{1-r_0} \otimes \dots \otimes \sigma_x^{1-r_{k-1}} 
  \right),
  $$
  where $\mathrm{CZ}_{k}$ is a multiple controlled-$\sigma_z$ gate.
  For all $s\in\{0,\dots, n\}$, we have
  $$
  O(r)  \ket{\bin(s)} =  (-1)^{\delta_{rs}} \ket{\bin(s)}.
  $$
\end{lemma}
\begin{proof}
  First, by construction $\sigma_x^{1-r_0} \otimes \dots \otimes  \sigma_x^{1-r_{k-1}} \ket{\bin(r)} = \ket{1}^{\otimes k}$.
  Second, by definition, $\mathrm{CZ}_{k} \ket{1}^{\otimes k} = -  \ket{1}^{\otimes k}$ and $\mathrm{CZ}_{k}$ leaves invariant the orthogonal of $\ket{1}^{\otimes k}$. 
  Thus,
  $
  O(r)\ket{\bin(r)} = - ( \sigma_x^{1-r_0} \otimes \dots \otimes  \sigma_x^{1-r_{k-1}} )   \ket{1}^{\otimes k} = -\ket{\bin(r)}.
  $
  Furthermore, for $s\neq r$, 
  $
  O(r)\ket{\bin(s)} =  ( \sigma_x^{1-r_0} \otimes \dots \otimes  \sigma_x^{1-r_{k-1}} )   \ket{1}^{\otimes k} = \ket{\bin(r)}.
  $
  This is the desired result.
\end{proof}
\cref{eq:rep_S_r} leverages \cref{prop:Grover_oracle} and the construction of \citet{rethinasamy2024logarithmic} to yield a circuit for $S_r$.
\begin{lemma}[Oracle for Grover's algorithm]\label{eq:rep_S_r}
  Let the unitary operator $V$ be as defined in \eqref{e:def_V}.
  Let also $O(r)$ be as defined in \cref{prop:Grover_oracle}
  and consider a unit-norm $\ket{\psi}\in (\mathbb{C})^{\otimes n}$. We have
  $$
  V^* \left(O(r)\otimes I^{\otimes n}\right) V \left(\ket{0}^{\otimes k} \otimes \ket{\psi}\right) = \ket{0}^{\otimes k} \otimes S_r\ket{\psi}.
  $$
\end{lemma}

\begin{proof}
  By construction, for all unit-norm $\ket{\psi}\in (\mathbb{C})^{\otimes n}$, we have \cref{e:augmented_state}, namely
  $
  V( \ket{0}^{\otimes k} \otimes \ket{\psi} ) = \sum_{s=0}^n  \ket{\bin(s)}\otimes  P_s\ket{\psi}.
  $
  By acting now with $O(r)$ on the first $k$ qubits and by using next the definition of $S_r$ in \eqref{eq:S_r}, we find 
  \begin{align}
    \left(O(r)\otimes I^{\otimes n}\right)V \left( \ket{0}^{\otimes k} \otimes \ket{\psi} \right) &= \sum_{s=0}^n  \ket{\bin(s)}\otimes  (-1)^{\delta_{rs}}P_s\ket{\psi} \nonumber \\
    &=\sum_{s=0}^n  \ket{\bin(s)}\otimes P_s S_r\ket{\psi}.\label{eq:effect_O(r)}
  \end{align}
  At this point, we leverage again \cref{e:augmented_state} but with $\ket{\psi}$ replaced by $S_r\ket{\psi}$ this time, i.e.,
  $$
  \ket{0}^{\otimes k} \otimes S_r\ket{\psi}  = V^* \sum_{s=0}^n  \ket{\bin(s)}\otimes  P_s S_r \ket{\psi}.
  $$
  Combining the latter identity with \cref{eq:effect_O(r)} yields the desired result.
\end{proof}
Next, \cref{prop:S_0} provides a circuit for implementing $S_0$.
\begin{lemma}\label{prop:S_0}
  Let $\mathrm{CZ}_{n}$ denote the multiple controlled-$\sigma_z$ gate of size $n$.
  Let 
  $$
  S_0 = \left(\sigma_x \otimes \dots \otimes \sigma_x\right)\mathrm{CZ}_{n}\left(\sigma_x \otimes \dots \otimes \sigma_x\right).
  $$
  We have 
  $
  S_0 \ket{0}^{\otimes n} = - \ket{0}^{\otimes n},
  $
  and, for all $\ket{\phi}$ orthogonal to $\ket{0}^{\otimes n}$, we have
  $
  S_0 \ket{\phi} = \ket{\phi}.
  $
\end{lemma}
\begin{proof}
  Again, by definition, $\mathrm{CZ}_{n} \ket{1}^{\otimes n} = -  \ket{1}^{\otimes n}$ and $\mathrm{CZ}_{n}$ leaves invariant the orthogonal of $\ket{1}^{\otimes n}$. 
\end{proof}
By composing the building blocks in \cref{eq:rep_S_r} and \cref{prop:S_0}, we have a circuit for Grover's operator.
\begin{theorem}[A circuit for Grover's operator]
  In the setting of \cref{sec:Amplitude_amplification}, let $Q$ be given in \cref{eq:Q}.
  Define 
  $$ 
    C = -(I^{\otimes k} \otimes \Cliff_{\sfX})(I^{\otimes k} \otimes S_0)(I^{\otimes k} \otimes \Cliff_{\sfX}^*) V^* \left( O(r)\otimes I^{\otimes n}\right)V
  $$ 
  where $S_0$ is given in \cref{prop:S_0} and $O(r)$ is defined in \cref{prop:Grover_oracle}.
  We have 
  $$
    C (\ket{0}^{\otimes k} \otimes \ket{\psi}) = \ket{0}^{\otimes k} \otimes Q\ket{\psi}.
  $$
\end{theorem}
\begin{proof}
  Using \cref{eq:rep_S_r}, we have
  \begin{align*}
    -C (\ket{0}^{\otimes k} \otimes \ket{\psi}) &= (I^{\otimes k} \otimes \Cliff_{\sfX})(I^{\otimes k} \otimes S_0)(I^{\otimes k} \otimes \Cliff_{\sfX}^*)\left(\ket{0}^{\otimes k} \otimes S_r\ket{\psi}\right)\\
    &=  \ket{0}^{\otimes k} \otimes \Cliff_{\sfX}S_0 \Cliff_{\sfX}^* S_r\ket{\psi},
  \end{align*}
  which concludes the proof.
\end{proof}
This concludes the presentation of our fast quantum sampler for projection DPPs. 
Before turning to experimental validation, we now examine the point process that results from measuring the occupation number without applying any rejection step.

\section{A  DPP with nonsymmetric correlation kernel  \label{sec:novel_DPP}}
To better describe the point process associated with \cref{eq:law_Ys}, we first need to introduce useful notations. 
Denote the upper triangular part of a square matrix $A$ by $\triu(A)$.
Let $\sfX \in \mathbb{R}^{n\times r}$ be a matrix with unit-norm columns.
Now, define the skewsymmetric matrix $\skew(A) = \triu(A) - \triu(A)^\top$.
For instance, by taking a matrix $\sfX = [\sfX_{:1} \ \sfX_{:2} \ \sfX_{3:}]$ with columns of unit norm, it comes 
$$
  \skew \left(\sfX^\top \sfX \right)
  =
  \begin{bmatrix}
    0 & \sfX_{:1}^\top \sfX_{:2} & \sfX_{:1}^\top \sfX_{:3}\\
    -\sfX_{:1}^\top \sfX_{:2} & 0 & \sfX_{:2}^\top \sfX_{:3}\\    
    -\sfX_{:1}^\top \sfX_{:3} & - \sfX_{:2}^\top \sfX_{:3} &  0
  \end{bmatrix}.
$$
\subsection{Definition and properties}
Consider again the state $\ket{\columns(\sfX)}=\Cliff(\sfX_{:1}) \dots \Cliff(\sfX_{:r}) \ket{\emptyset}$ as given in \cref{eq:dpp_state}.
Let $Y$ be the point process associated with the measure of the occupation numbers
\[
  [N_1, \dots, N_n] \text{ with } N_i \triangleq c_i^* c_i \text{ for } i\in\{1, \dots n\},
\] 
in the state $\ket{\columns(\sfX)}$; as given by \cref{eq:law_Ys}.
Hence, $[N_1, \dots, N_n]$ is associated with a random vector in $\{0,1\}^n$.
%
\begin{theorem}[DPP associated with a Clifford loader]\label{thm:pmf_DPP}
  Let $\sfX \in \mathbb{R}^{n\times r}$ be a matrix with unit-norm columns such that $\rank(\sfX) = r$.
  The point process $Y$ defined in \cref{eq:law_Ys} and associated with the measure of the occupation numbers $N_1, \dots, N_n$ in the state $\ket{\columns(\sfX)}$ given in \eqref{eq:def_psi} has the following law
  \begin{equation*}
    \mathbb{P}(Y = \calC) = \det\begin{bmatrix}
      0_{|\calC|}& \sfX_{\calC:}\\
      -(\sfX_{:\calC})^\top & \skew(\sfX^\top\sfX)
    \end{bmatrix},
  \end{equation*}
  where $\calC\subseteq \{1, \dots ,n\}$. In particular, we have
  $
    \mathbb{P}(\parity(|Y|) \neq  \parity(r)) = 0.
  $
\end{theorem}
Before giving a proof, we first emphasize that it is manifest that $\mathbb{P}(Y = \calC) \geq 0$ since it is the determinant of a skewsymmetric matrix\footnote{The eigenvalues of a real skewsymmetric matrix are pure imaginary and come in conjugated pairs.}.
It is however less obvious that $\mathbb{P}(Y = \calC) \leq 1$.
We show this fact in the following proof, by proving that $\sum_{\calC} \mathbb{P}(Y = \calC) = 1$.
Second, there is no reason for the above DPP to be invariant if the columns of $\sfX$ are permuted.
In \cref{rem:node_ordering_is_important} in the sequel, we discuss an example in which different orderings of the columns of $\sfX$ yield different DPP laws.
\begin{proof}
  Consider the definition of $\ket{\columns(\sfX)}$ in \eqref{eq:def_psi}.
A simple argument based on Wick's theorem (see \cref{sec:expression_psi} for more details) gives 
\begin{equation}
  \ket{\columns(\sfX)} = \sum_{\substack{0\leq \ell \leq r :\\  \text{parity}(r) = \text{parity}(\ell)}} \sum_{\substack{\calS: |\calS| = \ell\\ \text{ordered}}} 
  \pf 
  \begin{bmatrix}
    0_{\ell } & \sfX_{\calS:}\\
    -(\sfX^\top)_{:\calS} & \skew(\sfX^\top\sfX)
  \end{bmatrix}
  \prod_{i\in \rev(\calS)} c_i^* \ket{\emptyset},\label{eq:expression_colX}
\end{equation}
where $\sfX_{\calS:}\in \mathbb{R}^{\ell \times r}$ is by construction such that $\ell + r$ is even.
Note that, in the formula above, we take the convention that
\[
  \begin{bmatrix}
     & \sfX_{\emptyset:}\\
    -(\sfX^\top)_{:\emptyset} & \skew(\sfX^\top\sfX)
  \end{bmatrix} \triangleq 
    \skew(\sfX^\top\sfX)
  .
\]
To gain intuition, here is an example with $r = 3$. 
Let $i\in \{1,2,3\}$.
The amplitude in front of $c_i^* \ket{\emptyset}$ (corresponding to $\calS = (i)$) reads
\begin{equation*}
  \pf   \begin{bmatrix}
    0 & \sfX_{i1} & \sfX_{i2}& \sfX_{i3}\\
     & 0 & \sfX_{:1}^\top \sfX_{:2} & \sfX_{:1}^\top \sfX_{:3}\\    
     &  &  0 &  \sfX_{:2}^\top \sfX_{:3}\\
     &  &   & 0
  \end{bmatrix} =\sfX_{i1} (\sfX_{:2}^\top \sfX_{:3})  - \sfX_{i2} (\sfX_{:1}^\top \sfX_{:3}) + \sfX_{i3} (\sfX_{:1}^\top \sfX_{:2}).
\end{equation*}
We have -- as a particular case -- for any ordered subset such that $|\calS| = r$,
\begin{eqnarray*}
  \pf \begin{bmatrix}
    0_{r} & \sfX_{\calS :}\\
    - (\sfX^\top)_{:\calS}& \skew(\sfX^\top\sfX)
  \end{bmatrix} = (-1)^{r(r-1)/2} \det(\sfX_{\calS :}),
\end{eqnarray*}
as a consequence of \cref{lem:pf_saddle}, given in appendix,
so that we recover the second term in \cref{eq:remainplusdet}.
Now let $\calC$ be a subset of $\{1,\dots, n\}$ with $|\calC| \leq r$.
Coming back to the expression of $\ket{\columns(\sfX)}$, the probability $\mathbb{P}(Y = \calC)$ is given by the squared modulus of the coefficient of $\prod_{i\in \calC} c_i^* \ket{\emptyset}$ in \cref{eq:expression_colX} as a consequence of \cref{eq:law_Ys} (Born's rule), i.e.,
\begin{equation*}
  \mathbb{P}(Y = \calC) = 
  \left|
    \pf  
    \begin{bmatrix}
    0_{\ell} & \sfX_{\calC:}\\
    -(\sfX^\top)_{:\calC} &\skew(\sfX^\top\sfX) 
    \end{bmatrix}
    \right|^2 
    =
    \det
    \begin{bmatrix}
      0_{\ell} & \sfX_{\calC:}\\
      -(\sfX^\top)_{:\calC} &\skew(\sfX^\top\sfX) 
    \end{bmatrix},
\end{equation*}
since all matrices have real entries here.
Note that this is a well-normalized probability mass function. 
Let us briefly explain why. 
First, the matrix 
\begin{equation*}
  \begin{bmatrix}
    0_{\ell } & \sfX_{\calC:}\\
    -(\sfX^\top)_{:\calC} &\skew(\sfX^\top\sfX) 
  \end{bmatrix}
\end{equation*}
is skewsymmetric with an even number of rows and columns. 
Hence, its determinant cannot be negative.
Second, we show that $\mathbb{P}(Y = \calC)$ sums up to unity.
Let us use the following identity \citep[Theorem 2.1]{kulesza2012determinantal}
\begin{equation}
  \sum_{\substack{\calC^\prime\subseteq \{1,\dots, n^\prime\} \\
\calA^\prime\subseteq \calC^\prime}}\det L_{\calC^\prime\calC^\prime} = \det(L + 1_{\overline{\calA^\prime}}),\label{eq:identity_KT}
\end{equation}
by choosing $n^\prime = r+n$.
Here, we introduced the diagonal matrix $[1_{{\calC}}]_{i,j} = \delta_{ij} 1(i\in \calC)$ whose diagonal contains the indicator vector of $\calC$.
Upon denoting 
\begin{align}
  L = 
  \begin{bmatrix}
    0_{n } & \sfX\\
    -\sfX^\top &\skew(\sfX^\top\sfX) 
  \end{bmatrix},\label{eq:L}
\end{align}
we have that 
\begin{align*}
  \sum_{\substack{\calC^\prime\subseteq \{1,\dots, r+n\} \\
  n+\{1,\dots,r\}\subseteq \calC^\prime}}\det L_{\calC^\prime\calC^\prime} = \det(L + 1_{\overline{n+\{1,\dots,r\}}}) 
  = \det  
  \begin{bmatrix}
    I_{n} & \sfX\\
    -\sfX^\top &\skew(\sfX^\top\sfX) 
  \end{bmatrix}.
\end{align*}
By using the formula for the determinant of block matrices, we find
\begin{equation*}
  \det  \begin{bmatrix}
    I_{n} & \sfX\\
    -\sfX^\top &\skew(\sfX^\top\sfX) 
  \end{bmatrix}
  = 
  \det(\skew(\sfX^\top\sfX) +  \sfX^\top\sfX ).
\end{equation*}
Since $\skew(\sfX^\top\sfX) +  \sfX^\top\sfX $ is triangular with only ones on its diagonal, its determinant is equal to $1$. 
This completes the proof.
\end{proof}
Actually, the proof of \cref{thm:pmf_DPP} inspires the following generalization.
\begin{theorem}\label{thm:pmf_DPP_A}
  Let $X \in \mathbb{R}^{n\times r}$ be a matrix with linearly independent columns and let $A \in \mathbb{R}^{r \times r}$ be skewsymmetric. 
  The determinantal point process with law 
  \begin{equation*}
    \mathbb{P}(Y = \calC) = \frac{1}{\det(A + X^\top X)} \det\begin{bmatrix}
      0_{|\calC| } & X_{\calC:}\\
      -(X_{\calC :})^\top & A
    \end{bmatrix} \text{ with } \calC\subseteq \{1, \dots ,n\}
  \end{equation*}
  is well-defined.
\end{theorem}
\begin{proof}
  Note that $ \det(A + X^\top X) = \det( X^\top X) \det(S + I_{r})$ where $$S = ( X^\top X)^{-1/2} A ( X^\top X)^{-1/2}$$ is skewsymmetric.
  Thus, $ \det(A + X^\top X)>0$. 
  Hence, $\mathbb{P}(Y = \calC) \geq 0$.
  That it is the correct normalization factor follows from the arguments in the proof of \cref{thm:pmf_DPP}.
\end{proof}
Note that similar DPPs -- called \emph{$L$-ensembles} --  were described in \citep{gartrell2019learning,gartrell2021scalable} which consider the case of a skewsymmetric  matrix $L$ as in \cref{eq:L}. 
Incidentally, the DPP in \cref{thm:pmf_DPP_A} involves only specific principal minors of the matrix $L$ in \cref{eq:L}, namely the ones which always include the last $r$ rows and columns.
In light of this remark, this type of DPPs is a skewsymmetric analogue of the so-called \emph{extended $L$-ensembles} of \citet{tremblay2023extended}, where the authors only consider Hermitian matrices $L$.
For completeness, we give the correlation kernel of the DPP defined in \cref{thm:pmf_DPP_A} in \cref{corol:kernel_DPP}, which has a particularly intuitive form since it is a deformation of the orthogonal projector \cref{eq:target_kernel}.
\begin{corollary}\label{corol:kernel_DPP}
  Let $Y$ be the point process of \cref{thm:pmf_DPP_A}. 
  We have 
  \begin{equation*}
    \mathbb{P}(\calA \subseteq Y) = \det(X (A + X^\top X)^{-1}X^\top)_{\calA\calA}.
  \end{equation*}
\end{corollary}

\begin{proof}
  Let us again use the identity \cref{eq:identity_KT}
  by choosing $n^\prime = r+n$ and  $L$ to be the $(n+r)\times (n+r)$ matrix given by
  \begin{align}
    L = 
    \begin{bmatrix}
      0_{n} & X\\
      -X^\top & A
    \end{bmatrix},\label{eq:L}
  \end{align}
  and $\calA^\prime = r+ \calA$ with $\calA \subseteq \{1,\dots ,n\}$.
  At this point, we can compute 
  \[
    \det(L + 1_{\overline{n+\{1,\dots,r\}}}) = \det(A + X^\top X ) = 1,
  \]
  where we used the notation $[1_{{\calC}}]_{i,j} = \delta_{ij} 1(i\in \calC)$.
  Hence, $L + 1_{\overline{n+\{1,\dots,r\}}}$ is invertible.
  Thus, we find
  \begin{align}
    \mathbb{P}(\calA \subseteq Y) &= \sum_{\calC: \calA \subseteq \calC }\mathbb{P}(Y= \calC) \nonumber\\
    &= \det(L + 1_{\overline{\calA^\prime}})\nonumber\\
    &= \det(L + 1_{\overline{n+\{1,\dots,r\}}}) \det\left(L\left(L+ 1_{\overline{n+\{1,\dots,r\}}}\right)^{-1}\right)_{\calA^\prime\calA^\prime}.\label{eq:subset_proba}
  \end{align}
  To obtain the last equality, we used identities from \citep[proof of Theorem 2]{kulesza2012determinantal}.
  The result follows from the use of the well-known formula for the inverse of a block matrix.
\end{proof}
Next, we give in \cref{lem:spectral_decomposition} an expression of the spectral decomposition of the correlation kernel which is later used to derive the law of the number of points sampled by the process.
\begin{lemma}[Spectral decomposition]\label{lem:spectral_decomposition}
  Under the hypotheses of \cref{thm:pmf_DPP_A}, let $K = X (A + X^\top X)^{-1}X^\top$ and define the $r\times r$ matrix $S = (X^\top X)^{-1/2} A (X^\top X)^{-1/2}$.
  On the one side, if $r = 2 k+1$ for some $k\geq 1$, we have 
  \begin{equation}
        K = V\left(v_0 v_0^\top + \sum_{\ell=1}^{k} \frac{1 -\rmi \nu_\ell}{1 + \nu_\ell^2} u_\ell u_\ell^* + \frac{1 + \rmi \nu_\ell}{1 + \nu_\ell^2} \bar{u}_\ell \bar{u}_\ell^*\right)V^\top,\label{eq:K_spectral}
  \end{equation}
  where $V = X(X^\top X)^{-1/2}$ has orthonormal columns.
  In \cref{eq:K_spectral}, we used the Youla decomposition of $S$, namely the normalized eigenpairs $(i\nu_\ell, u_\ell)$ and $(-i\nu_\ell, \bar{u}_\ell)$ with $\nu_\ell$ real for all $\ell\in\{1,\dots,k\}$, whereas $v_0$ is the unpaired eigenvector of $S$ of zero eigenvalue.
  Also, we have the decomposition of the unit vectors $u_\ell = (a_\ell + \rmi b_\ell)/\sqrt{2}$ where $(a_\ell, b_\ell)_\ell$ are mutually orthogonal real unit vectors.
  On the other side, if $r = 2 k$ for some $k\geq 1$, the same decomposition holds in the absence of $v_0$.
\end{lemma}
\begin{proof}
  The result follows from writing
  $$
    K = X(X^\top X)^{-1/2} (S + I_{r})^{-1}(X^\top X)^{-1/2}X^\top, 
  $$
  and by using the Youla decomposition of the real skewsymmetric matrix $S$.
\end{proof}
Before going further, let us understand the expression in \cref{eq:K_spectral}.
Considering one term on the right-hand side for simplicity, we can write the paired eigenvectors in terms of a real part and an imaginary part as follows $u = (a + \rmi b)/\sqrt{2}$ and $\bar{u} = (a - \rmi b)/\sqrt{2}$.
Further, denote $\alpha = 1/(1+ \nu^2)$ and $\beta = \nu/(1+ \nu^2)$, which are such that that $0< \alpha \leq 1$ and $0\leq |\beta| \leq 1$.
Hence, we have
$$
\begin{bmatrix}
  u & \bar{u}
\end{bmatrix}
\begin{bmatrix}
  \alpha-\rmi \beta & 0\\
  0 & \alpha + \rmi \beta
\end{bmatrix}
\begin{bmatrix}
  u^*\\  \bar{u}^*
\end{bmatrix} 
=
\begin{bmatrix}
a & b
\end{bmatrix}
\begin{bmatrix}
  \alpha & \beta\\
  -\beta & \alpha
\end{bmatrix}
\begin{bmatrix}
  a^\top\\  b^\top
\end{bmatrix} 
= \alpha(aa^\top + bb^\top) + \beta(ab^\top - ba^\top),
$$
where we actually see the role played by $\beta$ in the skewsymmetric part of the matrix.

We are now ready to describe the law of the number of points of the process.
\begin{proposition}[Law of cardinal of a sample]
  In the setting of \cref{lem:spectral_decomposition}, let $Y\sim \mathrm{DPP}(K)$ where $K = X (A + X^\top X)^{-1}X^\top$ is $n\times n$ and with an $r\times r$ skewsymmetric matrix $A$.
  If $r = 2 k+1$, the law of the number of points in a sample of $Y$ is
  $$
  |Y| = 1 + 2 \sum_{\ell = 1}^{k}\text{Bern}\left(\frac{1}{1+\nu_\ell^2}\right)
  $$
  where the Bernoulli random variables are independent.
  Otherwise, if $r = 2k$, the law of $|Y|$ is the same as above with the constant ``$1$'' removed.
\end{proposition}
\begin{proof}
  Classically, the law of the cardinal of $Y$ is expressed as a determinant \citep[Corollary 1 in Supplemantary Material]{brunel2018learning}, namely
  \begin{align*}
    \mathbb{E}[z^{|Y|}] = \det( I_{n} - K + z K).
  \end{align*}
  Consider the case of odd $r$.
  Denoting $\mu_\ell = \frac{1 -\rmi \nu_\ell}{1 + \nu_\ell^2}$, the spectral decomposition in \cref{lem:spectral_decomposition} yields
  \begin{align*}
    \mathbb{E}[z^{|Y|}] = z \prod_{\ell=1}^{k}(1-\mu_\ell + z \mu_\ell)(1-\bar{\mu}_\ell + z \bar{\mu}_\ell) = z \prod_{\ell=1}^{k}\left(\frac{\nu_\ell^2}{1 + \nu_\ell^2} +  \frac{z^2}{1 + \nu_\ell^2}\right),
  \end{align*}
  where we observe that the above product involves moment generating functions of Bernoullis with success probability $1/(1 + \nu_\ell^2)$.
  The case of even $r$ follows from the same reasoning.
\end{proof}

\subsection{Example: sampling uniform spanning trees}
In this section, we particularize the DPP of \cref{thm:pmf_DPP} to highlight its connections with a well-known DPP on graph edges.
A spanning tree of a graph is a connected subgraph with the same vertex set and edges forming no cycle.
The edges of a spanning tree sampled uniformly are known to be a projection DPP \citep{Pemantle91}.
The point process of \cref{thm:pmf_DPP} is in fact a generalization of this DPP.

We begin by setting a few definitions.
Let $B$ denote the edge-vertex incidence matrix of an unweighted finite and connected graph. 
The Laplacian of this graph is customarily defined as  $L = B^\top B$.
For further use, denote by $\deg(i)$ the number of neighbors of node $i$ and let $D$ be the diagonal matrix with entry $ii$ being $\deg(i)$.
To sample spanning trees with the algorithm of \cref{sec:algo}, one needs to first choose a node $\varrho$, called the \emph{root}.
Then, we define the matrix $X = B_{:\widehat{\varrho}}$ obtained from $B$ by removing the column of the root node.
Note that $\|X_{:,i}\|^2_2 = \deg(i)$ for all nodes $i\neq \varrho$.
Thus, define the column-normalized incidence matrix $\mathsf{B} = B D^{-1/2}$.
\subsubsection{Probability distribution over dimer-rooted forests}
As can be seen by a direct substitution, the DPP associated with $\ket{\columns(\sfX)}$ is in fact a measure over edges of \emph{dimer-rooted forests}\footnote{We propose this terminology to highlight the role of contractions associated with the expression of pfaffians. We do not know yet their connection with the \emph{half-trees} defined by \citet{de2014principal}.} which are defined in \cref{def:dimer-rooted_forest}. 
We refer to \cref{fig:contracted_tree} for an illustration. 
\begin{definition}[dimer-rooted forest with root $\varrho$]\label{def:dimer-rooted_forest}
  Let $G$ be a finite connected graph with $n_v$ nodes and let $\varrho$ be one of its nodes.
  A dimer-rooted forest of $G$ with root $\varrho$ is a set of edges $\calC$ which satisfies the following conditions:
  \begin{itemize}
    \item[(i)] $\calC$ contains no cycle and is such that $\text{parity}(\calC) =\text{parity}(n_v-1)$,
    \item[(ii)] $\calC$ can be completed to a spanning subgraph by complementing it by an edge-disjoint set of dimers such that the connected components of $\calC$ are either trees rooted to $\varrho$ or trees rooted to one dimer.
  \end{itemize}
\end{definition}
Note that the dimers are not part of the dimer-rooted forest.
These dimer-rooted forests are closely related to dimer covers, namely sets of edges such that each node is the endpoint of exactly one edge. 
In a word, each dimer in the dimer-rooted forests corresponds to a contraction in the expansion of a Pfaffian as detailed in \cref{sec:tech}.
This is illustrated from the graph viewpoint in \cref{fig:square_graph} where the colored edges indicate the superposition of edges coming with the neighbourhood of each colored node.
\begin{figure}[t]
  \centering
  \begin{tikzpicture}[scale=0.80]
      \node[shape=circle,draw=red,line width=1pt] (A) at (0,0) {\tiny $1$}; 
      \node (AD) at (0,-1) {\tiny $13$}; 
      \node[shape=circle,draw=teal] (B) at (2,0) {\tiny $2$}; 
      \node (AB) at (1,0) {\tiny $12$}; 
      \node (BC) at (2,-1) {\tiny $24$}; 
      \node (DC) at (1,-2) {\tiny $34$}; 

      \node[shape=circle,draw=blue,line width=1pt] (C) at (2,-2) {\tiny $4$}; 
      \node[shape=circle,draw=black,line width=1pt] (D) at (0,-2) {\tiny $3$}; 
      \node at (0,-2.5) {\tiny root}; 

      \path [->,draw=red,line width=1pt] (A) edge[bend right] (B);
      \path [->,draw=red,line width=1pt] (A) edge[bend right] (D);
      \path [->,draw=teal,line width=1pt] (A) edge[bend left] (B);
      \path [->,draw=teal,line width=1pt] (B) edge[bend left] (C);
      \path [->,draw=blue,line width=1pt] (B) edge[bend right] (C);
      \path [->,draw=blue,line width=1pt] (D) edge[bend left] (C);

      \node[shape=circle,draw=teal,line width=1pt] (B) at (2,0) {\tiny $2$}; 

      \node at (9,-1) {$\begin{aligned}
        &\ket{\text{columns}(\mathsf{B}_{:\widehat{3}})}\\
        &= {\color{red}(\frac{f_{12}+f_{13}}{\sqrt{\deg(1)}})}{\color{teal}(\frac{f_{24}-f_{12}}{\sqrt{\deg(2)}})}{\color{blue}(\frac{-f_{24}-f_{34}}{\sqrt{\deg(4)}})}\ket{\emptyset}\\
        &\propto (-{\color{red}f_{12}} \underbrace{{\color{teal}f_{24}}{\color{blue}f_{24}}}_{1}   
        +  \underbrace{{\color{red}f_{12}}{\color{teal}f_{12}}}_{1}{\color{blue}f_{24}} 
        - {\color{red}f_{13}}\underbrace{{\color{teal}f_{24}}{\color{blue}f_{24}}}_{1})\ket{\emptyset} \text{\footnotesize(dimer-rooted forest)}\\ 
         & + \underbrace{(-{\color{red}f_{12}}{\color{teal}f_{24}}{\color{blue}f_{34}}-{\color{red}f_{13}} {\color{teal}f_{24}} {\color{blue}f_{34}}
          + {\color{red}f_{13}} {\color{teal}f_{12}}{\color{blue}f_{24}} + {\color{red}f_{13}}{\color{teal}f_{12}} {\color{blue}f_{34}})\ket{\emptyset}}_{\text{ spanning trees}}
      \end{aligned}$}; 

  \end{tikzpicture}
  \caption{A simple square graph and the corresponding Clifford loader $\Cliff(\mathsf{B}_{:1}) \Cliff(\mathsf{B}_{:2})\Cliff(\mathsf{B}_{:4})\ket{\emptyset}$; see \cref{eq:def_psi}. 
  We here defined the normalized incidence matrix $\mathsf{B} = B D^{-1/2}$ and  the notation $f_e = c_e + c_e^*$ for any edge $e$.
  The column $\varrho=3$ is removed from $\mathsf{B}$.
  The colored edges emanate from the node of the same color whereas their orientation is determined (by convention) by the matrix $B$. \label{fig:square_graph}}
\end{figure}
Indeed, in the decomposition of $\ket{\text{columns}(\mathsf{B}_{:\widehat{\varrho}})}$, the origin of the dimers is the annihilation (or contraction, i.e., $(c_e + c_e^*)^2 = \mathbb{I}$) of identical edges coming from neighbouring nodes.

\begin{theorem}[determinantal measure over dimer-rooted forests]\label{thm:contracted_trees}
  Consider a connected graph with $n_v$ nodes and $n_e$ edges. Denote by $B$ its $n_e\times n_v$ incidence matrix. 
  Let $\varrho\in \{1,\dots , n_v\}$ be any node of this graph.
  Let $Y$ be a DPP with law
  \begin{equation*}
    \mathbb{P}(Y = \calC) = \left(\prod_{\substack{i=1\\ i \neq \varrho}}^{n_v}\deg(i)\right)^{-1} \cdot \det\begin{bmatrix}
      0_{|\calC|}& B_{\calC \widehat{\varrho}}\\
      -(B_{\calC \widehat{\varrho}})^\top & \skew(L_{\widehat{\varrho}\widehat{\varrho}})
    \end{bmatrix} \text{ with } \calC\subseteq \{1, \dots ,n_e\}.
  \end{equation*}
  If $\mathbb{P}(Y = \calC)\neq 0$, then  $\calC$ is necessarily the set of simple edges of a dimer-rooted forest; see \cref{def:dimer-rooted_forest}.
\end{theorem}
\begin{proof}
  Define for convenience
  \begin{equation}
    {A} \triangleq \begin{bmatrix}
      0_{|\calC|} & B_{\calC \widehat{\varrho}}\\
      -(B_{\calC \widehat{\varrho}})^\top & \skew(L_{\widehat{\varrho}\widehat{\varrho}})
    \end{bmatrix}.\label{eq:skew-adjacency}
  \end{equation}
  First, we observe that if $\calC$ includes edges forming a cycle, then the corresponding rows of the first block of \cref{eq:skew-adjacency}, i.e.\
  $    \begin{bmatrix}
        0_{|\calC|} & B_{\calC \widehat{\varrho}}\\
      \end{bmatrix}
  $,
  are linearly dependent, since the corresponding rows of $B_{\calC \widehat{\varrho}}$ are linealy dependent.
  Therefore, $\calC$ cannot include any cycle otherwise $\det A = 0$.
  Similarly, $\text{parity}(\calC) =\text{parity}(n_v-1)$, since otherwise $A$ has an odd number of rows and columns and therefore, $\det A = 0$.
  We thus see that $(i)$ in \cref{def:dimer-rooted_forest} is necessary, and in what follows, we assume $\text{parity}(\calC) =\text{parity}(n_v-1)$.
  
  Second, note that $C = \skew(L_{\widehat{\varrho}\widehat{\varrho}})$ can be interpreted as a skewsymmetric adjacency matrix where each edge is endowed with an orientation as follows:
  \[
    {C}_{ij} = \begin{cases}
      +1 & \text{ if } ij \text{ is an oriented edge},\\
      -1 & \text{ if } ji \text{ is an oriented edge},\\
      0 & \text{ otherwise.}
    \end{cases}
  \]
  Recall that we denote edges by pairs $(i,j)$ with $i<j$; see \cref{fig:skew_adjacency} for an illustration.
  \begin{figure}[ht!]
    \centering
      \begin{tikzpicture}[scale=0.80]
          \node[shape=circle,draw=black] (A) at (-1,0) {\tiny $1$};
          \node[shape=circle,draw=black] (C) at (-1,2) {\tiny $2$};
          \node[shape=circle,draw=black] (E) at (2,2) {\tiny $3$};  
          \node (G) at (4,2) {\tiny $\varrho$};  
          \path [-,dashed] (E) edge (G);

          \path [->,thick] (A) edge (C);

          \path [->,thick] (C) edge (E);
          \node[shape=circle,draw=black] (A1) at (1-0.5,0) {\tiny $4$};  
          \path [->,thick] (A) edge (A1);
          \node[shape=circle,draw=black] (A2) at (2,0) {\tiny $5$}; 
          \path [->,thick] (A1) edge (A2);
          \node[shape=circle,draw=black] (A3) at (1-0.5,1) {\tiny $6$}; 
          \path [->,thick] (A1) edge (A3);
          \node[shape=circle,draw=black] (A4) at (2,1) {\tiny $7$}; 
          \path [->,thick] (A3) edge (A4);
          \node[shape=circle,draw=black] (A5) at (4,1) {\tiny $8$}; 
          \path [-,dashed] (A5) edge (G);
          \path [->,thick] (A5) edge (A4);
          \path [->,thick] (A4) edge (E);
          \node[shape=circle,draw=black] (A6) at (4,0) {\tiny $9$}; 

          \path [->,thick] (A5) edge (A6);
          \path [->,thick] (A6) edge (A2);
          \node[shape=circle,draw=black] (A8) at (7,0) {\tiny $10$}; 
          \path [->,thick] (A2) edge (A4);
          \node[shape=circle,draw=black] (A9) at (7,2) {\tiny $11$}; 
          \path [-,dashed] (A9) edge (G);

          \node[shape=circle,draw=black] (A11) at (5.5,1) {\tiny $12$}; 
          \node[shape=circle,draw=black] (A12) at (7,1) {\tiny $13$}; 
          \path [->,thick] (A6) edge (A8);

          \path [->,thick] (A11) edge (A12);
          \path [->,thick] (A5) edge (A11);
          \path [->,thick] (A9) edge (A12);
          \path [->,thick] (A8) edge (A12);
    \end{tikzpicture}
    \caption{Orientation associated to the matrix $\skew(L_{\widehat{\varrho}\widehat{\varrho}})$ for a simple graph. 
    Edge orientations are determined by the ordering of the nodes.
    Dashed edges are missing after removing $\varrho$.
    \label{fig:skew_adjacency}}
  \end{figure}
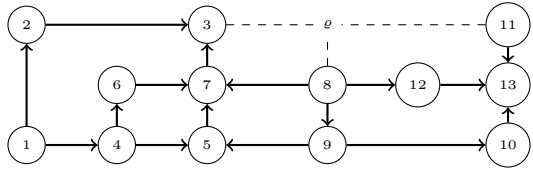
  Hence, by examining \cref{eq:skew-adjacency}, we realize  that the matrix ${A}$  can be interpreted as a skewsymmetric adjacency matrix of an augmented graph where the edges in $\calC$ are considered as additional vertices, namely whose vertex set is $$\mathcal{V} = \calC\cup \{1, \dots,\widehat{\varrho}, \dots, n_v\},$$
  where $\dots,\widehat{\varrho},\dots$ indicates that the vertex $\varrho$ is missing.
  Above, each edge $ij$ in $\calC$ is considered as an extra node of the augmented graph which is only connected to its neighbours $i$ and $j$; see \cref{fig:augmented_graph} for an illustration.
  Now, we recall that $\det A = (\pf {A})^2$ since ${A}$ is skewsymmetric with an even number of rows as a consequence of $(i)$. 
  By using the definition of the Pfaffian as a sum over contractions (see \cref{sec:expression_psi} for the exact expression), we see that, if $\pf {A} \neq 0$, there should exist at least one permutation $\sigma$ such that 
  $
  {A}_{\sigma(1)\sigma(2)}\dots {A}_{\sigma(2k-1)\sigma(2k)}\neq 0,
  $
  with $2k = |\calC| + (n_v - 1)$.\footnote{Classically, this matrix can be extended to a Tutte matrix by introducing as many indeterminates as edges. 
  The determinant of this Tutte matrix is a polynomial of the indeterminates which is not identically zero (as a polynomial) if and only if there exists a dimer configuration (also called perfect matching) of the graph.}
  Hence, this permutation determines a dimer cover of the augmented graph.
  Thus, coming back to the expression \cref{eq:skew-adjacency},  $\det A \neq 0$ implies that there exists a dimer configuration of the augmented graph, which by definition should visit all the nodes of the augmented graph, and hence should include a pair $(e,v_e)$ for all $e \in \calC$ where $v_e$ is an endpoint of $e$.
  This yields to condition $(ii)$ in \cref{def:dimer-rooted_forest}.
  Thus, the conditions $(i)$ and $(ii)$ are necessary.

\end{proof}

  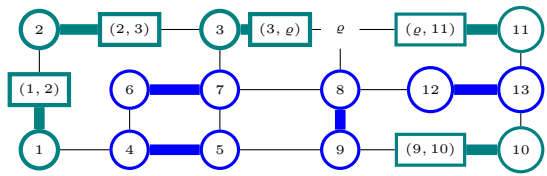
\begin{figure}[ht!]
    \centering
      \begin{tikzpicture}[scale=0.80]
          \node[shape=circle,draw=teal, ultra thick] (A) at (-1,0) {\tiny $1$};
          \node[shape=rectangle,draw=teal, ultra thick] (B) at (-1,1) {\tiny $(1,2)$};
          \node[shape=circle,draw=teal, ultra thick] (C) at (-1,2) {\tiny $2$};
          \node[shape=rectangle,draw=teal, ultra thick] (D) at (1-0.5,2) {\tiny $(2,3)$};  
          \node[shape=circle,draw=teal, ultra thick] (E) at (2,2) {\tiny $3$};  
          \node[shape=rectangle,draw=teal, ultra thick] (F) at (3,2) {\tiny $(3,\varrho)$};  
          \node[shape=circle] (G) at (4,2) {\tiny $\varrho$};  
          \path [-,draw=teal,line width=4pt] (A) edge (B);
          \path [-] (B) edge (C);
          \path [-,draw=teal,line width=4pt] (C) edge (D);
          \path [-] (D) edge (E);
          \path [-,draw=teal,line width=4pt] (E) edge (F);
          \path [-] (F) edge (G);
          \node[shape=circle,draw=blue,very thick] (A1) at (1-0.5,0) {\tiny $4$};  
          \path [-] (A) edge (A1);
          \node[shape=circle,draw=blue,very thick] (A2) at (2,0) {\tiny $5$}; 
          \path [-,blue,line width=4pt] (A1) edge (A2);
          \node[shape=circle,draw=blue,very thick] (A3) at (1-0.5,1) {\tiny $6$}; 
          \path [-] (A1) edge (A3);
          \node[shape=circle,draw=blue,very thick] (A4) at (2,1) {\tiny $7$}; 
          \path [-,blue,line width=4pt] (A3) edge (A4);
          \node[shape=circle,draw=blue,very thick] (A5) at (4,1) {\tiny $8$}; 
          \path [-] (A5) edge (G);
          \path [-] (A5) edge (A4);
          \path [-] (A4) edge (E);
          \node[shape=circle,draw=blue,very thick] (A6) at (4,0) {\tiny $9$}; 
          \path [-,blue,line width=4pt] (A5) edge (A6);
          \path [-] (A6) edge (A2);
          \node[shape=rectangle,draw=teal,very thick] (A7) at (5.5,0) {\tiny $(9,10)$}; 
          \path [-] (A7) edge (A6);
          \node[shape=circle,draw=teal,very thick] (A8) at (7,0) {\tiny $10$}; 
          \path [-,teal,line width=4pt] (A7) edge (A8);
          \path [-] (A2) edge (A4);
          \node[shape=circle,draw=teal,very thick] (A9) at (7,2) {\tiny $11$}; 
          \node[shape=rectangle,draw=teal,very thick] (A10) at (5.5,2) {\tiny $(\varrho,11)$}; 
          \path [-,teal,line width=4pt] (A9) edge (A10);
          \path [-] (G) edge (A10);
          \node[shape=circle,draw=blue,very thick] (A11) at (5.5,1) {\tiny $12$}; 
          \node[shape=circle,draw=blue,very thick] (A12) at (7,1) {\tiny $13$}; 
          \path [-,blue,line width=4pt] (A11) edge (A12);
          \path [-] (A5) edge (A11);
          \path [-] (A9) edge (A12);
          \path [-] (A8) edge (A12);
    \end{tikzpicture}
    \caption{An augmented graph covered by a dimer configuration (thick edges).
    Here, the edges in $\calC = \{(1,2),(2,3),(3,\varrho),(11,\varrho),(10,9)\}$ are depicted as teal rectangles whereas the perfect matching spanning the remaining nodes is colored in blue. 
    The nodes of the original graph are displayed as circles.
    \label{fig:augmented_graph}}
  \end{figure}
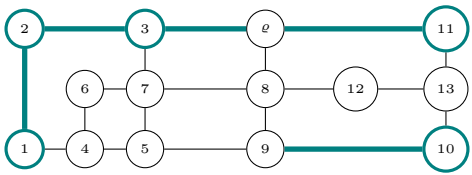
\begin{figure}[ht!]
  \centering
    \begin{tikzpicture}[scale=0.80]
        \node[shape=circle,draw=teal,very thick] (A) at (0,0) {\tiny $1$};
        \node[shape=circle,draw=teal,very thick] (C) at (0,2) {\tiny $2$};
        \node[shape=circle,draw=teal,very thick] (E) at (2,2) {\tiny $3$};  
        \node[shape=circle,draw=black] (G) at (4,2) {\tiny $\varrho$};  
        \path [-,draw=teal,line width=2pt] (A) edge (C);
        \path [-,draw=teal,line width=2pt] (C) edge (E);
        \path [-,draw=teal,line width=2pt] (E) edge (G);
        \node[shape=circle,draw=black] (A1) at (1,0) {\tiny $4$};  
        \path [-] (A) edge (A1);
        \node[shape=circle,draw=black] (A2) at (2,0) {\tiny $5$}; 
        \path [-] (A1) edge (A2);
        \node[shape=circle,draw=black] (A3) at (1,1) {\tiny $6$}; 
        \path [-] (A1) edge (A3);
        \node[shape=circle,draw=black] (A4) at (2,1) {\tiny $7$}; 
        \path [-] (A3) edge (A4);
        \node[shape=circle,draw=black] (A5) at (4,1) {\tiny $8$}; 
        \path [-] (A5) edge (G);
        \path [-] (A5) edge (A4);
        \path [-] (A4) edge (E);
        \node[shape=circle,draw=black] (A6) at (4,0) {\tiny $9$}; 
        \path [-] (A5) edge (A6);
        \path [-] (A6) edge (A2);
        \node[shape=circle,draw=teal,very thick] (A8) at (7,0) {\tiny $10$}; 
        \path [-] (A2) edge (A4);
        \node[shape=circle,draw=teal,very thick] (A9) at (7,2) {\tiny $11$}; 
        \path [-,teal,line width=2pt] (A9) edge (G);
        \node[shape=circle,draw=black] (A11) at (5.5,1) {\tiny $12$}; 
        \node[shape=circle,draw=black] (A12) at (7,1) {\tiny $13$}; 
        \path [-] (A11) edge (A12);
        \path [-] (A5) edge (A11);
        \path [-] (A9) edge (A12);
        \path [-] (A8) edge (A12);
        \path [-,teal,line width=2pt] (A6) edge (A8);
  \end{tikzpicture}
  \caption{A dimer-rooted forest with root $\varrho$ and whose edges -- colored in teal -- are $\calC = \{(1,2),(2,3),(3,\varrho),(11,\varrho),(10,9)\}$. 
  This dimer-rooted forest is associated with the augmented graph of \cref{fig:augmented_graph}.
  \label{fig:contracted_tree}}
\end{figure}
Considering a connected graph with no self-loop, let $D = \Diag(\deg)$ denote the diagonal degree matrix. It is customary to define the normalized Laplacian as 
\begin{equation}
L^{(N)} = D^{-1/2}L D^{-1/2}. \label{eq:normalized_laplacian}
\end{equation}
Key properties of $L^{(N)}$ are that its diagonal contains only ones and that it is positive semi-definite.
Let $L^{(N)}_{\widehat{\varrho}\widehat{\varrho}}$ be the normalized Laplacian with the row and column associated with the root $\varrho$ removed, which is non-singular.
By Hadamard's inequality, $\det L^{(N)}_{\widehat{\varrho}\widehat{\varrho}} \leq 1$.
The acceptance probability of the algorithm of \cref{sec:algo} reads
\begin{equation}
  a \triangleq \det L^{(N)}_{\widehat{\varrho}\widehat{\varrho}} = \frac{\det L_{\widehat{\varrho}\widehat{\varrho}}}{\prod_{i\neq \varrho}\deg(i)} = \frac{\sharp \text{STs}}{\prod_{i\neq \varrho}\deg(i)},\label{eq:a_as_det}
\end{equation}
where $\sharp \text{STs}$ denotes the number of spanning trees of the graph\footnote{Recall that $\det L_{\widehat{\varrho}\widehat{\varrho}}$ does not depend on $\varrho$ as a consequence of the matrix tree theorem.}.
This expression shows that, for a given graph, we can maximize the acceptance probability in \cref{a:rejection_sampling_Clifford} by choosing a node $\varrho$ with maximum degree. 
To fix ideas, we now list two examples of graphs which exhibit different typical decays for $a$.
\begin{example}[Complete graph]\label{ex:acceptance_complete_graph}
  For the complete graph of $n$ nodes with $n\geq 2$, $\sharp \text{STs} = n^{n-2}$; see e.g., \cref{fig:small_complete_plot} for an illustration. Thus, in this case, the acceptance probability reads  
  $$
  a = n^{-1} \left(\frac{n}{n-1}\right)^{n-1} = \Theta(n^{-1}), \text{  for any root node } \varrho.
  $$
  It is easy to see that it satisfies $2/n \leq a \leq e /n$ where $e$ is Euler's constant.
\end{example}

\begin{example}[Barbell graph]\label{ex:acceptance_barbell_graph}
  Consider a barbell graph of $n$ nodes with two complete graphs of $n^{\prime}$ nodes connected by a line graph of $n^{\prime \prime}$ nodes; see e.g., \cref{fig:small_barbell_plot} for an illustration with $n^\prime=3$ and $n^{\prime \prime} = 0$. In that case, by noting that $\sharp \text{STs} = n^{\prime 2{n}^\prime-4}$ for $n^\prime\geq 2$, we find that the acceptance probability reads
  $$
  a = \deg(\varrho) \left(\frac{n^\prime}{n^\prime -1 }\right)^{2 n^\prime -2} n^{\prime -4} 2^{-n^{\prime\prime }}, \text{ with $\varrho$ the root node}.
  $$
  Also, we have $4 n^{\prime -4} 2^{-n^{\prime\prime }} \leq a /\deg(\varrho) \leq e^2 n^{\prime -4} 2^{-n^{\prime\prime }}$. 
  In order to maximize the acceptance probability, we choose $\varrho$ as a connecting node between a clique and the line graph. 
  Then, we have $a = \Theta(n^{\prime -3}2^{-n^{\prime\prime }})$ where we can take for example $n^\prime = n^{\prime\prime } = n/3$.
\end{example}
Thus, the decay of the acceptance probability in \cref{ex:acceptance_barbell_graph} is very fast as the graph grows in size.
In \cref{sec:amplitude_amplification}, we propose a circuit to boost this probability.
Before reaching this section, we compare the determinantal formula \cref{eq:a_as_det} for $a$ to similar quantities appearing in rejection sampling (RS).
We indeed show in \cref{sec:vanilla_RS} that $a$ is the acceptance probability of a closely related classical RS algorithm.
%
\subsubsection{Acceptance probability for vanilla rejection sampling \label{sec:vanilla_RS}}
  Interestingly, the expression of the acceptance probability \cref{eq:a_as_det} -- the inverse number of rejections before acceptance -- can be compared to a similar quantity appearing in a famous classical algorithm for sampling uniform spanning trees, i.e., Wilson's cycle-popping algorithm \citep{Wilson96}.
  This algorithm which has both a random walk formulation or a stacks-of-cards formulation is discussed more generally in the \emph{partial rejection sampling} framework (PRS) in 
  \citep{jerrum2024fundamentals} and \citep{guo2019uniform}.
  First, by using the expression \cref{eq:normalized_laplacian}, we rewrite the acceptance probability in \cref{a:rejection_sampling_Clifford} as follows
  \begin{equation}
      a  = \det(I_{n-1} - \Pi_{\widehat{\varrho}\widehat{\varrho}}),\label{eq:a_as_det_Pi}
  \end{equation}
  where $\Pi$ is the transition matrix of the customary random walk with uniform jumps to neighbours, namely, $\Pi_{ij} = 1/\deg(i)$ if $j\sim i$ and $\Pi_{ij} = 0$ otherwise.
  At this point, we connect the expression of $a$ to the stacks-of-cards description of Wilson's cycle-popping algorithm.
  
  Consider a random variable $U_i$ per node $i\neq \varrho$ which yields a neighbour of $i$ -- say $j$ -- with probability $\Pi_{ij}$.
  Each $U_i$ is visualized as a card positioned over node $i$, and its value is interpreted as an edge between $i$ and a neighbour $U_i$.
  We assume that all these variables are mutually independent.
  By inspection, a realization of $U_{1}, \dots , \widehat{U_{\varrho}}, \dots, U_{n}$ yields an oriented spanning subgraph which may contain cycles. 
  By \emph{cycles}, we mean oriented cycles including backtracks (2-cycles).
  Denote this oriented spanning subgraph by $\mathrm{subgraph}(U_{1}, \dots , \widehat{U_{\varrho}}, \dots, U_{n})$.
  The stacks-of-cards algorithm goes as follows.
  First, fix an order of the cycles.
  \begin{itemize}
    \item Sample $U_{1}, \dots , \widehat{U_{\varrho}}, \dots, U_{n}$ independently.
    \item While $\mathrm{subgraph}(U_{1}, \dots , \widehat{U_{\varrho}}, \dots, U_{n})$ contains cycles, resample (or pop) all the $U_i$'s which are part of the smallest cycle in the order. Otherwise, stop and output the spanning tree given by the $U_i$'s.
  \end{itemize}
  Note that this algorithm finishes with probability one and that it outputs a uniform spanning tree as well as the history of erased cycles.
  These `popped' cycles can be organized as a \emph{heap of cycles}, i.e., a combinatorial structure defined by 
  \citep{viennot2006heaps};  see Section 5.7.2 of \citep{fanuel2024number} for a pedestrian exposition.
  Following \citep{jerrum2024fundamentals}, the expression \cref{eq:a_as_det_Pi} equals
  \begin{equation}
      a = \Pr_{\prod_{i\neq \varrho}U_i}[\mathrm{subgraph}(U_{1}, \dots , \widehat{U_{\varrho}}, \dots, U_{n}) \text{ contains no cycle}],\label{eq:Prob_no_cycle}
  \end{equation}
  or, in other words, $a$ is the probability that -- if $U_{1}, \dots , \widehat{U_{\varrho}}, \dots, U_{n}$ are sampled independently --  $\mathrm{subgraph}(U_{1}, \dots , \widehat{U_{\varrho}}, \dots, U_{n})$ is a spanning tree.

  Thus, we see that \cref{eq:Prob_no_cycle} is exactly the inverse expected number of rejections before acceptance if the edges are sampled following $\prod_{i\neq \varrho}U_i$.
  This is the classical case of \emph{vanilla} RS.
  In contrast, according to \citep[Corollary 7]{jerrum2024fundamentals}, the expected number of iteration before acceptance in the \emph{partial rejection sampling} (PRS) framework (Wilson's algorithm) reads 
  $$
  \frac{\Pr_{\prod_{i\neq \varrho}U_i}[\mathrm{subgraph}(U_{1}, \dots , \widehat{U_{\varrho}}, \dots, U_{n}) \text{ contains exactly one cycle}]}
  {\Pr_{\prod_{i\neq \varrho}U_i}[\mathrm{subgraph}(U_{1}, \dots , \widehat{U_{\varrho}}, \dots, U_{n}) \text{ contains no cycle}]},
  $$
  which can be much smaller compared with $1/a$ (\emph{vanilla} RS \cref{eq:Prob_no_cycle}) since the numerator can be much smaller than $1$.


%

\subsubsection{Sketching acceptance probability in the graph case\label{sec:sketch_graph_case}}
We discuss here the choice of parameter $\theta_1$ for the sketching algorithm of \citet{boutsidis2017randomized} and the associated complexity \cref{eq:cost_sketch_theta_1}.
The basic idea is that the complexity of the sketch is low provided that the graph has a good connectivity.
Recall that we choose $\sfX = \mathsf{B}_{:\widehat{\varrho}}\in \mathbb{R}^{n\times r}$ where $\mathsf{B} = B D^{-1/2}$ is the column-normalized incidence matrix of a connected graph with $n$ edges and $n_v = r + 1$ nodes.
\paragraph*{A connected graph with a hub node.}
To bound the spectrum of $\sfX^\top \sfX = L^{(N)}_{\widehat{\varrho}\widehat{\varrho}}$ from below, we begin by proving the following lemma, inspired from \citep[Section 3.3]{han2015large}.
\begin{lemma} \label{lem:lower_bound_L}
  Consider a connected graph with no self-loop and let $L = B^\top B$ be its combinatorial Laplacian. If there exists a node $\varrho$ connected to all the other nodes, we have
  $\lambda_{\min}(L_{\widehat{\varrho}\widehat{\varrho}}) \geq 1$.
\end{lemma}
\begin{proof}
  We use the Gershgorin circle theorem: $\lambda_{\min}(L_{\widehat{\varrho}\widehat{\varrho}})$ is at least in one closed disc $i_\star$ of center $\deg(i_\star)$ and radius $\sum_{j\neq i_\star, j\neq \varrho}1(i_\star\sim j)$.
  Since the graph is connected, the latter quantity reads  $\deg(i_\star) - 1$.
  Thus, $\lambda_{\min}(L_{\widehat{\varrho}\widehat{\varrho}})\geq \deg(i_\star) -(\deg(i_\star) - 1) = 1$.
\end{proof}
Under the hypotheses of \cref{lem:lower_bound_L}, by considering a Rayleigh quotient involving $L^{(N)} = D^{-1/2} L D^{-1/2}$, it is easy to see that
$
\lambda_{\min}(L^{(N)}_{\widehat{\varrho}\widehat{\varrho}}) \geq 1/\max_{i\neq \varrho}\deg(i).
$
Since the maximal degree of a node is a graph with $n_v$ nodes is smaller or equal to $n_v - 1$, we have the lower bound 
$$
\lambda_{\min}(L^{(N)}_{\widehat{\varrho}\widehat{\varrho}}) \geq 1/(n_v - 1).
$$
Therefore, we can choose the strict lower bound $\theta_1 = 1/n_v = 1/(r + 1)$.
Finally, since $\text{nnz}(\sfX^\top \sfX) = (n_v - 1) + 2 (n_e -\deg(\varrho))$ where $n=n_e$ is the number of edges, we find that the complexity \cref{eq:cost_sketch_theta_1} reads
$
\text{nnz}(\sfX^\top \sfX)/\theta_1 = \mathcal{O}(nr),
$
 with a dependence on $nr$ as announced in \cref{sec:intro}.
 \paragraph*{A path graph.}
Consider a path graph of $n_v > 1$ nodes with $\varrho$ being an endpoint node. 
Note that $L_{\widehat{\varrho}\widehat{\varrho}}$ is a $(n_v - 1) \times (n_v - 1)$ tridiagonal matrix of a specific form which  was studied by \citet{da2007eigenvalues}, in which case we have
$
\lambda_{\min}(L_{\widehat{\varrho}\widehat{\varrho}}) = 4 \sin^2(\pi/(4n_v - 2)).
$
Now, we consider the normalized Laplacian.
Since node degrees are either equal to $1$ or to $2$, we find that
$$
2 \sin^2\left(\frac{\pi}{4n_v - 2}\right) \leq \lambda_{\min}(L^{(N)}_{\widehat{\varrho}\widehat{\varrho}}) \leq 4 \sin^2\left(\frac{\pi}{4n_v - 2}\right).
$$
Thus we can choose $\theta_1 = \sin^2(\pi/(4n_v - 2))$ which is strictly smaller than $\lambda_{\min}(L^{(N)}_{\widehat{\varrho}\widehat{\varrho}})$.
Hence, $1/\theta_1 = \Theta(n_v^2)$.
Recalling that we consider a path graph with $n_v = r + 1$ nodes, we conclude that $
\text{nnz}(\sfX^\top \sfX)/\theta_1 = \Theta(nr^2)$.

%
\section{Sampling determinantal dimer-rooted forests \label{sec:numerics}}
In this section, we illustrate the properties of the DPP on graph edges given in \cref{thm:contracted_trees} by using both Qiskit simulation tools and a quantum computer.
We sample dimer-rooted forests since they are convenient for subset vizualization. 
Furthermore, our algorithms have the advantage to be exact so that the output has the correct structure, e.g., upon conditioning on the cardinal, the sampled edges form a spanning tree.
This is in contrast with approximate sampling methods based on down-up walks as in \citep{anari2024optimal} which are faster in theory but are unlikely to be structure-preserving. 
Before describing the numerics, we briefly describe loader architectures to improve the text consistency whereas more details are provided in \cref{app:Clifford} for the interested reader.
\subsection{Architectures for the loaders \label{sec:architectures_loaders}}
\begin{figure}[h!]
  \centering
  \begin{subfigure}{\textwidth}
        \centering
        \includegraphics[scale = 0.2,clip, trim=2.6cm 0cm 0cm 0cm]{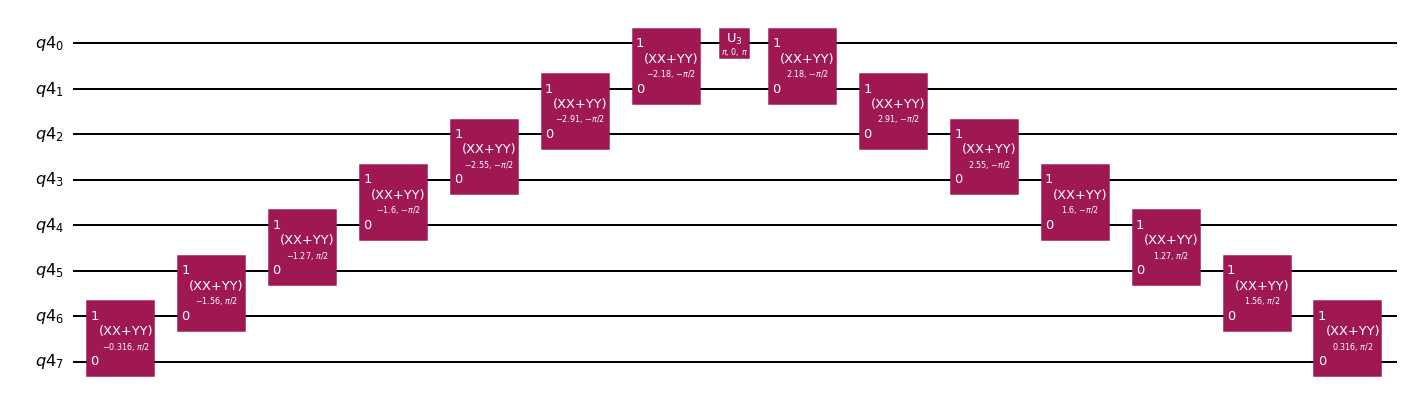}
        \caption{Pyramid Clifford loader.\label{fig:parallel_Clifford}}
  \end{subfigure}
  \begin{subfigure}{\textwidth}
        \centering
        \includegraphics[scale = 0.2,clip, trim=2.6cm 0cm 0cm 0cm]{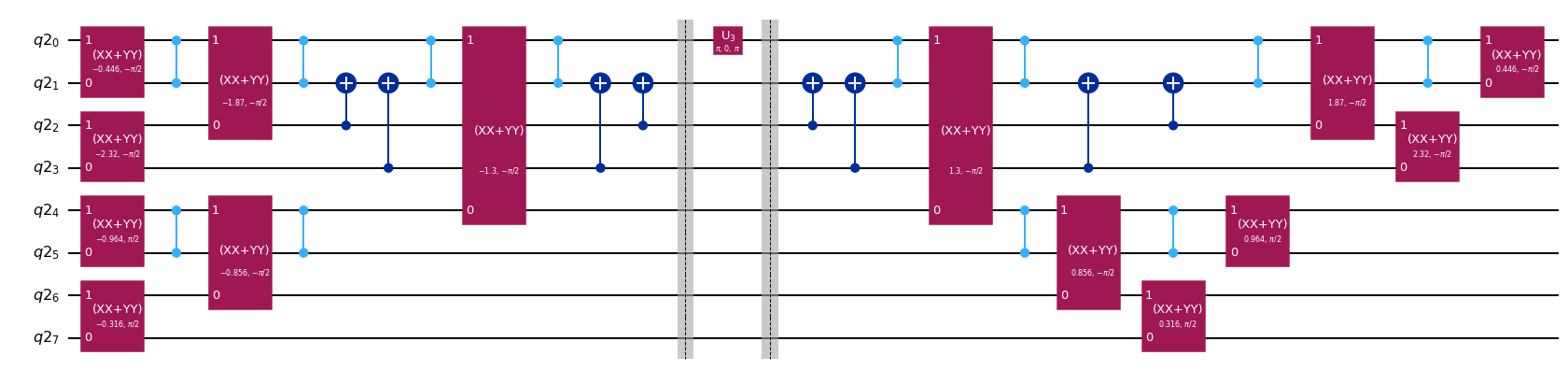}
        \caption{Parallel Clifford loader. \label{fig:pyramid_Clifford}}
  \end{subfigure}
\caption{A pyramid and a parallel Clifford architecture for $\Cliff(\sfx)$ where $\sfx$ is a $8\times 1$ unit vector.
In the pyramid loader, the depth is linear and all the gates are RBSs (red $\mathtt{XX+YY}$ gates); see \cref{lem:RBS}. 
In the parallel case, the depth is logarithmic and FBSs are used when the pairs of qubits are not consecutive (red $\mathtt{XX+YY}$ sandwiched with blue parity gates); see \cref{lem:FBS}. 
Note that the (grey) barriers are only there for an illustrative purpose.\label{fig:parallel_pyramid}
}
\end{figure}
Let a real vector $\sfx = [\sfx_1, \dots, \sfx_n]^\top$ such that $\|\sfx\|_2 = 1$.
We also fix an index $\ell_\triangle$ such that $1\leq \ell_\triangle \leq n$.
Denote $f_{i} = c_i + c_i^*$ for all $1\leq i \leq n$.
The Clifford Loader \cref{eq:Clifford_operator} can be written as the following operator
\begin{equation}
  \Cliff(\sfx) = \sfx_1 \cdot f_1 + \dots + \sfx_n\cdot f_n  =  \calU_{\ell_\triangle}(\sfx) f_{\ell_\triangle} \calU_{\ell_\triangle}(\sfx)^*,\label{eq:loader}
\end{equation}
where $\calU_{\ell_\triangle}(\sfx)$ is an appropriate unitary operator which depends on the index $\ell_\triangle$ and of the chosen circuit architecture. 
\citet*{kerenidis2022quantum} call $\calU_{\ell_\triangle}(\sfx)$ an \emph{unary} data loader and fix $\ell_\triangle= 1$, whereas they promote two circuit architectures: \emph{pyramid}  and  \emph{parallel} loaders, as mentioned in \cref{sec:CAR_and_Clifford}.
We refer to \cref{fig:parallel_Clifford} and \cref{fig:pyramid_Clifford} for an illustration of these circuits for which we also took $\ell_\triangle= 1$, as it can be seen by localizing the central $\mathtt{X} = \mathtt{U_3(\pi,0,\pi)}$ gate.
In light of these figures and by analogy with architecture, we call $\ell_\triangle$ the \emph{pyramidion} of the loader.
Following \citep*{kerenidis2022quantum}, we write $\calU_{\ell_\triangle}(\sfx)$ as a composition of two-qubit operators called Givens operators which are either RBS gates if the qubits are consecutive, or FBS gates otherwise.
The basic idea is that, for any $1\leq i < j \leq n$, we can code a Givens rotation by an angle $\theta$ with the help of an operator $\calU_{ij}(\theta)$ such that
\begin{align*}
  \begin{bmatrix}
    \calU_{ij} \ c_{i}^* \ \calU_{ij}^*\\
    \calU_{ij}\ c_{j}^* \ \calU_{ij}^*
  \end{bmatrix}
   =
   \begin{bmatrix}
    \cos \theta & \sin \theta\\
    -\sin \theta & \cos \theta 
   \end{bmatrix} 
   \begin{bmatrix}
    c_{i}^*\\
    c_{j}^*
  \end{bmatrix} \text{ and } \calU_{ij} c_{k}^*\calU_{ij}^* = c_{k}^*, \quad \forall k \notin\{i,j\}.
\end{align*}
Hence, recalling that $f_{i} = c_i + c_i^*$ for all $1\leq i \leq n$, the coordinates of $\sfx$ in the linear combination \cref{eq:loader} can be obtained by successively applying Givens operators as follows:
\begin{equation}
  1 \cdot f_{\ell_\triangle}^* \xrightarrow[]{ \calU_{\ell_\triangle j_1}} \cos \theta_1 \cdot f_{\ell_\triangle} + \sin \theta_1 \cdot f_{j_1} \xrightarrow[]{\dots} \dots \xrightarrow[]{\dots} \sfx_1 \cdot f_1 + \dots + \sfx_n \cdot f_n. \label{eq:chain_of_Us}
\end{equation}
By reading this sequence of operations from right to left, we observe that acting by conjugation with the inverse unitary operator puts a zero in front of a factor $f$. 
For instance, the inverse of the operation on the left-hand side in \cref{eq:chain_of_Us} puts a zero in front of $f_{j_1}$.
Thus, in this paper, the parameters of the loaders are identified by reading \cref{eq:chain_of_Us} from right to left and by determining the rotation angles in order to zero out a specific coefficient.
The sequence of Givens (zeroing) transformations \cref{eq:chain_of_Us} is determined by the architecture.
In the \emph{pyramid architecture}, Givens operators act on all the consecutive pairs of qubits in a staircase manner; see \cref{fig:pyramid_Clifford}. 
Namely, from right to left, by acting on $n-1$ and $n$, a zero is put in front on $f_n$; and this operation is repeated next by using qubits $n-2$ and $n-1$, etc, until only $1 \cdot f_1$ remains.
In contrast, the \emph{parallel architecture} involves pairs of qubits organized as in a half-interval search; see \cref{fig:parallel_Clifford}.
\begin{remark}
  Note that, in the process of construction of the pyramid and parallel loaders, we have to assume that the entry which has to be zeroed out is not already equal to zero.
  In the latter case, no gate has to be appended to the decomposition.
\end{remark}
The operator $\calU_{ij}(\theta)$ is represented in the computational basis as an FBS gate (or RBS if $j = i+1$).
RBS gates are implemented with $\mathtt{XX+YY}$ gates, whereas FBS gates also require controlled-$\sigma_z$ gates and  controlled-$\sigma_x$ gates; see \cref{fig:FBS_gate_in_detail} for an illustration.
For convenience, these gates are discussed in \cref{sec:RBS} and \cref{sec:FBS}, respectively.
Different choices of $\ell_\triangle$ and of decompositions of $U_{\ell_\triangle}(\sfx)$ correspond to different architectures.
\begin{figure}
  \centering
  \begin{subfigure}[b]{0.45\textwidth}
    \centering
    \includegraphics[scale = 0.2,clip, trim=2.6cm 0cm 0cm 0cm]{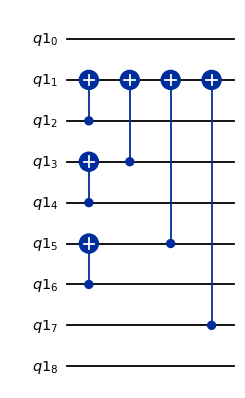}
    \caption{Log-depth parity gate.\label{fig:parity_gate}}
  \end{subfigure}
  \begin{subfigure}[b]{0.45\textwidth}
    \centering
    \includegraphics[scale = 0.2,clip, trim=2.8cm 10cm 0cm 0cm]{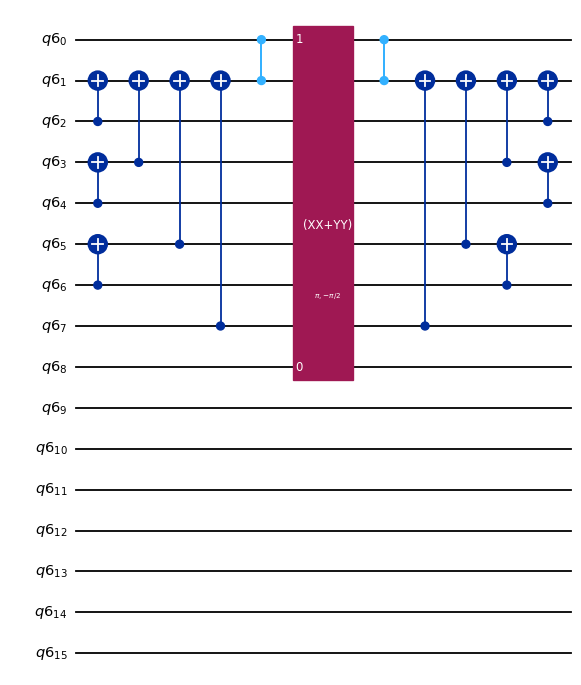}
    \caption{FBS gate.\label{fig:FBS_gate_in_detail}}
  \end{subfigure}
  \caption{Left-hand side: logarithmic depth gate computing the parity of the qubits between qubit $0$ and qubit $8$. 
  This parity is stored in qubit $1$.
  Right-hand side: FBS gate between qubit $0$ and qubit $8$.\label{fig:FBS_gate}}
\end{figure}

The pyramid and parallel architectures are suited for loading dense vectors since they do not adapt to the sparsity structure of $x$. 
In \cref{sec:sparse_architecture}, we propose another architecture called \emph{sparse} Clifford loaders in order to load vectors such as columns of an edge-vertex incidence matrix of a graph.
In that case, the pyramidion of the loader is chosen according to the sparsity pattern of $\sfx$.
The sparse loader of $\sfx$ requires $\mathrm{nnz}(\sfx)$ two-qubit gates and has $\mathcal{O}(n)$ depth.

%

\subsubsection{Sparse architecture \label{sec:sparse_architecture}}

\begin{figure}[h!]
  \begin{subfigure}[b]{0.3\textwidth}
        \centering
        \includegraphics[scale = 0.28,clip, trim=2.6cm 0cm 0cm 0cm]{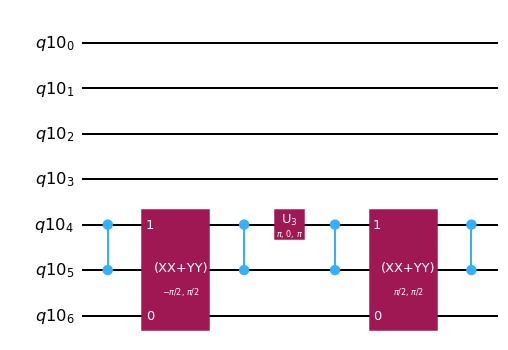}
        \caption{Sparse Clifford loader. }
  \end{subfigure}
  \hfill
  \begin{subfigure}[b]{0.3\textwidth}
  \centering
    \begin{tikzpicture}[scale = 0.5]
      \draw (0, 0) node (0){$ \kbordermatrix{
         &\\
         (0,1) &0 \\
        (0,2) &0 \\
         (1,2) &0 \\
        (2,\varrho) &0 \\ 
        {\color{blue}(\varrho,4)} &\frac{1}{\sqrt{2}}\\
        (\varrho,5) &0\\
        {\color{blue}(4,5)} &\frac{-1}{\sqrt{2}}
      }$};
    \end{tikzpicture}
  \caption{$\sfx = B_{:4}$\label{fig:sparse_Clifford}}
  \end{subfigure}
  \hfill
    \begin{subfigure}[b]{0.3\textwidth}
      \centering
      \begin{tikzpicture}[scale = 1.]
        \draw
          (0, 0) node[shape=circle,draw=black] (0){0}
          (0, 1) node[shape=circle,draw=black] (1){1}
          (1, 0.5) node[shape=circle,draw=black] (2){2}
          (2, 0.5) node[shape=circle,draw=black] (3){$\varrho$}
          (3, 0) node[shape=circle,draw=blue] (4){\color{blue} 4}
          (3, 1) node[shape=circle,draw=black] (5){5};
          \draw[-,thick] (0) to (1);
          \draw[-,thick] (0) to (2);
          \draw[-,thick] (1) to (2);
          \draw[-,thick,blue] (4) to (5);
        \draw[-,thick]  (2) to (3);
        \draw[-,thick,blue]  (3) to (4);
        \draw[-,thick]  (3) to (5);
      \end{tikzpicture}
      \caption{Barbell graph.}
    \end{subfigure}
  \caption{Loading the column of the edge-vertex incidence matrix $B$ (barbell graph) corresponding to vertex $4$  with a sparse Clifford architecture made of one FBS gate. 
  In this case, the pyramidion is $\ell_{\triangle} = 5$, i.e., edge $(\varrho,4)$.}
\end{figure}
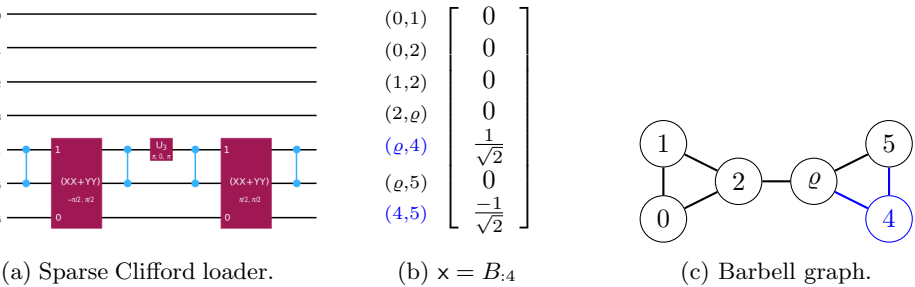
To construct sparse versions of Clifford loaders, we adapt the strategy of \citep{JSKSB18}\footnote{\citet{JSKSB18} proposed circuits to prepare Slater determinants by acting with Givens rotations on the rows of a dense matrix and by using consecutive qubits.}.
Namely, Givens rotations are used to zero out -- one-by-one -- all the non-vanishing entries of the starting vector until this vector has a unique non-vanishing entry.
\begin{algorithm}[h!]
  \begin{algorithmic}
  \Require{Index $1 \leq \ell_\triangle \leq n$ and numerical precision $\mathtt{tol}>0$}.
  \State{Initialize $\sfx^\prime = \sfx$ and  $\calU = \mathbb{I}$.}
  \For{ $i = n, \dots, 1$} \hfill \emph{\teal{$\sharp$ loop backward from the end of $\sfx$.}}
    \If{$|\sfx^\prime_i| \geq \mathtt{tol}$}   

    \State Find the largest $j$ s.t.\ $1\leq j \leq i$ and $|\sfx^\prime_j| \geq \mathtt{tol}$. 
    \If{$j = i$} \hfill \emph{\teal{$\sharp$ pyramidion.}}
    \State Set $\ell_\triangle = i$ and break.
  
    \Else \hfill \teal{\emph{$\sharp$ use $j$ as pivot.}}
    \State Find $\theta$ and $x_\star$ such that \hfill \emph{\teal{$\sharp$ zero out $\sfx^\prime_i$ with nearest $\sfx^\prime_j \neq 0$.}}
    $$     
    \begin{bmatrix}
      x_\star\\
      0
    \end{bmatrix}
     =
     \begin{bmatrix}
      \cos \theta & \sin \theta\\
      -\sin \theta & \cos \theta 
     \end{bmatrix} 
     \begin{bmatrix}
      \sfx^\prime_j\\
      \sfx^\prime_i
    \end{bmatrix}.
    $$
    \State Update $\sfx^\prime_j \leftarrow x_\star$ and $\sfx^\prime_i \leftarrow 0$.
    \State Update $\calU \leftarrow \calU_{ij}(\theta) \circ \calU$ where $\calU_{ij}(\theta)$ is given in \cref{eq:rotation_ij}.
    \EndIf
    \EndIf
    \EndFor
    \Ensure{$\Cliff(\sfx) =  \calU (c_{\ell_\triangle}^* + c_{\ell_\triangle}) \calU^*$.}
  \end{algorithmic}
  \caption{Sparse Clifford loader of $\sfx\in \mathbb{R}^{n\times 1}$ with $\|\sfx\|_2 = 1$.}
  \label{a:sparse}
\end{algorithm}
In \cref{a:sparse}, we describe a procedure where we zero out each non-zero entry of the data vector by using the closest non-zero entry as a \emph{pivot}.
We take $\mathtt{tol} = 10^{-8}$ as the numeric threshold to distinguish non-zero entries.
Note that, in the context of graphs, small circuits built with the sparse architecture contain much fewer gates as it can be seen from the example of \cref{fig:hist_barbell_sherbrooke} (bottom).
This sparse architecture contains $12$ $\mathtt{XX+YY}$ gates and $8$ controlled-$\sigma_z$ gates, whereas the parallel architecture contains $20$ $\mathtt{XX+YY}$ gates, $24$ controlled-$\sigma_z$ gates and $16$ controlled-$\sigma_x$ gates.
\FloatBarrier
\subsection{Simulation results}
We consider two graphs: a complete graph and a barbell graph, for which the acceptance probability of our algorithm is discussed, respectively, in \cref{ex:acceptance_complete_graph} and \cref{ex:acceptance_barbell_graph}.
This determinantal measure over edges is a specific case of the DPP with skewsymmetric $L$-matrix of \cref{thm:pmf_DPP}.
For these simulations, we sample dimer-rooted forests with and without the amplification circuit described in \cref{sec:amplitude_amplification}.
\begin{figure}
  \begin{subfigure}[b]{0.5\textwidth}
    \centering
  \begin{tikzpicture}[scale = 1.5]
      \draw
        (0, 0) node[shape=circle,draw=black] (0){0}
        (0.5, 1) node[shape=circle,draw=black] (1){1}
        (1, 0) node[shape=circle,draw=black] (2){2}
        (0.5, 0.333) node[shape=circle,draw=black] (3){$\varrho$};
      \begin{scope}[-]
        \draw (0) to (1);
        \draw (0) to (2);
        \draw (0) to (3);
        \draw (1) to (2);
        \draw (1) to (3);
        \draw (2) to (3);
      \end{scope}
    \end{tikzpicture}
    \caption{Complete graph with $\varrho=3$. \label{fig:small_complete_plot}}
  \end{subfigure}
  \hfill
  \begin{subfigure}[b]{0.5\textwidth}
    \centering
    \begin{tikzpicture}[scale = 1.5]
      \draw
        (0, 0) node[shape=circle,draw=black] (0){0}
        (0, 1) node[shape=circle,draw=black] (1){1}
        (1, 0.5) node[shape=circle,draw=black] (2){2}
        (2, 0.5) node[shape=circle,draw=black] (3){$\varrho$}
        (3, 0) node[shape=circle,draw=black] (4){4}
        (3, 1) node[shape=circle,draw=black] (5){5};
      \begin{scope}[-]
        \draw (0) to (1);
        \draw (0) to (2);
        \draw (1) to (2);
        \draw (2) to (3);
        \draw (3) to (4);
        \draw (3) to (5);
        \draw (4) to (5);
      \end{scope}
    \end{tikzpicture}
    \caption{Barbell graph with $\varrho=3$.\label{fig:small_barbell_plot}}
 
  \end{subfigure}  
  \begin{subfigure}[b]{0.5\textwidth}
    \centering
    \includegraphics[scale = 0.4]{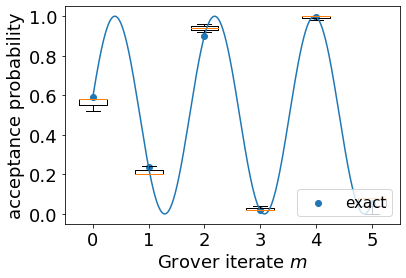}
    \caption{Acceptance \emph{vs} $m$ (complete graph).\label{fig:complete_sin**2}}
  \end{subfigure}
    \hfill
  \begin{subfigure}[b]{0.5\textwidth}
    \centering
    \includegraphics[scale = 0.4]{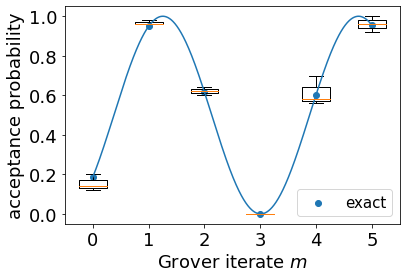}
    \caption{Acceptance \emph{vs} $m$ (barbell graph).\label{fig:barbell_sin**2}}
  \end{subfigure}

  \caption{Effect of amplification on the acceptance probability in \cref{a:rejection_sampling_Clifford_amplified} for sampling determinantal dimer-rooted forests in a complete graph (left-hand side) and a barbell graph (right-hand side).
  The blue dots denote the exact value of the amplified probability. 
  The boxes are obtained by simulation of a logarithmic Clifford loader.
  For each $0\leq m \leq 5$, the acceptance frequency is estimated over $50$ independent DPP samples, and error bars are obtained by repeating this estimation $3$ times.
  The orange curve is the extrapolation of the $\sin^2(x)$ function of \cref{rem:rule_of_thumb}.}
\end{figure}
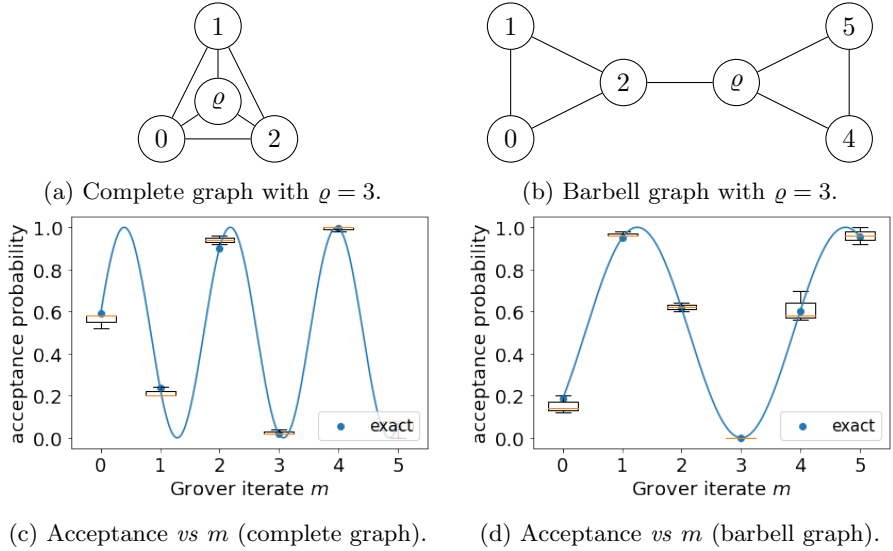
Although \cref{a:rejection_sampling_Clifford_amplified} determines the number of Grover steps $m$, for an illustrative purpose, we also consider the acceptance probability \cref{eq:amplified_acceptance_prob} as a function of $m$.
In \cref{fig:complete_sin**2} and \cref{fig:barbell_sin**2}, the periodic effect of amplitude amplification \cref{eq:Q**m_psi} -- of the form of a $sin^2(x)$ function -- is made manifest: for the complete graph where the initial acceptance probability is already large, and for boosting the acceptance probability in the case of the barbell graph.
As a matter of fact, one Grover iteration ($m=1$) is sufficient to amplify substancially the acceptance probability in the case of the barbell graph (\cref{fig:barbell_sin**2}) whereas the same number of iterations reduces this probability in the case of the complete graph (\cref{fig:complete_sin**2}, in which case \cref{a:rejection_sampling_Clifford_amplified} simply chooses $m=0$).
Adding several Grover iterations yields to an oscillatory behaviour where amplification as well as reduction can occur.

\begin{figure}
  \centering
  \begin{overpic}[width=\textwidth,scale = 0.055]{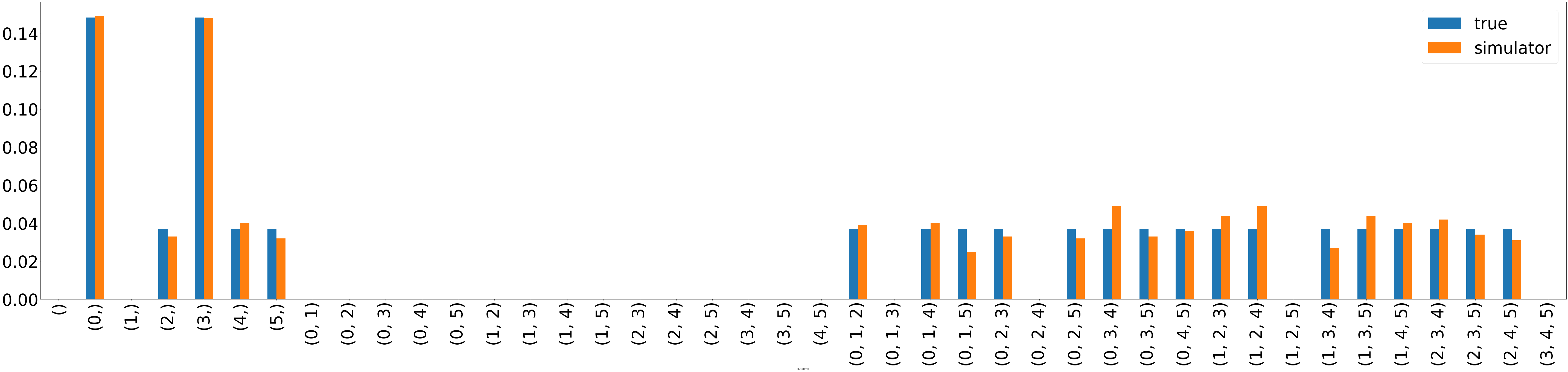}
    \put (67,14){\footnotesize spanning trees}
    \put( 54,  12) {\line(1,0){ 43}}
  \end{overpic}
  \caption{Simulated subset frequencies (orange) with a parallel Clifford loader \emph{vs} exact subset probabilities (blue) in the case of the complete graph of $4$ nodes depicted in \cref{fig:small_complete_plot}. 
  Subset frequencies are estimated over $1000$ independent DPP samples.
  At the right of the histogram, we observe the equal probability mass given to the spanning trees of the complete graph which are subsets of $3$ edges in this case. \label{fig:hist_complete}}
\end{figure}

Also, we directly report an estimate of the probability of any subset of edges.
Edges are enumerated as follows. 
An orientation is fixed by writing an edge as $(i,j)$ with $i<j$, and all the edges are then ordered lexicographically; for example, in \cref{fig:small_complete_plot}, $(0,1)$ is edge $\mathtt{(0,)}$, next $(0,2)$ is edge $\mathtt{(1,)}$, etc.
\Cref{fig:hist_complete} and \cref{fig:hist_barbell} provide histograms comparing empirical and exact subset probabilities.
These plots illustrate the fact that the determinantal measure only gives mass to dimer-rooted forests of odd cardinal since these two graphs have an even number of nodes; see \cref{def:dimer-rooted_forest}.
\begin{figure}
  \centering
\begin{subfigure}[b]{\textwidth}
  \centering
  \includegraphics[width=\textwidth,scale = 0.055]{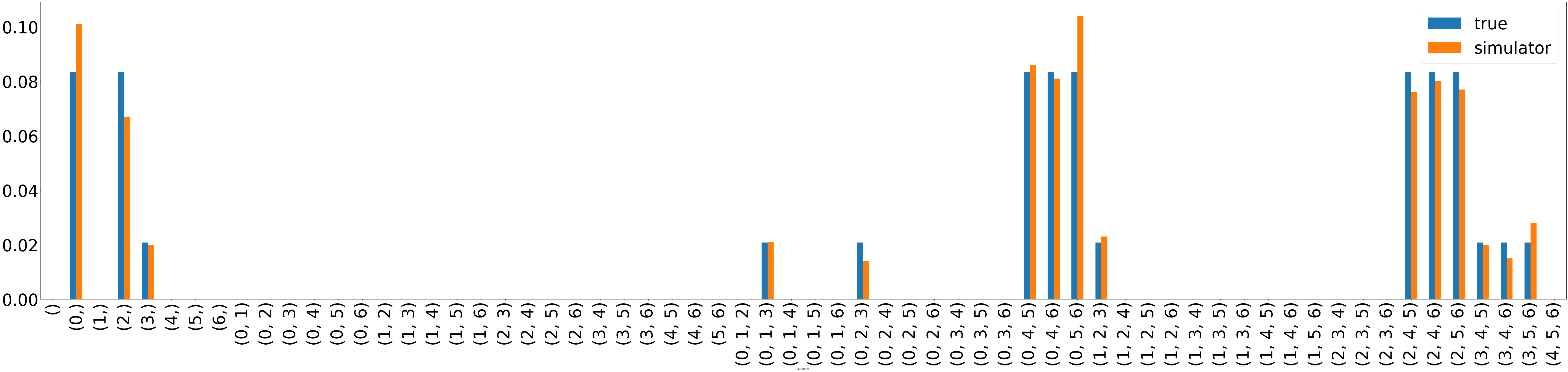}
\end{subfigure}
\hfill
\begin{subfigure}[b]{\textwidth}
  \centering
  \begin{overpic}[width=\textwidth,scale = 0.055]{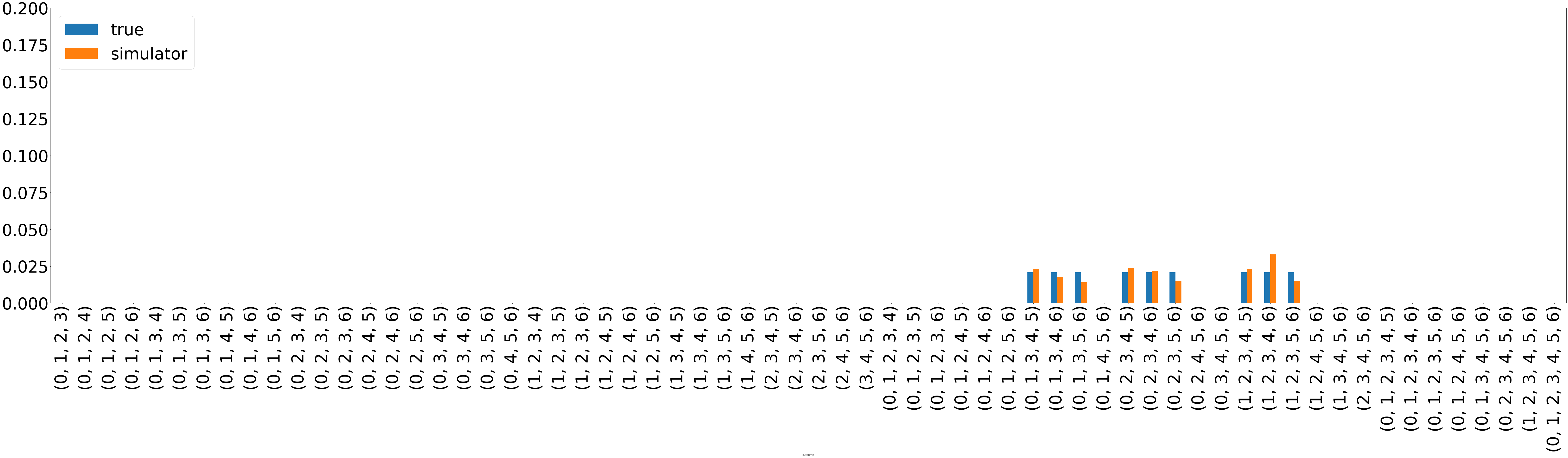}
    \put (67,15){\footnotesize spanning trees}
    \put( 65,  14) {\line(1,0){ 18}}

  \end{overpic}
\end{subfigure}
\caption{Simulated subset frequencies (orange) with the parallel Clifford loader \emph{vs} exact subset probabilities (blue) in the case of the barbell graph of $6$ nodes given in \cref{fig:small_barbell_plot}.
Subset frequencies are estimated over $1000$ independent DPP samples.
At the right of the second row, we observe the equal probability mass given to the spanning trees which are subsets of $5$ edges in this case. 
\label{fig:hist_barbell}}
\end{figure}
\begin{remark}[Example of dimer-rooted forests with zero probability]\label{rem:node_ordering_is_important}
  By inspecting for instance \cref{fig:hist_barbell}, we observe that the edge denoted by $\mathtt{(1,)}$ on the $x$-axis -- that is the single edge corresponding to the node pair $(0,2)$ in \cref{fig:small_barbell_plot} -- has zero probability to be sampled.
  Nonetheless, the single edge $(0,2)$ is a dimer-rooted forest with root $\varrho = 3$.
  This lack of symmetry is a consequence of the orientation induced by the node ordering as we now explain.
  Indeed, as mentioned in the proof of \cref{thm:contracted_trees}, $C=\skew(L_{\widehat{\varrho}\widehat{\varrho}})$ can be interpreted as a skewsymmetric adjacency matrix.
  Thus, for all edge $(i,j)$ with $i<j$, we have $C_{ij}=1$ and therefore the edge is oriented from $i$ to $j$, as illustrated in \cref{fig:oriented_cycle_skew(L)}.
  Now, by using \cref{thm:contracted_trees},  the probability that $(0,2)$ is sampled equals
  $$
  \det
  \kbordermatrix{
    & (0,2) & 0 & 1 & 2\\
    (0,2) & 0 & 1 & 0 & -1\\
    0 & -1 & 0 & 1 & 1\\
    1 & 0 & -1& 0 & 1 \\
    2 & 1 & -1 & -1 & 0
  }
  =0.
  $$
  Considering the augmented graph of \cref{fig:oriented_cycle_augmented}, the above determinant vanishes since the row of node $(0,2)$ is collinear with the row of node $1$.
  \begin{figure}[h]
  \begin{subfigure}[b]{0.5\textwidth}
    \centering
    \begin{tikzpicture}[scale = 1.5]
      \draw
        (0, 0) node[shape=circle,draw=black] (0){0}
        (0, 1) node[shape=circle,draw=black] (1){1}
        (1, 0.5) node[shape=circle,draw=black] (2){2}
        (2, 0.5) node[shape=circle,draw=black] (3){$\varrho$}
        (3, 0) node[shape=circle,draw=black] (4){4}
        (3, 1) node[shape=circle,draw=black] (5){5};
      \begin{scope}[-]
        \draw[->,thick] (0) to (1);
        \draw[->,thick] (0) to (2);
        \draw[->,thick] (1) to (2);
        \draw[->,thick] (4) to (5);
      \end{scope}
      \draw[dotted] (2) to (3);
      \draw[dotted] (3) to (4);
      \draw[dotted] (3) to (5);
    \end{tikzpicture}
    \caption{ $ \skew(L_{\widehat{\varrho}\widehat{\varrho}})$. \label{fig:oriented_cycle_skew(L)}}

  \end{subfigure}
  \hfill
  \begin{subfigure}[b]{0.5\textwidth}
    \centering
    \begin{tikzpicture}[scale = 1.5]
      \draw
        (0, 0) node[shape=circle,draw=black] (0){0}
        (0, 1) node[shape=circle,draw=black] (1){1}
        (1, 0.5) node[shape=circle,draw=black] (2){2}
        (2, 0.5) node[shape=circle,draw=black] (3){$\varrho$}
        (3, 0) node[shape=circle,draw=black] (4){4}
        (3, 1) node[shape=circle,draw=black] (5){5};
      \node[shape=rectangle,draw=teal, ultra thick] (B) at (0.7,0) {\tiny $(0,2)$};
      \draw[->,thick] (B) to (0);
      \draw[->,thick] (2) to (B);

      \begin{scope}[-]
        \draw[->,thick] (0) to (1);
        \draw[->,thick] (0) to (2);
        \draw[->,thick] (1) to (2);
        \draw[->,thick] (4) to (5);
      \end{scope}
      \draw[dotted] (2) to (3);
      \draw[dotted] (3) to (4);
      \draw[dotted] (3) to (5);
    \end{tikzpicture}
    \caption{
    $\left(\begin{smallmatrix}
        0 & B_{\calC \widehat{\varrho}}\\
        -(B_{\calC \widehat{\varrho}})^\top & \skew(L_{\widehat{\varrho}\widehat{\varrho}})
    \end{smallmatrix}\right)
    $ with $\calC=\{(0,2)\}$. \label{fig:oriented_cycle_augmented}}

  \end{subfigure}
  \caption{Oriented graphs associated with two skewsymmetric adjacency matrices obtained in the case of the barbell graph of \cref{fig:small_barbell_plot}; see \cref{rem:node_ordering_is_important}. \label{fig:oriented_cycle}}
\end{figure}
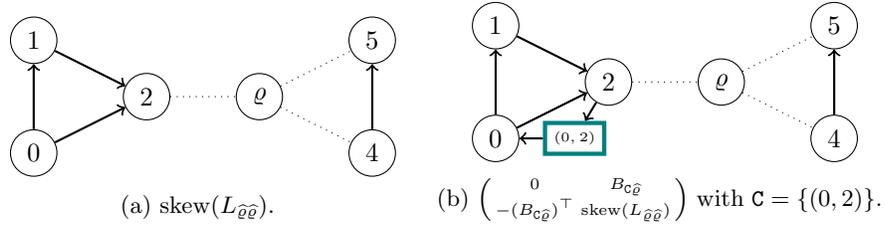
\end{remark}
\subsection{Results on Quantum Computing Units}
To illustrate the effect of noise, we execute the experiment of \cref{fig:hist_barbell} (still without amplification) on one of the freely available Quantum Computing Units (QPUs) of \texttt{IBM Quantum Plateform} with $127$ qubits, namely  $\mathtt{ibm\_sherbrooke}$, by using this time a sparse Clifford loader.

Concerning the timing, $8192$ samples were obtained in about $4s$ excluding transpilation time, according to \texttt{IBM Quantum Plateform} dashboard. 
The frequencies of sampled subsets are displayed in \cref{fig:hist_barbell_sherbrooke}. 
In the first two rows of \cref{fig:hist_barbell_sherbrooke}, the subsets with large probabilities are indeed sampled although other subsets with vanishing probability also appear due to the presence of noise. 
Note that we used the highest possible level of optimization of the transpiler, namely $\mathtt{optimization\_level}=3$, in order to reduce the errors as much as possible.
The results without optimization ($\mathtt{optimization\_level}=0$) are more noisy as it can be seen in \cref{fig:hist_barbell_sherbrooke_no_optim}.
\begin{figure}
  \centering
\begin{subfigure}[b]{\textwidth}
  \centering
  \includegraphics[width=\textwidth,scale = 0.055]{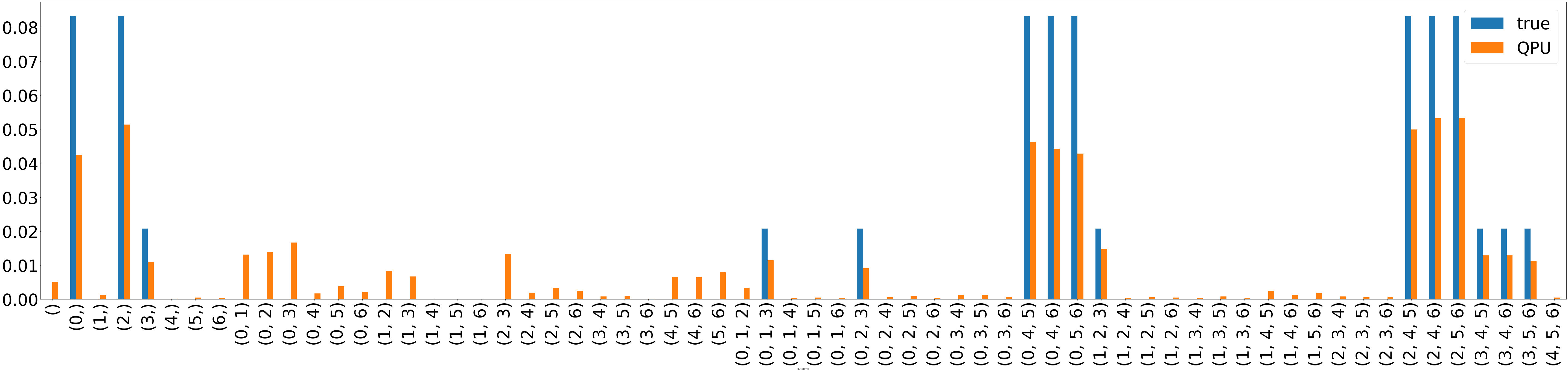}
\end{subfigure}
\hfill
\begin{subfigure}[b]{\textwidth}
  \centering
  \includegraphics[width=\textwidth,scale = 0.055]{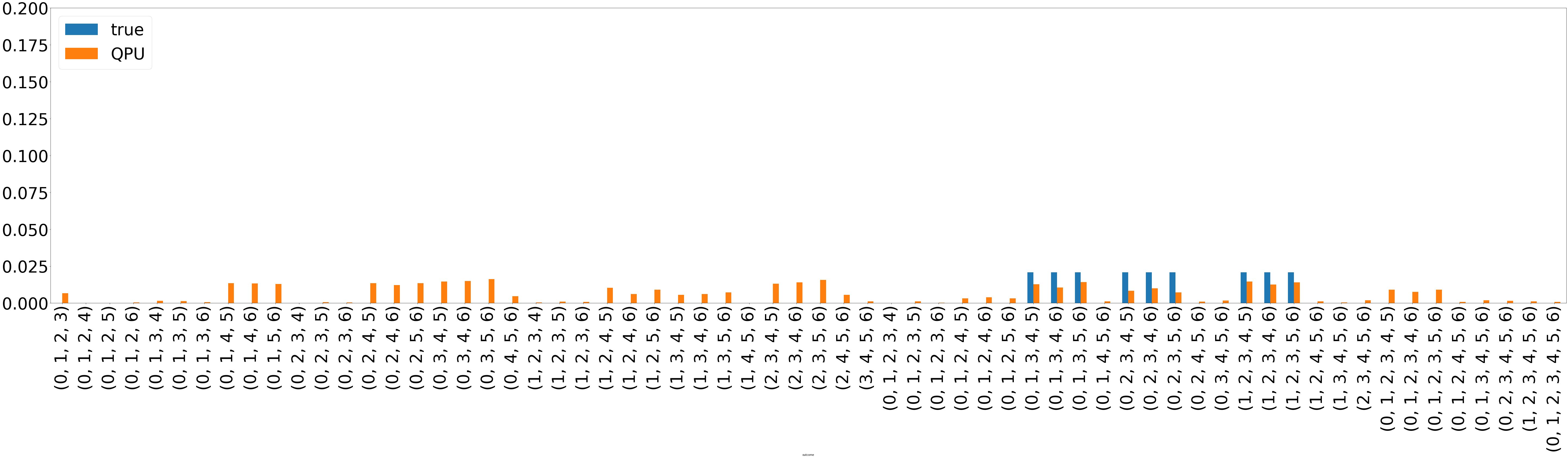}
\end{subfigure}
\begin{subfigure}[b]{\textwidth}
  \centering
  \includegraphics[width=\textwidth,scale = 0.055]{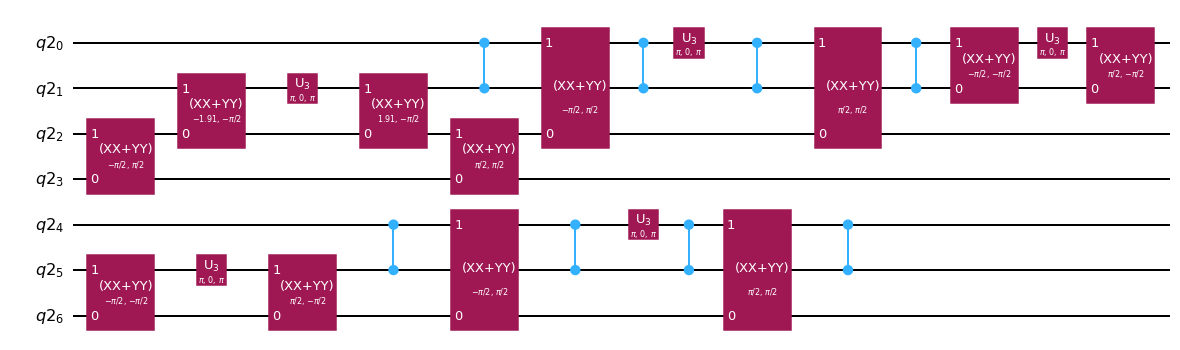}
\end{subfigure}
\caption{Row $1$ and Row $2$: observed subset frequencies on $\mathtt{ibm\_sherbrooke}$ (orange) with the \emph{sparse} Clifford loader \emph{vs} exact subset probabilities (blue) in the case of the barbell graph of \cref{fig:small_barbell_plot}.
Row 3: the corresponding \emph{sparse} Clifford loader with FBS gates.
Subset frequencies are estimated over $8192$ independent DPP samples. Qiskit transpiler is used with the option $\mathtt{optimization\_level}=3$.
\label{fig:hist_barbell_sherbrooke}}
\end{figure}

In general, we found that the parallel and sparse architectures -- which use FBS gates -- are more error prone than the pyramid architecture.
The reason might be that the latter only use RBS gates on neighbouring qubits, while the former use additional controlled-$\sigma_z$ and controlled-$\sigma_x$ gates; see \cref{fig:FBS_gate} for an example.

\begin{figure}
  \centering
\begin{subfigure}[b]{\textwidth}
  \centering
  \includegraphics[width=\textwidth,scale = 0.055]{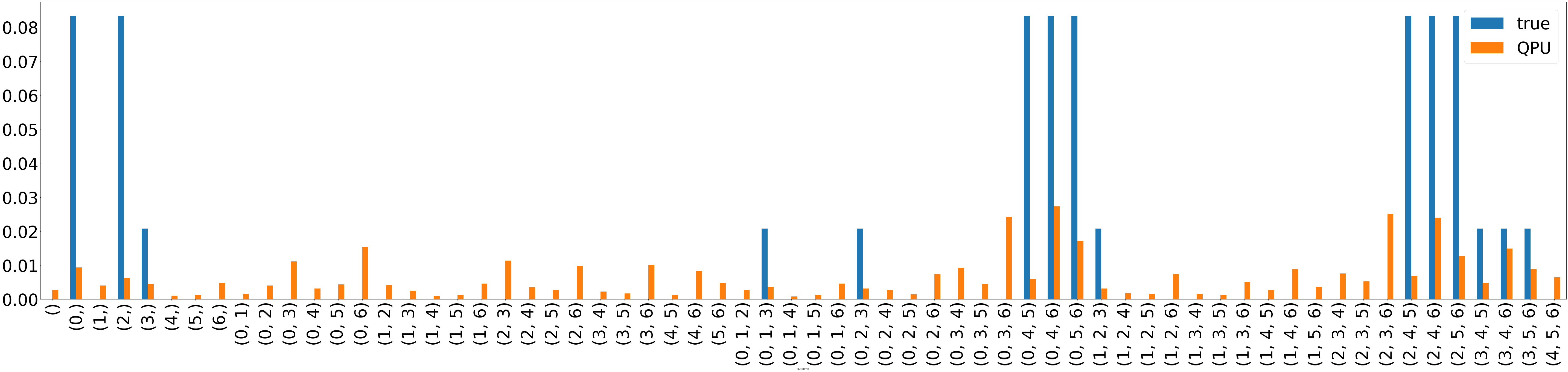}
\end{subfigure}
\hfill
\begin{subfigure}[b]{\textwidth}
  \centering
  \includegraphics[width=\textwidth,scale = 0.055]{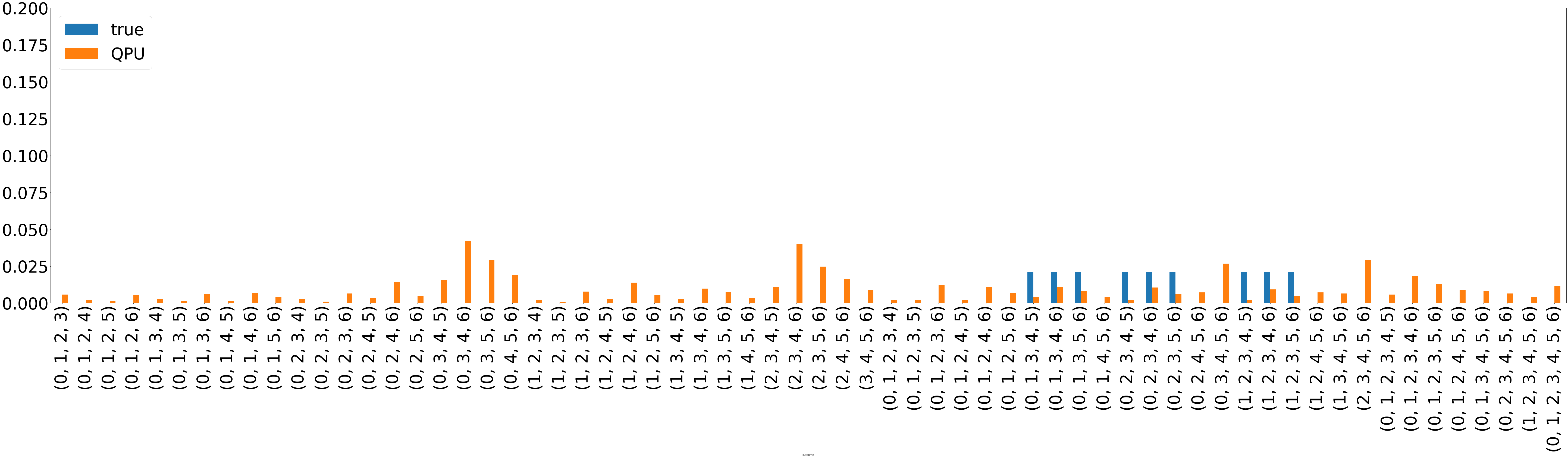}
\end{subfigure}

\caption{Row $1$ and Row $2$: observed subset frequencies on $\mathtt{ibm\_sherbrooke}$ (orange) with the \emph{sparse} Clifford loader \emph{vs} exact subset probabilities (blue) in the case of the barbell graph of \cref{fig:small_barbell_plot}. In contrast with \cref{fig:hist_barbell_sherbrooke}, the Qiskit transpiler uses the lowest optimization level, i.e.,  $\mathtt{optimization\_level}=0$.
\label{fig:hist_barbell_sherbrooke_no_optim}}
\end{figure}

\FloatBarrier
\section{Conclusion}
We have presented an algorithm that samples a DPP when the kernel is a projection onto a space given by a non-orthonormal spanning set of vectors, namely the columns of the $n\times r$ matrix $X$.
The computational bottleneck of our algorithm is a classical preprocessing step that approximates $\det \sfX^\top \sfX$.
Whenever this step is $o(nr^2)$, we have a faster DPP sampling algorithm than classical counterparts. 

Our algorithm takes the form of a classical preprocessing through a normalization step and an estimation of the necessary number of Grover steps, followed by a bounded number of repeated runs of a quantum circuit until an easily-checked cardinality constraint is met. 
In passing, we have characterized the distribution of the output of our algorithm when only one run of the circuit is performed, without checking the cardinality constraint. This side contribution sheds light on the robustness of a previous algorithm by \cite{kerenidis2022quantum} and brings a new DPP forward, by identifying a DPP kernel of a type unknown in previous applications of DPPs, yet that is natural in the quantum formalization of DPPs.
We have illustrated our algorithm on small-scale graph examples, which also shows the potential of our algorithm is numerical applications of uniform spanning trees and related distributions.
This connects our approach to recent work by \cite{apers2022quantum}, who propose a quantum algorithm for graph sparsification that assumes a QRAM memory and does not explicitly sample uniform spanning trees.

There are many potential extensions of our algorithm. 
For instance, one could design DPP-specific error-correcting codes, using known statistical properties of DPPs to detect e.g. readout errors.
We could also adapt our circuit to the particular machine on which it is run, by taking into account the estimated accuracies of one- and two-qubit gates provided by the constructor.
This is partially done automatically when running circuits e.g. on IBM machines, a step known as \emph{transpilation}, but the procedure could be taylored to our particular task.
More formally, it would also be interesting to define Clifford loaders in order to sample quaternionic DPPs \citep{kassel2022determinantal} or the recently introduced multideterminantal DPPs \citep{kenyon2025multideterminantal}.

Other interesting avenues for future work are the completion of our quantum circuit to treat a downstream task such as linear regression with column subset selection \emph{ directly on the quantum computer}, instead of outputting a DPP sample. Moreover, it would be interesting to see if other QR-related numerical algebraic tasks than DPP sampling can benefit from the combination of Clifford loaders and amplitude rejection.
\section*{Acknowledgements}

This work was supported by the ERC grant BLACKJACK (ERC-2019-STG-851866)
and the ANR AI chair BACCARAT (ANR-20-CHIA-0002).
We acknowledge the use of IBM Quantum services for this work. The views expressed are those of the authors, and do not reflect the official policy or position of IBM or the IBM Quantum team.
\appendix
\section{Clifford loaders\label{app:Clifford}}
We revisit here the construction of Clifford loaders by \citep{johri2021nearest} and  \citet*{kerenidis2022quantum}  with the operator algebra approach of \citet*{bardenet2024sampling}; see \cref{sec:Parallel_Pyramid_Loaders}.

%
\subsection{Parallel and pyramid architectures\label{sec:Parallel_Pyramid_Loaders}}

In the next subsection, we highlight a connection -- which is implicit in \citet*{kerenidis2022quantum} -- between the Clifford loaders and spherical coordinate systems.
\subsubsection{Spherical coordinate systems}

The basic idea to construct $\calU(\sfx)$ in \cref{eq:loader} is to use spherical coordinates of $\sfx$.
We can consider the case where $n$ is a power of $2$ since any vector can be increased in length by adding zeros so that its length is $2^k$ for some $k$.
Hyperspherical coordinate systems of $\sfx$ are associated with binary trees, as it is explained in \citep[Chapter 6]{nikiforov1991classical}.
Following \citep{johri2021nearest}, we are interested in circuits with short depth which are associated with binary trees of short depth where parallel operations are executed.
We give an example of this parallelism for $n=4$ in \cref{fig:log_tree}, so that the generalization to higher $n$ becomes intuitive.
\begin{figure}[ht!]
  \centering
    \begin{tikzpicture}{scale=0.7}
        \node[shape=circle,draw=black] (A) at (0,0) {$r_1$};
        \node at (0,0+0.7) {$\theta_1$};

        \node[shape=circle,draw=black] (B) at (-1,1) {$r_2$};
        \node at (-1,1+0.7) {$\theta_2$};

        \node[shape=circle,draw=black] (C) at (1,1) {$r_3$};
        \node at (1,1+0.7) {$\theta_3$};

        \node[shape=circle,draw=black] (D) at (-1.5,2.5) {$\sfx_1$};
        \node[shape=circle,draw=black] (E) at (-0.5,2.5) {$\sfx_2$};
        \node[shape=circle,draw=black] (F) at (0.5,2.5) {$\sfx_3$};
        \node[shape=circle,draw=black] (G) at (1.5,2.5) {$\sfx_4$};

        \node at (-1.5-0.8,2.5+0.6+0.05) {$\sfx^\top = $};

        \node at (-1.5-0.3,2.5+0.6) {$[$};
        \node at (-1.5,2.5+0.6) {$\sfx_1$};
        \node at (-0.5,2.5+0.6) {$\sfx_2$};
        \node at (0.5,2.5+0.6) {$\sfx_3$};
        \node at (1.5,2.5+0.6) {$\sfx_4$};
        \node at (1.5+0.3,2.5+0.6) {$]$};

        \path [-] (A) edge (B);
        \path [-] (A) edge (C);
        \path [-] (B) edge (D);
        \path [-] (B) edge (E);
        \path [-] (C) edge (F);
        \path [-] (C) edge (G);

        \node at (4, 3){coordinates};

        \node at (4, 2.5){$\mathsf{x_1}= \cos \theta_1 \cos \theta_2$};
        \node at (4, 2){$\mathsf{x_2}= \cos \theta_1 \sin \theta_2$};  \node at (4, 1.5){$\mathsf{x_3}=\sin \theta_1 \cos \theta_3$};
        \node at (4, 1){$\mathsf{x_4}=\sin\theta_1 \sin \theta_3$};

        \node at (8, 2.5){$r_2= \sqrt{\mathsf{x_1}^2 + \mathsf{x_2}^2}$};
        \node at (8, 2.){$r_3= \sqrt{\mathsf{x_3}^2 + \mathsf{x_4}^2}$};
        \node at (8, 1.5){$r_1= \sqrt{r_1^2 + r_2^2}$};

        \node at (8, 3){parameters};

        \node at (8, 1.-0.2){$\cos \theta_2 = \mathsf{x_1}/r_2$};
        \node at (8, .5-0.2){$\cos \theta_3 = \mathsf{x_3}/r_3$};
        \node at (8, -0.2){$\cos \theta_1 = r_2/r_1$};
  \end{tikzpicture}
  \caption{Hyperspherical coordinates of a unit norm $\sfx\in \mathbb{R}^4$ (parallel shape).
  For simplicity, we assume that none of the components of $\sfx$ vanish.
  The parameters $r_3,r_2,r_1$ are obtained by going along the tree from the leafs to the root. 
  The angles $\theta_1,\theta_2, \theta_3$ are computed by going from the root to the leafs and using $r_1,r_2,r_3$.
  \label{fig:log_tree}}
\end{figure}
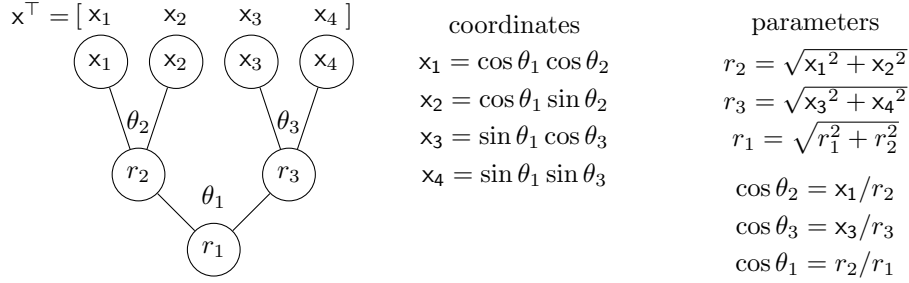
To compute the spherical coordinates of $\sfx$, we go in the tree from the leafs to the roots to evaluate the $r_i$s.
Next, the angles $\theta_i$s are obtained by visiting the nodes from the root to the leaves.
We see that the computation of all the coordinates of $\sfx$ requires $\mathcal{O}(n)$ operations.
Note that different binary trees correspond to different spherical coordinate systems.
For example, for $\sfx\in \mathbb{R}^8$, 
the tree of largest depth (i.e., least parallel) is given in \cref{fig:least_depth}.

\begin{figure}[ht!]
  \resizebox{.8\textwidth}{!}
  {
  \begin{tikzpicture}{scale = 0.5}

    \node[scale=0.75] at (-1.5-0.3-0.5,2.5 +0.5+0.05) {$\sfx^\top = $};

    \node[scale=0.75] at (-1.5-0.3,2.5 +0.5) {$[$};

    \node[scale=0.75] at (-1.5,2.5 +0.5) {$\sfx_1$};
    \node[scale=0.75]  at (-0.5,2.5 +0.5) {$\sfx_2$};
    \node[scale=0.75]  at (0.5,2.5 +0.5) {$\sfx_3$};
    \node[scale=0.75]  at (1.5,2.5 +0.5) {$\sfx_4$};
    \node[scale=0.75]  at (-1.5+4,2.5 +0.5) {$\sfx_5$};
    \node[scale=0.75]  at (-0.5+4,2.5 +0.5) {$\sfx_6$};
    \node[scale=0.75]  at (0.5+4,2.5 +0.5) {$\sfx_7$};
    \node[scale=0.75]  at (1.5+4,2.5 +0.5) {$\sfx_8$};

    \node[scale=0.75] at (1.5+4+0.3,2.5 +0.5) {$]$};

    \node[shape=circle,draw=black,scale=0.75] (D1) at (-1.5,2.5) {$\sfx_1$};
    \node[shape=circle,draw=black,scale=0.75] (E1) at (-0.5,2.5) {$\sfx_2$};
    \node[shape=circle,draw=black,scale=0.75] (F1) at (0.5,2.5) {$\sfx_3$};
    \node[shape=circle,draw=black,scale=0.75] (G1) at (1.5,2.5) {$\sfx_4$};
    \node[shape=circle,draw=black,scale=0.75] (D2) at (-1.5+4,2.5) {$\sfx_5$};
    \node[shape=circle,draw=black,scale=0.75] (E2) at (-0.5+4,2.5) {$\sfx_6$};
    \node[shape=circle,draw=black,scale=0.75] (F2) at (0.5+4,2.5) {$\sfx_7$};
    \node[shape=circle,draw=black,scale=0.75] (G2) at (1.5+4,2.5) {$\sfx_8$};

    \node[shape=circle,draw=black,scale=0.75] (r1) at (-1.5 + 0.125*3.5,2.5 - 0.125*4.5) {$r_7$};
    \node[shape=circle,draw=black,scale=0.75] (r2) at (-1.5 + 0.25*3.5,2.5 - 0.25*4.5) {$r_6$};
    \node[shape=circle,draw=black,scale=0.75] (r3) at (-1.5 + 0.375*3.5,2.5 - 0.375*4.5) {$r_5$};
    \node[shape=circle,draw=black,scale=0.75] (r4) at (-1.5 + 0.5*3.5,2.5 - 0.5*4.5) {$r_4$};
    \node[shape=circle,draw=black,scale=0.75] (r5) at (-1.5 + 0.625*3.5,2.5 - 0.625*4.5) {$r_3$};
    \node[shape=circle,draw=black,scale=0.75] (r6) at (-1.5 + 0.75*3.5,2.5 - 0.75*4.5) {$r_2$};
    \node[shape=circle,draw=black,scale=0.75] (r7) at (-1.5 + 0.875*3.5,2.5 - 0.875*4.5) {$r_1$};

    \path [-] (D1) edge (r1);

    \path [-] (r7) edge (r6);
    \path [-] (r5) edge (r6);
    \path [-] (r5) edge (r4);
    \path [-] (r3) edge (r4);
    \path [-] (r3) edge (r2);
    \path [-] (r1) edge (r2);

    \path [-] (E1) edge (r1);
    \path [-] (F1) edge (r2);
    \path [-] (G1) edge (r3);
    \path [-] (D2) edge (r4);
    \path [-] (E2) edge (r5);
    \path [-] (F2) edge (r6);
    \path [-] (G2) edge (r7);

    \node at (9-0.5,2.5){coordinates};
    \node at (9-0.5, 2.5 -0.5){$\mathsf{x_1}= \cos\theta_1\dots \cos\theta_7$};
    \node at (9-0.5, 2-0.5){$\vdots$};  
    \node at (9-0.5, -0.){$\vdots$};
    \node at (9-0.5, -0.5){$\mathsf{x_6}=\cos\theta_1\cos\theta_2\sin\theta_3$};
    \node at (9-0.5, -1){$\mathsf{x_7}=\cos\theta_1\sin\theta_2$};
    \node at (9-0.5, -1.5){$\mathsf{x_8}=\sin\theta_1$};
\end{tikzpicture}
}
  \caption{Hyperspherical coordinates of a unit vector $\sfx\in \mathbb{R}^8$ with pyramid shape.
  Again, we assume that none of the components of $\sfx$ vanish.
  \label{fig:least_depth}}
\end{figure}
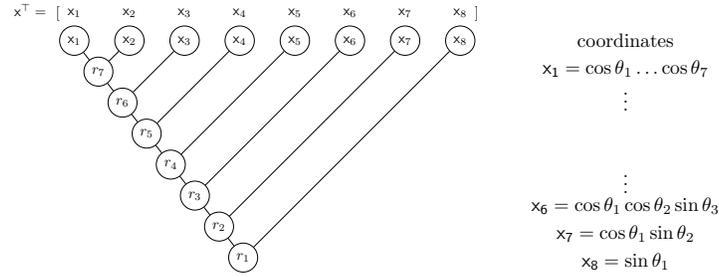

\subsubsection{RBS gates \label{sec:RBS}}
We now give the intuition behind the use of the Givens operator to load a unit vector in $\mathbb{R}^4$ with parallel rotations as in the case of the tree of \cref{fig:log_tree}.
Consider the vector 
$
\mathsf{y} = 
[
  \cos \theta,
   \sin \theta,
   0,
   0
]^\top.
$
Actually, there exists a unitary operator $\calU(\mathsf{y})$ -- called here Givens operator -- such that $\calU\ket{\emptyset}= \ket{\emptyset}$ and
\begin{align}
  \begin{bmatrix}
    \calU \ c_{1}^* \ \calU^*\\
    \calU\ c_{2}^* \ \calU^*
  \end{bmatrix}
   =
   \begin{bmatrix}
    \cos \theta & \sin \theta\\
    -\sin \theta & \cos \theta 
   \end{bmatrix} 
   \begin{bmatrix}
    c_{1}^*\\
    c_{2}^*
  \end{bmatrix} \text{ and } \calU c_{i}^*\calU^* = c_{i}^*, \quad \forall i \notin\{1,2\}; \label{eq:rotation}
\end{align}
as defined e.g.\ in \citep[Section 5.2.1]{bardenet2024sampling}.
We can load $\mathsf{y}$ by acting on the Fock vacuum as follows:
$$
  \Cliff(\mathsf{y}) \ket{\emptyset} = \cos \theta \ c_{1}^*\ket{\emptyset}+ \sin \theta \ c_{2}^*\ket{\emptyset}.
$$
In the computational basis, this reads
$
  \Cliff(\mathsf{y}) \ket{0000} =  \cos \theta \ket{1000}+ \sin \theta \ket{0100}.
$
The following result --  whose proof is elementary -- is used in order to implement Clifford loaders whenever qubits are consecutive.
\begin{lemma}[RBS Gate]\label{lem:RBS}
  Let $\calU$ be the operator defined in \cref{eq:rotation}.
  Consider the vector space generated by the following orthonormal basis: 
  $$\left(\ket{\emptyset}= \ket{0000}, c_1^* \ket{\emptyset} = \ket{1000}, c_2^* \ket{\emptyset} = \ket{0100}, c_2^*c_1^*\ket{\emptyset} = \ket{1100}\right).$$
  Denote this basis by $\mathcal{B}$.
  The representation of $\calU$ restricted to this subspace is
  \begin{equation}
  \rho(\calU)_{\mathcal{B}\mathcal{B}}
  = 
  \kbordermatrix{
    &\ket{0000} & \ket{1000} & \ket{0100} & \ket{1100}\\
    \ket{0000}&1 & 0 & 0 & 0\\
    \ket{1000}&0 & \cos\theta & \sin \theta & 0\\
    \ket{0100}&0 & -\sin \theta & \cos \theta & 0\\
    \ket{1100}&0 & 0 & 0 & 1
  },\label{eq:RBS}
  \end{equation}
  whereas $\rho(\calU)$ is the identity in the subspace generated by $\ket{0010}$ and $\ket{0001}$.
  The matrix \cref{eq:RBS} is called in \citep*{kerenidis2022quantum} a Reconfigurable Beam Splitter (RBS) gate.
\end{lemma}

Roughly speaking, a RBS represents a Givens operator in the one-particle subspace.
Considering now a space of $n$ qubits, we denote the operator acting in the subspace of the $i$th and $j$th qubits as \cref{eq:RBS} by $\mathrm{RBS}_{ij}(\theta)$ with $1\leq i < j\leq n$.
This gate can be realized in Qiskit by a $\mathtt{XX+YY}$ gate.
\subsubsection{FBS gates \label{sec:FBS}}

In \cref{lem:FBS} below, we formalize a result by \citet{kerenidis2022quantum} which describes the Givens operator in a subspace with more than one particle.
%
%
\begin{lemma}[FBS Gate]\label{lem:FBS}
  Consider the setting of \cref{lem:RBS}.
  Let $1\leq i < j \leq n$ and denote by $\calU_{ij}(\theta)$  the operator defined by $\calU_{ij}(\theta)\ket{\emptyset} = \ket{\emptyset}$, 
  \begin{align}
    \begin{bmatrix}
      \calU_{ij}(\theta) \ c_{i}^* \ \calU_{ij}(\theta)^*\\
      \calU_{ij}(\theta)\ c_{j}^* \ \calU_{ij}(\theta)^*
    \end{bmatrix}
     =
     \begin{bmatrix}
      \cos \theta & \sin \theta\\
      -\sin \theta & \cos \theta 
     \end{bmatrix} 
     \begin{bmatrix}
      c_{i}^*\\
      c_{j}^*
    \end{bmatrix} \text{ and } \calU_{ij}(\theta) c_{k}^*\calU_{ij}(\theta)^* = c_{k}^*, \ \forall k \notin\{i,j\}. \label{eq:rotation_ij}
  \end{align}
  For any ordered subset $\calC$, we denote by $\ket{\calC}$ the state of the computational basis representing $\prod_{k\in \calC} c_k^* \ket{\emptyset}$ in the Jordan-Wigner representation.
  Let $\calS$ be an ordered subset of $\{1,\dots, n\}$ such that  $i,j\in \calS$.
  Denote by $\mathcal{B}$ the orthonormal basis 
  $$\left(\ket{\calS}, \ket{\calS\setminus i}, \ket{\calS \setminus j},\ket{\calS\setminus \{i,j\}}\right).$$
  The representation of $\calU$ restricted to this subspace is
  \begin{equation}
  \rho(\calU)_{\mathcal{B}\mathcal{B}}
  = 
  \kbordermatrix{
    &\ket{\calS} & \ket{\calS\setminus i} & \ket{\calS\setminus j} & \ket{\calS\setminus \{i,j\}}\\
    \ket{\calS} &1 & 0 & 0 & 0\\
    \ket{\calS\setminus i} &0 & \cos\theta & (-1)^{f(i,j,\calS)}\sin \theta & 0\\
    \ket{\calS\setminus j} &0 & (-1)^{1 +f(i,j,\calS)}\sin \theta & \cos \theta & 0\\
    \ket{\calS\setminus \{i,j\}}&0 & 0 & 0 & 1
  },\label{eq:FBS}
  \end{equation}
  where $f(i,j,\calS)= \sum_{i< k< j : k\in \calS} 1$. Furthermore, $\rho(\calU_{ij}(\theta)) \ket{\calS^\prime} = \ket{\calS^\prime}$ for all $\calS^\prime$ such that $i,j\notin \calS^\prime$.
  We denote the operator acting in the subspace of the $i$th and $j$th qubits as \cref{eq:FBS} by $\mathrm{FBS}_{ij}(\theta)$ with $1\leq i < j\leq n$.
\end{lemma}
The above matrix is called in \citep*{kerenidis2022quantum} a Fermionic Beam Splitter (FBS) gate.
The only difference with \cref{lem:RBS} is the possible presence of an extra minus sign which is due to the (fermionic) anticommutation relations and counts the parity of the number of particles between $i$ and $j$.
In our understanding, this is why this gate is called fermionic.

For completeness, we give a sketch of proof of \cref{lem:FBS}.
\begin{proof}[Sketch of proof of \cref{lem:FBS}]
  We compute $\calU \ket{\calS\setminus i}$.
  To begin, we observe that 
  \[
    \ket{\calS\setminus i} = \left(\prod_{k\in \calS : k<i} c^*_k\right) \left(\prod_{k\in \calS : i<k<j} c^*_k\right)  c_j^*  \left(\prod_{k\in \calS : k>j}c_k^*\right)\ket{\emptyset}.
  \]
  Now, by applying $\calU$, we find 
  \begin{align*}
      \calU\ket{\calS\setminus i} &= \left(\prod_{k\in \calS : k<i} c^*_k\right) \left(\prod_{k\in \calS : i<k<j} c^*_k\right) \calU c_j^* \calU^* \calU  \left(\prod_{k\in \calS : k>j}c_k^*\right)\ket{\emptyset}\\
      &= \left(\prod_{k\in \calS : k<i} c^*_k\right) \left(\prod_{k\in \calS : i<k<j} c^*_k\right) \left(-\sin \theta c_i^* + \cos\theta c_j^*\right)  \left(\prod_{k\in \calS : k>j}c_k^*\right)\ket{\emptyset}.
  \end{align*}
  Finally, the anticommutation relations on the first term on the right-hand side yields
  \begin{align*} 
    -\sin \theta\left(\prod_{k\in \calS : k<i} c^*_k\right) (-1)^{f(i,j,\calS)}c_i^* \left(\prod_{k\in \calS : i<k<j} c^*_k\right)  \left(\prod_{k\in \calS : k>j}c_k^*\right)\ket{\emptyset}.
  \end{align*}
  Thus, the desired result is
  \[
    \calU\ket{\calS\setminus i} = -(-1)^{f(i,j,\calS)} \sin \theta \ket{\calS\setminus j} + \cos \theta \ket{\calS\setminus i},
  \]
  where $f(i,j,\calS)$ is the cardinal of $\{ k\in \calS : i<k<j\}$.
\end{proof}
Next,  in order to make this paper more self-contained, we slightly rephrase \citep[Proposition 2.6]{kerenidis2022quantum} which gives a decomposition of FBS gates.

An FBS gate between qubits $i$ and $j$ can be realized by combining RBS gates with gates computing the parity of the intermediate qubits.
Let $j > i+1$ to avoid trivial cases.
Denote by $\mathrm{CZ}_{i,j}$ the controlled-$\sigma_z$ gate between $i$ and $j$.
Let $P_{i+1,j}$ be the \emph{parity gate} computing the parity of the qubits between $i$ and $j$ and storing it in qubit $i+1$; see \cref{fig:parity_gate} for an illustration.
This gate is constructed with controlled-$\sigma_x$ gates and has log-depth.
For all $1\leq i < j \leq n$, we have
$$
  \mathrm{FBS}_{ij}(\theta) = P_{i+1,j}\mathrm{CZ}_{i,i+1}\mathrm{RBS}_{ij}(\theta)\mathrm{CZ}_{i,i+1} P_{i+1,j}^*.
$$



%
\section{Useful technical results \label{sec:tech}}

%
The following proposition is a consequence of Wick's theorem.
\begin{proposition}\label{prop:overlap}
  Let $k$ be an even integer and let $\sfX \in \mathbb{R}^{n\times k}$ have linearly independent columns of unit $2$-norm.
  If $k$ is even, we have
  $$
  \braket{0}{\columns(\sfX)} = \pf\left(\skew(\sfX^\top \sfX)\right).
  $$
  In particular, if $\sfX \in \mathbb{R}^{n\times r}$ and $\sfX^\prime \in \mathbb{R}^{n\times r^\prime}$ have linearly independent columns of unit $2$-norm and are such that $\text{parity}(r) = \text{parity}(r^\prime)$, we have
  $$
  \braket{\columns(\sfX^\prime)}{\columns(\sfX)} = \pf \begin{bmatrix}
    \skew(\sfX^\top \sfX) & \sfX^\top \sfX^\prime\\
    -\sfX^{\prime\top} \sfX & \skew(\sfX^{\prime\top} \sfX^\prime).
  \end{bmatrix}
  $$
  Otherwise, if $\text{parity}(r) \neq \text{parity}(r^\prime)$, it holds that $\braket{\columns(\sfX^\prime)}{\columns(\sfX)} =0$.
\end{proposition}
\begin{proof}
  We use Wick's theorem -- see e.g. Theorem 3 in \citep{bardenet2022point}  and references therein -- and the contraction $\bra{\emptyset}\Cliff(\sfx)\Cliff(\sfx^\prime)\ket{\emptyset}= \sfx^\top \sfx^\prime$  as a consequence of \cref{eq:Anticommutation_Cs}.
\end{proof}
As a consequence of \cref{prop:overlap}, if the columns of $\sfX$ are also orthogonal and $r^\prime = r$, we recover that $\braket{\columns(\sfX^\prime)}{\columns(\sfX)} \propto \det(\sfX^\top \sfX^\prime)$ by using e.g.\ \cref{lem:pf_saddle}.
The latter determinant can be expressed in terms of the product of cosines of principal angles between the vector spaces $\text{col}(\sfX)$ and $\text{col}(\sfX^\prime)$; see \citep[page 10]{kerenidis2022quantum}.
\subsection{Expression of $\ket{\columns(\sfX)}$ \label{sec:expression_psi}}
To determine the expression of the coefficient in front of the basis vector $c_{s_\ell}^* \dots c_{s_1}^*\ket{\emptyset}$ in the decomposition \cref{eq:expression_colX} of $\ket{\columns(\sfX)}$, we use the version of Wick's Theorem in \citep[Section 3.5]{bardenet2022point} with respect to the Fock vacuum $\ket{\emptyset}$.
A particular case is the following: consider linear combinations of the $c_i$'s and $c_i^*$'s (defined in \cref{e:CAR}) that we denote $\beta_1, \dots, \beta_m$ where $m$ is even.
Then, we have
\begin{align}
    \bra{\emptyset}\beta_1\dots\beta_m \ket{\emptyset}= \sum_{\sigma \text{ contraction}} \mathrm{sgn}(\sigma)  \bra{\emptyset}\beta_{\sigma(1)}\beta_{\sigma(2)}\ket{\emptyset}\dots  \bra{\emptyset}\beta_{\sigma(m-1)}\beta_{\sigma(m)}\ket{\emptyset}. \label{eq:Wick}
\end{align}
We quickly remind the definition of a contraction. 
For $m$ even, we remind a contraction is a permutation such that $\sigma(1)< \sigma(3) < ... < \sigma(m - 1)$, and $\sigma(2i - 1) < \sigma(2i)$ for $i \leq m/2$.
The key point to notice is that the Pfaffian of a skewsymmetric matrix $A \in \mathbb{C}^{2k\times 2k}$ reads
\begin{align}
  \pf(A) = \sum_{\sigma \text{ contraction}} \mathrm{sgn}(\sigma) A_{\sigma(1)\sigma(2)}\dots A_{\sigma(2k-1) \sigma(2k)}. \label{eq:pf}
\end{align}
At this point, we use the following identities: 
\begin{equation}
  \bra{\emptyset}\Cliff(\sfx)\Cliff(\sfx^\prime)\ket{\emptyset}= \sfx^\top \sfx^\prime \text{ and } \bra{\emptyset}\Cliff(\sfx)c_s^*\ket{\emptyset}= \sfx_s. \label{eq:trivial_contractions}
\end{equation}
Trivially, $\bra{\emptyset}c_s c_{s^\prime}\ket{\emptyset}= 0$.
Now, we fix the ordered subset $\calS = (s_1, \dots, s_\ell)$ and note that $\prod_{i\in \rev(\calS)}c_i^*\ket{\emptyset}= c^*_{s_\ell} \dots c^*_{s_1}\ket{\emptyset}$.
We consider the expression of
\begin{align}
  \bra{\emptyset} c_{s_1} \dots c_{s_\ell}\Cliff(\sfX_{:1}) \dots \Cliff(\sfX_{:r})  \ket{\emptyset},\label{eq:expression_contraction}
\end{align}
in the light of \cref{eq:Wick}.

We define now a skewsymmetric matrix as follows.
For $i < j$, we have
\[
  A_{ij} = \begin{cases}
    \bra{\emptyset}c_{s_{i}} c_{s_{j}}\ket{\emptyset} & \text{ if } 1 \leq i<j \leq \ell\\
    \bra{\emptyset}c_{s_{i}} \Cliff(\sfX_{:(j-\ell)})\ket{\emptyset} & \text{ if } 1 \leq i \leq \ell \text{ and }  1 \leq j - \ell \leq r\\
    \bra{\emptyset}\Cliff(\sfX_{:(i-\ell)}) \Cliff(\sfX_{:(j-\ell)})\ket{\emptyset} & \text{ if }  1 \leq i - \ell < j - \ell \leq r
  \end{cases}.
\]
By using \cref{eq:trivial_contractions} and completing the lower triangular part of $A$ to ensure its skewsymmetry, we find
\begin{align*}
  A = \begin{bmatrix}
    0_{\ell} & \sfX_{\calS : }\\
    - (\sfX^\top)_{:\calS} & \skew(\sfX^\top\sfX)
  \end{bmatrix}.
\end{align*}
And thus,
$
  \bra{\emptyset} c_{s_1} \dots c_{s_\ell}\Cliff(\sfX_{:1}) \dots \Cliff(\sfX_{:r})  \ket{\emptyset}=\pf A
$.
This completes the proof.
\subsection{Pfaffian identities}
\begin{lemma}
  \label{lem:pf_saddle}
  Let $S,X$ be $r\times r$ real matrices with $S^\top  = - S$.
  We have 
  \[
  \pf
  \begin{bmatrix}
    0_r & X\\
    -X^\top & S
  \end{bmatrix}
  =
  (-1)^{r(r-1)/2}
  \det(X).
  \]
\end{lemma}
\begin{proof}
We readily check the identity
$$
\begin{bmatrix}
  0_r & X\\
  -X^\top & S
\end{bmatrix}=   
\begin{bmatrix}
  0_r & -X\\
  I_r  & -S/2
\end{bmatrix}
\begin{bmatrix}
  0_r & I_r\\
  -I_r & 0_r
\end{bmatrix}
\begin{bmatrix}
  0_r & I_r\\
  -X^\top  & S/2
\end{bmatrix}.
$$
Now, we use the well-known formula $\pf (T^\top S T) = \det(T) \pf(S)$ where $S$ is skewsymmetric, which gives
$$
\pf
\begin{bmatrix}
  0_r & X\\
  -X^\top & S
\end{bmatrix}
=
\det
\begin{bmatrix}
  0_r & -X\\
  {I_r}  & -S/2
\end{bmatrix}
\pf
\begin{bmatrix}
  0_r & {I_r}\\
  -{I_r} & 0_r
\end{bmatrix}.
$$
Note that, by computing the signature of the appropriate permutation, we have
$$
\pf
\begin{bmatrix}
  0_r & I_r\\
  -I_r & 0_r
\end{bmatrix}
=
(-1)^{r(r-1)/2}.
$$
The last step is to use 
$$
\det
\begin{bmatrix}
  A & B\\
  C  & D
\end{bmatrix}
=
\det(AD - BC)
$$
where $A,B,C,D$ has the same size and $C,D$ commute. 
Since the identity matrix trivially commutes with $-S/2$, we obtain
$$
\pf
\begin{bmatrix}
  0_r & X\\
  -X^\top & S
\end{bmatrix}
=
(-1)^{r(r-1)/2}
\det
\begin{bmatrix}
  0_r & -X\\
  I_r  & -S/2
\end{bmatrix}
= (-1)^{r(r-1)/2}
\det(X).
$$
\end{proof} 

\subsection{More about the correlation kernel}
Though it is guaranteed by construction, we explicitly prove that the correlation kernel of this DPP has non-negative minors.
\begin{lemma}\label{lem:K_is_valid}
  Let the assumptions of \cref{thm:pmf_DPP_A} hold. Let
\begin{equation*}
  K = X \left(A + X^\top X\right)^{-1}X^\top.
\end{equation*}
The matrix $K$ is such that $ v^\top K v \geq 0$ for all unit norm $v\in \mathbb{R}^n$.
\end{lemma}
\begin{proof}
  We actually show that $v^\top (K+K^\top) v \geq 0$.
  It is sufficient to show that $(A + X^\top X)^{-1}$ is positive semi-definite.
  We rather consider  
  \begin{align*}
      (A + X^\top X)^{-1} = (X^\top X)^{-1/2}(S + {I_r})^{-1}(X^\top X)^{-1/2},
  \end{align*}
  where $S = (X^\top X)^{-1/2} A  (X^\top X)^{-1/2}$ is skewsymmetric.
  Now, adding its transpose to $(S + {I_r})^{-1}$, we have
  \begin{align*}
    (S + {I_r})^{-1} + (-S + {I_r})^{-1} &= (-S + {I_r} + S + {I_r})(S + {I_r})^{-1}(-S + {I_r})^{-1}\\
    &= 2(S + {I_r})^{-1}(-S + {I_r})^{-1}\\
    & =2 (S + {I_r})^{-1}(S + {I_r})^{\top -1}
  \end{align*}
  Hence, by a direct substitution, we find that 
  \begin{align*}
     (K + K^\top)/2 & =  X (X^\top X)^{-1/2} (S + {I_r})^{-1}(S + {I_r})^{\top -1} (X^\top X)^{-1/2} X^\top 
  \end{align*}
  is manifestly positive semidefinite.
\end{proof}
We now explain how \cref{lem:K_is_valid} yields a condition of the principal minors of $K$. 
Let us use the following lemma.
\begin{lemma}[Lemma 1 in \citep{gartrell2019learning}]\label{lem:minors}
  Let $A$ be a square real matrix. If $A + A^\top$ is positive semidefinite, all principal minors of $A$ are non-negative. 
\end{lemma}
\begin{proof}
This follows from  \citep[generating method 4.2]{tsatsomeros2002generating} which shows that if a real square matrix $A$ is such that $A + A^\top$ is strictly positive definite, then all its principal minors are strictly positive. 
The desired result for the positive semidefinite case follows by using a density argument; see supplementary material of \citep{gartrell2019learning} for a complete proof. 
\end{proof}

\printbibliography

\end{document}